\documentclass[11pt,a4paper]{article}

\usepackage[T1]{fontenc}
\usepackage[utf8]{inputenc}
\usepackage{kotex} 

\usepackage[a4paper,margin=20mm]{geometry}
\usepackage[dvipsnames,svgnames,table]{xcolor}

\usepackage{mathtools}
\usepackage{amssymb,amsfonts}
\usepackage{amsthm}
\usepackage{bm}
\usepackage{bbm}
\usepackage{mathrsfs}

\usepackage{graphicx}
\usepackage{subcaption}
\usepackage{booktabs}
\usepackage{tabularx}
\usepackage{longtable}
\usepackage{multirow}
\usepackage{makecell}
\usepackage{diagbox}
\usepackage{adjustbox}

\usepackage{algorithm}
\usepackage{algpseudocode}

\usepackage[shortlabels]{enumitem}
\usepackage{float}

\usepackage{tikz}
\usetikzlibrary{fit}
\usepackage[framemethod=TikZ]{mdframed}

\usepackage{natbib}
\usepackage[affil-it]{authblk}

\usepackage[colorlinks,
            citecolor=Black,
            linkcolor=Red,
            urlcolor=Blue]{hyperref}

\usepackage{comment}

\newif\ifdraft
\draftfalse

\ifdraft
  \usepackage[colorinlistoftodos,textsize=scriptsize]{todonotes}
  \includecomment{note}
\else
  \usepackage[disable]{todonotes}
  \excludecomment{note}
\fi

\usepackage{hyperref}
\usepackage{xr-hyper}

\newcommand{\E}{\mathbb{E}}

\newcommand{\bs}{\boldsymbol}

\newcommand{\be}{\begin{equation}}
\newcommand{\ee}{\end{equation}}
\newcommand{\bea}{\begin{eqnarray*}}
\newcommand{\eea}{\end{eqnarray*}}
\newcommand{\bean}{\begin{eqnarray}}
\newcommand{\eean}{\end{eqnarray}}

\newcommand{\ben}{\begin{enumerate}}
\newcommand{\een}{\end{enumerate}}
\newcommand{\bi}{\begin{itemize}}
\newcommand{\ei}{\end{itemize}}

\newcommand{\bt}{\begin{theorem}}
\newcommand{\et}{\end{theorem}}
\newcommand{\bl}{\begin{lemma}}
\newcommand{\el}{\end{lemma}}
\newcommand{\bd}{\begin{definition}}
\newcommand{\ed}{\end{definition}}
\newcommand{\bp}{\begin{proposition}}
\newcommand{\ep}{\end{proposition}}
\newcommand{\bc}{\begin{corollary}}
\newcommand{\ec}{\end{corollary}}
\newcommand{\brem}{\begin{remark}}
\newcommand{\erem}{\end{remark}}

\newcommand{\bcen}{\begin{center}}
\newcommand{\ecen}{\end{center}}
\newcommand{\bsv}{\begin{semiverbatim}}
\newcommand{\esv}{\end{semiverbatim}}

\newcommand{\bbE}{\mathbb{E}}

\newcommand{\bbP}{\mathbb{P}}
\newcommand{\bbR}{\mathbb{R}}
\newcommand{\bbX}{\mathbb{X}}
\newcommand{\bbY}{\mathbb{Y}}

\newcommand{\calB}{\mathcal{B}}
\newcommand{\calC}{\mathcal{C}}

\newcommand{\calL}{\mathcal{L}}

\newcommand{\calN}{\mathcal{N}}
\newcommand{\calP}{\mathcal{P}}

\newcommand{\calS}{\mathcal{S}}
\newcommand{\calT}{\mathcal{T}}

\newcommand{\bbeta}{\boldsymbol{\beta}}
\newcommand{\bepsilon}{\boldsymbol{\epsilon}}

\newcommand{\bmu}{\boldsymbol{\mu}}

\newcommand{\bPsi}{\boldsymbol{\Psi}}
\newcommand{\btheta}{\boldsymbol{\theta}}
\newcommand{\bOmega}{\boldsymbol{\Omega}}

\newcommand{\bGamma}{\boldsymbol{\Gamma}}

\newcommand{\bSigma}{\boldsymbol{\Sigma}}

\newcommand{\bA}{\boldsymbol{A}}
\newcommand{\bbA}{\boldsymbol{\bar{A}}}
\newcommand{\hbA}{\boldsymbol{\hat{A}}}
\newcommand{\btA}{\boldsymbol{\tilde{A}}}

\newcommand{\bB}{\boldsymbol{B}}
\newcommand{\bC}{\boldsymbol{C}}

\newcommand{\bD}{\boldsymbol{D}}

\newcommand{\bI}{\boldsymbol{I}}
\newcommand{\bJ}{\boldsymbol{J}}
\newcommand{\bK}{\boldsymbol{K}}
\newcommand{\bL}{\boldsymbol{L}}
\newcommand{\bM}{\boldsymbol{M}}
\newcommand{\bU}{\boldsymbol{U}}
\newcommand{\bV}{\boldsymbol{V}}
\newcommand{\bX}{\boldsymbol{X}}
\newcommand{\bY}{\boldsymbol{Y}}
\newcommand{\bZ}{\boldsymbol{Z}}

\newcommand{\bS}{ \boldsymbol{S}}
\newcommand{\bR}{ \boldsymbol{R}}
\newcommand{\bN}{ \boldsymbol{N}}
\newcommand{\bF}{ \boldsymbol{F}}
\newcommand{\bP}{ \boldsymbol{P}}
\newcommand{\bQ}{ \boldsymbol{Q}}
\newcommand{\bE}{ \boldsymbol{E}}
\newcommand{\bG}{ \boldsymbol{G}}
\newcommand{\bH}{ \boldsymbol{H}}

\newcommand{\btmu}{\boldsymbol{\tilde{\mu}}}

\newcommand{\btOmega}{\boldsymbol{\tilde{\Omega}}}

\newcommand{\tf}{\tilde{f}}

\theoremstyle{plain}
\newtheorem{theorem}{Theorem}
\newtheorem{proposition}[theorem]{Proposition}
\newtheorem{lemma}[theorem]{Lemma}
\newtheorem{corollary}[theorem]{Corollary}

\theoremstyle{definition}
\newtheorem{definition}{Definition}
\newtheorem{assumption}{Assumption}

\theoremstyle{remark}
\newtheorem{remark}{Remark}

\newcommand{\tr}{\operatorname{tr}}
\newcommand{\Ktr}{\operatorname{Ktr}}
\newcommand{\vecr}{\operatorname{vec}}

\DeclareMathOperator*{\argmax}{argmax}

\hypersetup{colorlinks=true, linkcolor=blue, citecolor=blue, urlcolor=blue}

\providecommand{\keywords}[1]{\par\noindent\textbf{Keywords: }#1}

\title{Laplace Variational Inference for Bayesian Envelope Models}
\author[1]{Seunghyeon Kim}
\author[2]{Kwangmin Lee}
\author[3]{Yeonhee Park}
\affil[1]{Department of Mathematics and Statistics, Chonnam National University}
\affil[2]{Department of Big Data Convergence, Chonnam National University}
\affil[3]{Department of Statistics, Sungkyunkwan University}

\begin{document}

\maketitle

\abstract{Envelope models provide a sufficient dimension reduction framework for multivariate regression analysis. Bayesian inference for these models has been developed primarily using Markov chain Monte Carlo (MCMC) methods. Specifically, Gibbs sampling and Metropolis–Hastings algorithms suffer from slow mixing and high computational cost. Although automatic differentiation variational inference (ADVI) has been explored for Bayesian envelope models, the resulting gradient-based optimization is often numerically unstable due to severe ill-conditioning of the posterior distribution. To address this issue, we propose a novel reparameterization of the posterior distribution that alleviates the ill-conditioning inherent in conventional variational approaches. Building on this reparameterization, we develop an efficient variational inference procedure. Since the resulting likelihood remains nonconjugate, we approximate the corresponding variational factor using a Laplace approximation within a coordinate-ascent variational inference (CAVI) framework. We establish theoretical results showing that, at each one-step coordinate update, the Laplace approximation error relative to the exact variational inference coordinate update converges to zero. Simulation studies and a real-data analysis demonstrate that the proposed method substantially improves computational efficiency while maintaining estimation accuracy and model-selection performance relative to existing approaches.}


\section{Introduction}\label{intro}

\citet{cook2010envelope} introduced the envelope model as a dimension reduction approach to improving estimation efficiency in multivariate regression. The key idea is to recognize that not all variation in multivariate responses is useful for explaining the effects of predictors. The response envelope model isolates a low-dimensional subspace that contains all predictor-related information and discards the remaining variation that is irrelevant to the regression. By removing such immaterial information, envelope methods can achieve substantial gains in estimation efficiency—often comparable to having many additional observations—while retaining the essential regression signal.

Envelope methods have been extended to various contexts, including quantile regression \citep{ding2020envelope}, matrix- and tensor-variate regression \citep{ding2018matrix, li2017parsimonious}, and multivariate probit regression \citep{lee2024bayesian}. Recent advances in the Bayesian envelope framework include \citet{khare2017bayesian}, \citet{leebayesian}, \citet{shen_park_chakraborty_zhang_2023}, \citet{chakraborty2024comprehensive}, and \citet{lee2024bayesian}. 

Bayesian envelope models provide a principled framework for uncertainty quantification and model selection. However, existing Bayesian implementations are often computationally inefficient. In particular, currently available methods rely on Markov chain Monte Carlo (MCMC) algorithms which tend to exhibit slow mixing and high autocorrelation, thereby substantially limiting scalability. 

\citet{khare2017bayesian} proposed a random-scan Gibbs sampler on the Stiefel manifold, but sampling directly on the manifold is computationally expensive. To improve computational efficiency, \citet{chakraborty2024comprehensive} developed a comprehensive Bayesian framework based on a Euclidean parameterization \citep{cook2016note} and employed a Metropolis-within-Gibbs sampler to draw from the posterior (see Section~\ref{sec2} for details). However, Metropolis–Hastings (MH) and Gibbs updates tend to suffer from poor mixing and high autocorrelation in higher dimensions, leading to substantial computational costs. 

In this setting, variational inference (VI) can be considered an alternative approach to alleviate the computational burden of existing MCMC-based methods. VI reformulates posterior inference as an optimization problem, which can substantially reduce computation time \citep{jordan1999introduction, wainwright2008graphical, blei2017variational}. Owing to these advantages, VI has been increasingly adopted as an alternative to MCMC across a wide range of statistical models, including logistic regression models \citep{10.1214/19-STS712}, hierarchical models \citep{10.1214/21-BA1266}, and multinomial probit models \citep{loaiza2023fast}.

Two common strategies for optimizing the VI objective are coordinate-ascent variational inference (CAVI) \citep{10.5555/1162264} and stochastic variational inference (SVI) \citep{10.5555/2567709.2502622}; see \citet{blei2017variational} for a review. These approaches are particularly straightforward in conditionally conjugate models. By conditionally conjugate models, we mean models in which each full conditional belongs to the same family as its prior, yielding closed-form coordinate updates. However, Bayesian envelope models are conditionally nonconjugate—some full conditionals do not admit conjugate forms—and therefore closed-form updates are not available.

Although automatic differentiation variational inference (ADVI) has been explored for Bayesian envelope models \citep{chakrabortysupplement}, our empirical investigation suggests that its performance can be highly sensitive under the standard Euclidean parameterization. In particular, we observed substantial variability in estimation accuracy across replicated datasets, and in higher envelope dimensions the algorithm occasionally exhibited numerical instability. These phenomena may be attributable to the interaction between the envelope geometry and the curvature of the variational objective under the standard parameterization. Such sensitivity motivates the development of a more stable coordinate-wise approximation tailored to this class of models.

To address this issue, we employ a coordinate–ascent Laplace variational inference (CALVI) \citep{wang2013variational}, which retains a coordinate-ascent scheme by applying a Laplace approximation to otherwise intractable updates. We propose a novel reparameterization strategy designed to enable the Laplace approximation, since applying CALVI to the parameterization used in \citet{khare2017bayesian} and \citet{chakraborty2024comprehensive} is computationally costly. This reparameterization effectively shifts the computational bottleneck associated with first- and second-order differentiation to a set of more manageable parameters, thereby substantially enhancing computational efficiency and rendering CALVI practically feasible for envelope models.

We also investigate the theoretical properties of the proposed CALVI procedure. Specifically, we study the approximation error induced by replacing the exact coordinate-wise variational update for a nonconjugate block with its Laplace approximation. While CALVI is widely used in practice, the accuracy of this Laplace-based coordinate update relative to the exact CAVI update has not been thoroughly analyzed. We identify conditions in a general setting under which the resulting approximation error, measured in total variation (TV) distance after appropriate rescaling, becomes asymptotically negligible. Under additional localization and curvature assumption for the Bayesian response envelope model, we establish a theoretical justification for the one-step CALVI algorithm.


The remainder of this paper is organized as follows. Section~\ref{sec2} reviews the Bayesian response envelope model. Section~\ref{sec3} provides a brief review of VI and CALVI. In Section~\ref{sec4}, we address the computational challenges that arise when applying CALVI to envelope models and propose a reparameterization that facilitates the tractable application of CALVI to the Bayesian response envelope model. Section~\ref{sec5} presents theoretical results regarding the approximation error of the CALVI approach. Section~\ref{sec6} presents numerical studies through simulations and a real-world data analysis. Finally, Section~\ref{conclude} concludes the paper with a summary of the main contributions and findings. The Supplementary Material provides detailed proofs and extensions to the Bayesian predictor envelope model.

\section{Review of Bayesian response envelope model}
\label{sec2}

This section reviews the response envelope model introduced by \citet{cook2010envelope}, with particular emphasis on Bayesian inference and its computational challenges. We first introduce the response envelope model, highlighting how it improves estimation efficiency in multivariate linear regression model. We then review existing Bayesian inference methods, identifying the key computational bottlenecks associated with sampling from constrained parameter spaces. The discussion in this section lays the groundwork for the methodological developments in subsequent sections.

Consider the multivariate linear regression model with an $r$-dimensional response $\boldsymbol{Y}$ and a $p$-dimensional predictor $\boldsymbol{X}$ given by
\begin{align*}
    \bY = \bmu + \bbeta\bX + \bepsilon, \quad \bepsilon \sim \calN_r(\boldsymbol{0}, \bSigma), \label{eq:1} \tag{1}
\end{align*}
where $\bmu \in \bbR^{r}$, $\bbeta \in \bbR^{r \times p}$, and $\bSigma \in \calC_{r}$ denote the intercept, coefficient matrix, and covariance, respectively, with $\calC_{d}$ denoting the set of $d \times d$ symmetric positive definite matrices. 

The response envelope model \citep{cook2010envelope} improves the efficiency of estimating $\bbeta$ by restricting attention to a low-dimensional subspace of the response space. 
Formally, let $\calS \subset \bbR^{r}$ be a subspace and denote by $\bbP_{\calS}$ the orthogonal projector onto $\calS$, with $\mathbb{Q}_{\calS} = \bI_{r} - \bbP_{\calS}$. 
We refer to $\bbP_{\calS}\bY$ as the \textit{material} part of the response, which contains all information relevant to the regression, and to $\mathbb{Q}_{\calS}\bY$ as the \textit{immaterial} part, which represents variation unrelated to the predictors. The response envelope model assumes:
\begin{enumerate}
    \item[(i)] $\mathbb{Q}_{\calS} \bY \mid \bX \stackrel{d}{=} \mathbb{Q}_{\calS} \bY$ : the immaterial part is conditionally independent of $\bX$;
    \item[(ii)] $\operatorname{cov}(\mathbb{Q}_{\calS} \bY, \bbP_{\calS} \bY \mid \bX) = 0$ : the material and immaterial parts are conditionally uncorrelated given $\bX$.
\end{enumerate}

Under conditions (i)--(ii), let $u = \dim (\calS)$ and let $\bGamma \in \bbR^{r \times u}$ and $\bGamma_0 \in \bbR^{r \times (r-u)}$ be orthonormal bases of $\calS$ and $\calS^{\perp}$, respectively. Then the response envelope model admits the following parametric formulation:
\begin{align*}
\bY = \bmu + \bGamma\boldsymbol{\eta}\bX + \bepsilon, \qquad \bepsilon \sim \calN_r(\boldsymbol{0}, \bGamma\bOmega\bGamma^\top + \bGamma_0\bOmega_0\bGamma^\top_0), \label{eq:2} \tag{2}
\end{align*}
where $\boldsymbol{\eta} \in \bbR^{u \times p}$ satisfies $\bbeta = \bGamma \boldsymbol{\eta}$, and $\bOmega \in \calC_u$, $\bOmega_0 \in \calC_{r-u}$. In this representation, only the material part $\bGamma^\top \bY$ is informative about $\bbeta$, thereby yielding efficiency gains relative to \eqref{eq:1}, whereas $\bGamma_0^\top \bY$ contributes immaterial variation.

Bayesian inference for the response envelope model~\eqref{eq:2} is challenging because the envelope subspace $\mathcal{S} = \text{span}(\bGamma)$, which is the primary inferential target, lies on the $r \times u$ Grassmann manifold. Although the model is parameterized through an orthonormal basis matrix $\bGamma$ on the $r \times u$ Stiefel manifold, the basis is not identifiable due to its invariance under orthogonal transformations. 

To address the identifiability issue, \citet{khare2017bayesian} imposed restrictions on the model parameters in~\eqref{eq:2} by constraining $\bOmega = \operatorname{diag}(\omega_1, \dots, \omega_u)$ and $\bOmega_0 = \operatorname{diag}(\omega_{0,1}, \dots, \omega_{0,(r-u)})$, with $\omega_1 > \cdots > \omega_u$ and $\omega_{0,1} > \cdots > \omega_{0, r-u}$, which ensures that the envelope model~\eqref{eq:2} is identifiable.

However, their sampler requires rejection sampling from truncated inverse-gamma distributions (for the diagonal elements of $\bOmega$) and sampling from the generalized Bingham distribution (for $\bGamma$). Both steps are computationally intensive and tend to exhibit slow mixing in high-dimensional settings.

To overcome the computational limitations associated with manifold-constrained inference, \citet{chakraborty2024comprehensive} employed \citet{cook2016note}'s parameterization, which identifies the envelope subspace via a parameter defined in Euclidean space. When the first $u$ rows of $\bGamma$ are full-rank, the envelope basis $\bGamma$ can be expressed in terms of an unconstrained matrix $\bA \in \bbR^{(r-u) \times u}$ as follows: 
\begin{flalign*}
    \bGamma = \begin{pmatrix}
        \bGamma_1 \\
        \bGamma_2
    \end{pmatrix} = \begin{pmatrix}
        \boldsymbol{I}_u \\
        \bGamma_2 \bGamma_1^{-1}
    \end{pmatrix} \bGamma_1 \equiv \begin{pmatrix}
        \boldsymbol{I}_u \\
        \bA
    \end{pmatrix} \bGamma_1 \equiv 
    \boldsymbol{C}_{\bA} \bGamma_1,
\end{flalign*}
where $\bGamma_1 \in \bbR^{u \times u}$, $\bGamma_2 \in \bbR^{(r-u) \times u}$, $\bA \equiv \bGamma_2 \bGamma_1^{-1}$, and $\boldsymbol{C}_{\bA} \equiv \begin{pmatrix} \boldsymbol{I}_u \\ \bA \end{pmatrix} \in \bbR^{r \times u}$. Then $\bGamma$ and $\bGamma_0$ can be expressed as:
\begin{align*}
    \bGamma(\bA) := \boldsymbol{C}_{\bA}(\boldsymbol{C}_{\bA}^\top \boldsymbol{C}_{\bA})^{-1/2}, \qquad
    \bGamma_0(\bA) := \boldsymbol{D}_{\bA}(\boldsymbol{D}_{\bA}^\top \boldsymbol{D}_{\bA})^{-1/2}, \tag{3} \label{eq:3}
\end{align*}
where $\boldsymbol{D}_{\bA} \equiv \begin{pmatrix} -\bA^\top \\ \boldsymbol{I}_{r-u} \end{pmatrix} \in \bbR^{r \times (r-u)}$. Without loss of generality, we may assume—after permuting rows if necessary—that the first $u$ rows of $\bGamma$ form a full-rank block.

Using \citet{cook2016note}'s parameterization, \citet{chakraborty2024comprehensive} specified the response envelope model as follows:
\begin{align*}
    \boldsymbol{Y} &= \bmu + \bGamma(\bA)\boldsymbol{\eta}\boldsymbol{X} + \bepsilon, \qquad \bepsilon \sim \mathcal{N}_r\left(\boldsymbol{0}, \bGamma(\bA)\boldsymbol{\Omega}\bGamma^\top(\bA) + \bGamma_0(\bA)\boldsymbol{\Omega}_0\bGamma^\top_0(\bA)\right). \label{eq:4} \tag{4} 
\end{align*}
They specified the prior distributions as follows:
\begin{enumerate}
    \item[(i)] $\pi(\bmu) \propto 1$ is an improper flat prior on $\bbR^{r}$.
    \item[(ii)] $\bOmega \sim \mathcal{IW}_u(\bPsi, \nu^{(1)})$, where $\bPsi \in \mathcal{C}_u$, $\nu^{(1)} > u-1$.
    \item[(iii)] $\bOmega_0 \sim \mathcal{IW}_{r-u}(\bPsi_0, \nu^{(0)})$, where $\bPsi_0 \in \mathcal{C}_{r-u}$, $\nu^{(0)} > r-u-1$.
    \item[(iv)] $\boldsymbol{\eta} \mid \bA, \bOmega \sim \mathcal{MN}_{u,p}(\bGamma^\top(\bA)\boldsymbol{B}_0, \bOmega, \boldsymbol{M}^{-1})$, where $\bB_0 \in \bbR^{r \times p}$ and $\boldsymbol{M} \in \mathcal{C}_{p}$.
    \item[(v)] $\bA \sim \mathcal{MN}_{r-u,u}(\bA_0, \boldsymbol{U}_0^{(A)}, \boldsymbol{V}_0^{(A)})$, where $\bA_0 \in \bbR^{(r-u) \times u}$,  $\boldsymbol{U}_0^{(A)} \in \mathcal{C}_{r-u}$, and  $\boldsymbol{V}_0^{(A)} \in \mathcal{C}_{u}$,
\end{enumerate}
where $\mathcal{MN}$ denotes the matrix-normal distribution and $\mathcal{IW}$ denotes the inverse-Wishart distribution. Let $\bbY \in \mathbb{R}^{n \times r}$ denote the response matrix, whose $i$th row is
$\bY_i^\top \in \mathbb{R}^r$. Similarly, let $\bbX \in \mathbb{R}^{n \times p}$ denote the predictor matrix, whose $i$th row is $\bX_i^\top \in \mathbb{R}^p$. Let $\bbY_{\bmu} = \bbY - \boldsymbol{1}_n \bmu^\top$, where $\boldsymbol{1}_n$ is the $n$-vector of ones. Then the unnormalized log posterior distribution is:
\begin{flalign*}
    & \log p(\bmu, \boldsymbol{\eta}, \mathbf{\Omega}, \mathbf{\Omega}_0,\bA \mid \mathbb{Y}) = \notag \\ 
    & \ \textit{const.} -\frac{n + p + \nu^{(1)} + u + 1}{2} \log \left| \mathbf{\Omega} \right| -\frac{n + \nu^{(0)} + r-u + 1}{2} \log \left| \mathbf{\Omega}_0 \right| \\
    & \ -\frac{1}{2} \tr \left[(\mathbb{Y}_{\bmu}\mathbf{\Gamma}(\bA) -\mathbb{X}\boldsymbol{\eta}^\top ) \mathbf{\Omega}^{-1} (\mathbb{Y}_{\bmu}\mathbf{\Gamma}(\bA) -\mathbb{X}\boldsymbol{\eta}^\top )^\top \right] \label{eq:5} \tag{5}\\
    & \ -\frac{1}{2} \tr \left[ \mathbf{\Gamma}_0^\top (\bA) \mathbb{Y}_{\bmu}^\top \mathbb{Y}_{\bmu} \mathbf{\Gamma}_0(\bA) \mathbf{\Omega}_0^{-1} \right] \label{eq:6} \tag{6} \\
    & \ -\frac{1}{2} \tr\left[ \boldsymbol{M}(\boldsymbol{\eta}-\mathbf{\Gamma}^\top (\bA) \mathbf{B}_0)^\top \mathbf{\Omega}^{-1} (\boldsymbol{\eta}-\mathbf{\Gamma}^\top (\bA) \mathbf{B}_0) \right] -\frac{1}{2} \tr\left[ \boldsymbol{\Psi}\mathbf{\Omega}^{-1} \right]  -\frac{1}{2} \tr\left[\mathbf{\Psi}_0 \mathbf{\Omega}_0^{-1} \right] \\
     & \ - \frac{1}{2} \tr\left[\boldsymbol{V}_0^{(A) -1} (\bA - \bA_0)^\top \boldsymbol{U}_0^{(A) -1} (\bA - \bA_0) \right]. \label{eq:7} \tag{7} &&
\end{flalign*}

The primary object of interest in the response envelope model is the envelope subspace $\mathcal{S}=\mathrm{span}(\bGamma)$ (and the induced regression coefficient $\bbeta=\bGamma\boldsymbol{\eta}$). Under the Euclidean parameterization~\eqref{eq:3}, the subspace $\mathcal{S}$ is fully characterized by the unconstrained matrix $\bA$, so posterior inference on the envelope structure can be carried out through inference on $\bA$. The remaining parameters $(\bmu,\boldsymbol{\eta},\bOmega,\bOmega_0)$ play auxiliary roles in the likelihood specification and are typically treated as nuisance parameters conditional on $\bA$. Importantly, given $\bA$, these blocks admit standard conjugate full conditional distributions (normal/matrix-normal and inverse-Wishart forms under the priors listed above).

\citet{chakraborty2024comprehensive} proposed a blockwise Metropolis-within-Gibbs sampler. Specifically, they proposed an MH algorithm with a likelihood-driven proposal to update $\bA$ to alleviate the poor mixing of a random-walk proposal (see details in \citet{chakrabortysupplement} Section C.7).

However, the resulting algorithms of \citet{khare2017bayesian} and \citet{chakraborty2024comprehensive} still suffer from poor mixing and slow convergence, especially when the parameter dimension is large. Consistent with this, \citet{chakraborty2024comprehensive} report average effective sample size (ESS) ratios (their Metropolis-within-Gibbs sampler relative to the manifold-based Gibbs sampler of \citet{khare2017bayesian}) of about $0.9$ across $\bbeta$-coordinates and replications; see \citet{chakrabortysupplement}, Section~C.13.3. Moreover, the likelihood-driven proposal for $\bA$ is constructed solely from the likelihood and does not incorporate prior information, potentially reducing posterior efficiency in the presence of informative priors. Altogether, these limitations underscore the fundamental scalability challenges of MCMC-based inference for the response envelope model and motivate the development of computationally efficient alternatives.

\section{Review of Laplace variational inference}
\label{sec3}

VI is a scalable alternative to MCMC that approximates the
posterior distribution by a more tractable family. We refer to this approximating distribution as the variational distribution, denoted by $q(\btheta)$. Under the mean-field variational family (MFVF) assumption, the variational distribution factorizes across parameter blocks as
\[
q(\btheta)=\prod_{i=1}^K q_i(\btheta_i),
\]
where each $q_i(\btheta_i)$ is referred to as a variational factor. VI minimizes the Kullback--Leibler (KL) divergence
$\mathrm{KL}\!\left(q(\btheta)\,\|\,p(\btheta\mid\bbY)\right)$, which is equivalent to maximizing the evidence lower bound (ELBO) \citep{jordan1999introduction}
\[
\mathcal{L}(q)
:= \E_q[\log p(\bbY,\btheta)]-\E_q[\log q(\btheta)].
\]

\paragraph{Coordinate-ascent variational inference.}
CAVI optimizes the ELBO by iteratively updating one variational factor at a time while holding the remaining factors fixed. Let $q^{(t)}(\btheta)=\prod_i q_i^{(t)}(\btheta_i)$ denote the variational distribution
at iteration $t$. Given the current iterate $q_{-i}^{(t)}(\btheta_{-i}):=\prod_{j\neq i} q_j^{(t)}(\btheta_j)$,
the $(t+1)$th update for the $i$th factor is obtained by solving
\[
q_i^{(t+1)}(\btheta_i)
\;\propto\;
\exp\!\Big\{
\E_{q_{-i}^{(t)}}\big[\log p(\btheta_i\mid\btheta_{-i},\bbY)\big]
\Big\}.
\tag{8}\label{eq:8}
\]
More precisely, the update is carried out sequentially: at iteration $t$, the update for $q_i$ uses $q_j^{(t+1)}$ for $j<i$ and $q_j^{(t)}$ for $j>i$. For notational simplicity, we write $q_{-i}^{(t)}$ to denote this collection. CAVI proceeds by cycling over $i=1,\dots,K$ and repeating these updates until convergence.

\paragraph{Conditionally conjugate blocks.}
A parameter block $\btheta_i$ in $q_i(\btheta_i)$ is said to be conditionally conjugate in the variational sense if the update in~\eqref{eq:8} yields a factor in a known parametric family \citep{wang2013variational}. Formally, $\btheta_i$ is conditionally conjugate if there exists a family
$\mathcal F_i$ such that
\[
q_i^{(t+1)}(\btheta_i)\in\mathcal F_i
\quad\text{and}\quad
p(\btheta_i\mid\btheta_{-i},\bbY)\in\mathcal F_i.
\]
For example, when $\mathcal F_i$ is the collection of normal distributions,
\[
\E_{q_{-i}^{(t)}}[\log p(\btheta_i\mid\btheta_{-i},\bbY)]
= -\tfrac12\btheta_i^\top \mathbf H^{(t)}\btheta_i
+ \mathbf b^{(t)\top}\btheta_i + \text{const},
\]
and the update is
\[
q_i^{(t+1)}(\btheta_i)
= \calN\!\big(\mathbf H^{(t)-1}\mathbf b^{(t)},\,\mathbf H^{(t)-1}\big),
\]
where $\mathbf{H}^{(t)}$ and $\mathbf{b}^{(t)}$ are the quadratic and linear coefficients in $\btheta_i$ determined by the expansion of $\E_{q_{-i}^{(t)}}[\log p(\btheta_i\mid\btheta_{-i},\bbY)]$.

Conditionally conjugate blocks play a crucial computational role in CAVI. Because the optimal variational factor $q_i^{(t+1)}(\btheta_i)$ belongs to a known parametric family, its normalizing constant and moments are available in closed form. As a result, expectations with respect to $q_i(\btheta_i)$ that appear in the coordinate updates of other blocks can be evaluated analytically. This tractability avoids numerical integration or gradient-based optimization for these blocks, leading to substantial computational savings and improved numerical stability.

\paragraph{Nonconjugate blocks and Laplace updates.}
When $\btheta_i$ is conditionally nonconjugate, the update~\eqref{eq:8} is not tractable. To address this, \citet{wang2013variational} proposed incorporating a Laplace approximation within each coordinate update. Define the local objective at iteration $t$ as
\[
f_i^{(t)}(\btheta_i)
:= \E_{q_{-i}^{(t)}}\big[\log p(\btheta_i\mid\btheta_{-i},\bbY)\big].
\tag{9}\label{eq:9}
\]
Let
\[
\hat{\btheta}_i^{(t)}
:= \arg\max_{\btheta_i} f_i^{(t)}(\btheta_i)
\]
be the local maximizer. We approximate $f_i^{(t)}$ by a second-order Taylor expansion around $\hat{\btheta}_i^{(t)}$, yielding the following Gaussian update
\begin{align*}
&q_i^{(t+1)}(\btheta_i) \\
&\approx \exp \Big\{ f_i^{(t)}(\hat{\btheta}^{(t)}_i) + \frac{1}{2} (\btheta_i - \hat{\btheta}^{(t)}_i)^\top f_i^{(t)\, \prime \prime}(\hat{\btheta}^{(t)}_i) (\btheta_i - \hat{\btheta}^{(t)}_i) \Big\} \\
&= \calN\!\left(
\hat{\btheta}_i^{(t)},
-\big[f_i^{(t)\prime\prime}(\hat{\btheta}_i^{(t)})\big]^{-1}
\right) =: q_i^{L(t+1)}(\btheta_i).
\tag{10}\label{eq:10}
\end{align*}
We refer to CAVI augmented with this Laplace-based coordinate update as \emph{coordinate-ascent Laplace variational inference} (CALVI).

\section{CALVI for the Bayesian response envelope model}
\label{sec4}

In this section, we show how CALVI can be effectively applied to the Bayesian response envelope model. A key challenge arises from the envelope parameter matrix $\bA$, which is conditionally nonconjugate because it enters the likelihood through the orthonormal basis matrices $\bGamma(\bA)$ and $\bGamma_0(\bA)$ in~\eqref{eq:5} and~\eqref{eq:6}. As shown in Section~\ref{sec4.1}, a direct Laplace update for $\bA$ under the Euclidean parameterization of \citet{cook2016note} leads to severe computational bottlenecks, driven by computationally expensive matrix derivative calculations. To address this issue, Section~\ref{sec4.2} introduces a novel reparameterization that removes the dominant computational bottlenecks and yields a tractable and numerically stable CALVI procedure for the Bayesian response envelope model.

\subsection{\texorpdfstring
{Computational issues of parameterization~\eqref{eq:3}}
{Computational issues of parameterization (3)}}
\label{sec4.1}

In the envelope models parameterized via~\eqref{eq:3}, applying a Laplace approximation to the nonconjugate update for $\bA$ is computationally challenging. At each CALVI iteration, the Laplace step requires differentiating the local objective
\begin{align*}
    f^{(t)}(\bA)
:= \E_{q_{-\bA}^{(t)}}\!\big[\log p(\bA\mid\bmu, \boldsymbol{\eta}, \bOmega, \bOmega_0,\bbY)\big], \tag{11} \label{eq:11}
\end{align*} 
where $q_{-\bA}^{(t)} = q^{(t)}(\bmu)q^{(t)}(\bs{\eta})q^{(t)}(\bOmega)q^{(t)}(\bOmega_0)$. The explicit form of $f^{(t)}(\bA)$ and its full derivative are given in the Supplementary Material, Section~\ref{sec:A}. The resulting computational difficulty is driven by a single dominant term.

Specifically, the bottleneck arises from differentiating the matrix inverse square-root terms appearing in $\bGamma(\bA)$ and $\bGamma_0(\bA)$, which induce large Kronecker-sum matrix inversions.
Let
\begin{align*}
    \bJ(\bA) &:= \bC_{\bA}^\top \bC_{\bA} = \bI_u + \bA^\top \bA \in \bbR^{u \times u}, \\
    \bJ_0(\bA) &:= \bD_{\bA}^\top \bD_{\bA} = \bI_{r-u} + \bA \bA^\top \in \bbR^{(r-u) \times (r-u)}.
\end{align*}
Differentiation of $\bJ(\bA)^{-1/2}$ and $\bJ_0(\bA)^{-1/2}$ with respect to $\bA$ leads to the Kronecker-sum expressions in~\eqref{eq:12}.
\begin{align*}
    \frac{\partial \, \vecr \left( \bJ(\bA)^{-1/2} \right)}{\partial \, \vecr(\boldsymbol{C}_{\bA})^\top} =& - \Big[ \underbrace{\bJ(\bA)^{-1/2} \otimes \boldsymbol{I}_{u} + \boldsymbol{I}_{u} \otimes \bJ(\bA)^{-1/2} }_{(i)} \Big]^{-1} \\
    &\quad \times \Big[ \bJ(\bA)^{-1} \otimes \bJ(\bA)^{-1} \boldsymbol{C}_{\bA}^\top +   \bJ(\bA)^{-1} \boldsymbol{C}_{\bA}^\top \otimes \bJ(\bA)^{-1} \mathcal{K}_{r,u} \Big], \\
    \frac{\partial \, \vecr \left( \bJ_0(\bA)^{-1/2} \right)}{\partial \, \vecr(\boldsymbol{D}_{\bA})^\top} =& - \Big[ \underbrace{\bJ_0(\bA)^{-1/2} \otimes \boldsymbol{I}_{r-u} + \boldsymbol{I}_{r-u} \otimes \bJ_0(\bA)^{-1/2}}_{(ii)} \Big]^{-1} \\
    &\quad \times \Big[ \bJ_0(\bA)^{-1} \otimes \bJ_0(\bA)^{-1} \boldsymbol{D}_{\bA}^\top + \bJ_0(\bA)^{-1} \boldsymbol{D}_{\bA}^\top \otimes \bJ_0(\bA)^{-1} \mathcal{K}_{r,(r-u)} \Big], \label{eq:12} \tag{12}
\end{align*}
where $\vecr(\cdot)$ denotes the columnwise vectorization operator, and $\mathcal K_{m,n}$ denotes the $(mn)\times(mn)$ commutation matrix. The matrices appearing in~\eqref{eq:12}(i) and~\eqref{eq:12}(ii) have dimensions $u^2\times u^2$ and $(r-u)^2\times(r-u)^2$, respectively. Consequently, evaluating these terms requires $\mathcal{O}\!\left(\max\{u,\,r-u\}^{\,4.752}\right)$ operations in theory when fast matrix inversion algorithms, such as Coppersmith--Winograd \citep{10.1145/28395.28396}, are employed. Because this computation is performed at every iteration of the Laplace update, a direct implementation of CALVI under the parameterization~\eqref{eq:3} is computationally burdensome.

Importantly, this difficulty is not specific to the response envelope model. Any envelope model based on the parameterization of \citet{cook2016note} induces the same Kronecker-sum bottleneck~\eqref{eq:12}; see, for example, the predictor envelope model in Supplementary Material, Section~\ref{sec:C}.

\subsection{Reparameterization and CALVI algorithm}
\label{sec4.2}

Section~\ref{sec4.1} identified the dominant cost in a direct Laplace update for $\bA$ under the Euclidean parameterization~\eqref{eq:3}: differentiating the inverse square-root factors $\bJ(\bA)^{-1/2}$ and $\bJ_0(\bA)^{-1/2}$ induces the Kronecker-sum inversions in~\eqref{eq:12}. Under the reparameterization introduced below, these inverse square-root terms no longer appear in $f^{(t)}(\bA)$.

\paragraph{Reparameterization.}
We introduce the mapping $(\bmu, \boldsymbol{\eta}, \bOmega, \bOmega_0, \bA) \mapsto (\btmu, \boldsymbol{\tilde{\eta}}, \btOmega, \btOmega_0, \bA)$ defined as
\begin{align*}
    \boldsymbol{\tilde{\eta}} &= \bJ(\bA)^{1/2}\boldsymbol{\eta}, \\
    \btOmega &= \bJ(\bA)^{1/2}\boldsymbol{\Omega}\bJ(\bA)^{1/2}, \\
    \btOmega_0 &= \bJ_0(\bA)^{1/2}\boldsymbol{\Omega}_0\bJ_0(\bA)^{1/2}.
    \tag{13}\label{eq:13}
\end{align*}
We also define the centered intercept
\begin{align*}
\btmu
&:= \bmu + \bC_{\bA}\,\bJ(\bA)^{-1}\,\boldsymbol{\tilde{\eta}}\,\bar{\boldsymbol{X}}, \\
\bar{\boldsymbol{X}}
&:= \bbX^\top \boldsymbol{1}_n/n\in\bbR^p,
\qquad
\bbX_c := \bbX - \boldsymbol{1}_n\bar{\boldsymbol{X}}^\top.
\tag{14}\label{eq:14}
\end{align*}
This centering ensures that the likelihood depends on $\bbX$ only through the centered design
$\bbX_c$ while preserving conjugacy in the $(\btmu,\boldsymbol{\tilde\eta},\btOmega,\btOmega_0)$
blocks.

The key motivation behind~\eqref{eq:13} is that it cancels the inverse square-root factors
appearing in $\bGamma(\bA)$ and $\bGamma_0(\bA)$ within the quadratic forms of the likelihood. For example, using $\bGamma(\bA)=\bC_{\bA}\bJ(\bA)^{-1/2}$ and $\bOmega^{-1}=\bJ(\bA)^{1/2}\btOmega^{-1}\bJ(\bA)^{1/2}$, the material part quadratic form in~\eqref{eq:5} simplifies to an expression involving only $\bC_{\bA}$ and $\btOmega^{-1}$, with no $\bJ(\bA)^{-1/2}$ remaining. An analogous cancellation holds for the immaterial part via $\bD_{\bA}$ and $\btOmega_0^{-1}$.

\paragraph{Key computational implication: the Laplace derivatives for $\bA$.}
Under the MFVF scheme 
\begin{align*}
q^{(t)}(\btmu,\boldsymbol{\tilde\eta},\btOmega,\btOmega_0,\vecr(\bA)) =
q^{(t)}(\btmu)\,q^{(t)}(\boldsymbol{\tilde\eta})\,q^{(t)}(\btOmega)\,q^{(t)}(\btOmega_0)\,q^{(t)}(\vecr(\bA)), \tag{15} \label{eq:15}
\end{align*} 
define the $\bA$-coordinate objective at iteration $t$ by
\begin{align*}
\tf^{(t)}(\bA)
:= \E_{q^{(t)}_{-\bA}}\!\left[\log p\!\left(\bA\mid \btmu,\boldsymbol{\tilde\eta},\btOmega,\btOmega_0,\bbY\right)\right],
\tag{16}\label{eq:16}
\end{align*}
where $q^{(t)}_{-\bA}:=q^{(t)}(\btmu)\,q^{(t)}(\boldsymbol{\tilde\eta})\,q^{(t)}(\btOmega)\,q^{(t)}(\btOmega_0)$. Equivalently, $\tf^{(t)}(\bA)$ can be expressed using the expected log joint, up to an additive constant independent of $\bA$. The explicit moment-expanded form of $\tf^{(t)}(\bA)$ is provided in the Supplementary Material, Section~\ref{sec:B.1}.

Throughout this section, we adopt isotropic inverse-Wishart prior scales,
setting $\bPsi=\psi^{(1)}\bI_u$ and $\bPsi_0=\psi^{(0)}\bI_{r-u}$ for constants
$\psi^{(1)},\psi^{(0)}>0$. Under this choice, the $\bA$-coordinate objective
$\tf^{(t)}(\bA)$ depends on $\bA$ only through the log-determinant term
$\log|\bJ_0(\bA)| := \log|\bI_{r-u}+\bA\bA^\top|$ and linear/quadratic forms
induced by $\bC_{\bA}$ and $\bD_{\bA}$, yielding a simple and efficient Laplace
update. More general prior scales can be handled via a whitening reparameterization
based on the Cholesky decomposition of $\bPsi$ and $\bPsi_0$.

As a result, the gradient of $\tf^{(t)}(\bA)$ with respect to $\bA$ admits an explicit
matrix-valued expression, while the second-order derivative can be characterized as
a Hessian with respect to the vectorized parameter $\vecr(\bA)$.
Specifically, the gradient and Hessian take the following generic forms:
\begin{align*}
&\nabla_{\bA}\tf^{(t)}(\bA)
=
\kappa\, \bJ_0(\bA)^{-1}\bA
+ 2\, \boldsymbol{C}_1^{(t)} \bA \boldsymbol{C}_2^{(t)}
+ \boldsymbol{C}_3^{(t)\top}, \tag{17}\label{eq:17} \\
&\nabla^2_{\vecr(\bA)}\tf^{(t)}(\bA)\\
&\quad=
\kappa\,(\bI_{u} - \bA^\top \bJ_0(\bA)^{-1} \bA)\otimes \bJ_0(\bA)^{-1} \\
&\qquad
- \kappa\,\big((\bJ_0(\bA)^{-1}\bA)^\top\otimes
      (\bJ_0(\bA)^{-1}\bA)\big)\mathcal{K}_{(r-u),u} + 2\,\boldsymbol{C}_2^{(t)} \otimes \boldsymbol{C}_1^{(t)}, \tag{18}\label{eq:18}
\end{align*}
where $\kappa=2n+\nu^{(1)}+\nu^{(0)}$,
$\boldsymbol{C}_1^{(t)}\in\bbR^{(r-u)\times(r-u)}$ and
$\boldsymbol{C}_2^{(t)}\in\bbR^{u\times u}$ are symmetric matrices, and
$\boldsymbol{C}_3^{(t)}\in\bbR^{u\times(r-u)}$ depends only on the variational moments of
the $q^{(t)}(\btmu, \boldsymbol{\tilde{\eta}}, \btOmega, \btOmega_0)$ at iteration $t$ (and on the data and hyperparameters).
Here, the coefficient $\kappa$ collects contributions from the log-determinant term
arising in the likelihood and from the inverse-Wishart priors under the
reparameterization; see Supplementary Material, Section~\ref{sec:B.1}, for the full derivation.

As a consequence of the reparameterization, the $\bA$-coordinate objective~\eqref{eq:16} admits explicit closed-form expressions for both the gradient and Hessian; see~\eqref{eq:17}--\eqref{eq:18}. Evaluating these derivatives requires only forming $\bJ_0(\bA)^{-1}$ (an $(r-u)\times(r-u)$ inverse) and standard matrix multiplications, and is therefore computationally tractable in the dimensions considered here. In contrast, under the Euclidean parameterization of~\citet{cook2016note}, the Laplace derivatives inherit the Kronecker--sum bottlenecks in~\eqref{eq:12}, and second-order derivatives involve repeated Kronecker--sum solves together with additional Kronecker products (Supplementary Material, Section~\ref{sec:A.4}). Finally, forming the Laplace covariance still requires solving a linear system in dimension $u(r-u)$, but this step is substantially less burdensome than the Kronecker--sum computations that dominate the Euclidean parameterization.

\paragraph{Laplace update for the nonconjugate block $\bA$.}
At iteration $t$, given $q_{-\bA}^{(t)}$, we form $\tf^{(t)}(\bA)$ in~\eqref{eq:16} and update
$q^{(t+1)}(\vecr(\bA))$ using the Laplace step~\eqref{eq:10}:
\begin{align*}
&q^{(t+1)}(\vecr(\bA)) \approx q^{L}(\vecr(\bA)) = \calN_{(r-u)u}\!\left(\vecr(\hbA^{(t)}),\, H(\vecr(\hbA^{(t)}))\right),
\tag{19}\label{eq:19}
\end{align*}
where
\begin{align*}
&\hbA^{(t)}=\arg\max_{\bA}\ \tf^{(t)}(\bA), \\
&H(\vecr(\hbA^{(t)}))= -\left[ \nabla^2_{\vecr(\bA)}\tf^{(t)}(\hbA^{(t)})\right]^{-1}.
\end{align*}
In practice, $\hbA^{(t)}$ is obtained via numerical optimization \citep{cook2016note} based on the derivatives in \eqref{eq:17}.

\paragraph{Closed-form updates for the conjugate blocks.}
After updating the nonconjugate envelope factor $q^{(t+1)}(\vecr(\bA))$ via the Laplace step,
the remaining parameter blocks become conditionally conjugate under the MFVF
assumption~\eqref{eq:15}. Consequently, their coordinate updates admit closed-form
expressions obtained from~\eqref{eq:8}, and we choose the corresponding variational
families as
\begin{align*}
    & q^{(t+1)}(\btmu) \sim \calN_{r}(\bmu^{(t+1)}_{q}, \bSigma^{(t+1)}_{q}), \nonumber\\
    & q^{(t+1)}(\boldsymbol{\tilde{\eta}}) \sim
      \mathcal{MN}_{u,p}(\boldsymbol{\eta}^{(t+1)}_{q},
      \bU^{(\boldsymbol{\eta})(t+1)}_{q},
      \bV^{(\boldsymbol{\eta})(t+1)}_{q}), \nonumber\\
    & q^{(t+1)}(\btOmega) \sim
      \mathcal{IW}_{u}(\bPsi^{(1)(t+1)}_{q}, \nu^{(1)(t+1)}_{q}), \nonumber\\
    & q^{(t+1)}(\btOmega_0) \sim
      \mathcal{IW}_{r-u}(\bPsi^{(0)(t+1)}_{q}, \nu^{(0)(t+1)}_{q}).
\tag{20}\label{eq:20}
\end{align*}
These families are chosen to match the functional forms of the corresponding full
conditional posterior distributions under the reparameterized model.

The resulting variational parameters can be expressed in closed form as functions of
(i) the current Laplace approximation $(\hbA^{(t)},\, H(\vecr(\hbA^{(t)})))$,
(ii) the current variational moments of the other conjugate blocks, and
(iii) sufficient statistics of the data $(\bbX,\bbY)$.
For example, the update for $\boldsymbol{\tilde{\eta}}$ is given by
\begin{align*}
    q^{(t+1)}(\boldsymbol{\tilde{\eta}})
    &\sim \mathcal{MN}_{u,p}(\boldsymbol{\eta}^{(t+1)}_q,
    \boldsymbol{U}^{(\boldsymbol{\eta})(t+1)}_q,
    \boldsymbol{V}^{(\boldsymbol{\eta})(t+1)}_q),\\
    \boldsymbol{\eta}^{(t+1)}_q
    &= \bC_{\hbA^{(t)}}^\top
       (\mathbb{X}_c^\top \mathbb{Y} + \boldsymbol{M}\boldsymbol{B}_0^\top)^\top
       (\mathbb{X}_c^\top \mathbb{X}_c + \boldsymbol{M})^{-1},\\
    \boldsymbol{U}^{(\boldsymbol{\eta})(t+1)}_q
    &= \bPsi^{(1)(t+1)}_q/\nu^{(1)(t+1)}_q, \\
    \boldsymbol{V}^{(\boldsymbol{\eta})(t+1)}_q
    &= (\mathbb{X}_c^\top \mathbb{X}_c + \boldsymbol{M})^{-1},
\end{align*}
where $\bC_{\hbA^{(t)}}=\E_{q^{(t+1)}(\vecr(\bA))}[\bC_{\bA}]$ is computed under the Gaussian
Laplace factor~\eqref{eq:19}.
Explicit expressions for the remaining updates are provided in the Supplementary Material, Section~\ref{sec:B.2}.

\paragraph{Algorithm summary and convergence monitoring.}
For convergence monitoring, we evaluate the approximate ELBO at each iteration:
\begin{align*}
\tilde{\calL}^{(t)}(q)
&\approx \tf^{(t)}(\hbA^{(t)}) - \frac{(r-u)u}{2}\\
&\quad - \E_{q^{(t)}(\btmu, \boldsymbol{\tilde{\eta}}, \btOmega, \btOmega_0, \bA)}\!
\left[
    \log q^{(t)}(\btmu, \boldsymbol{\tilde{\eta}}, \btOmega, \btOmega_0, \bA)
\right],
\tag{21}\label{eq:21}
\end{align*}
where $\tf^{(t)}(\hbA^{(t)})$ denotes the plug-in Laplace approximation to the nonconjugate contribution. The full expression of $\tilde{\calL}(q)$ (including constants) is provided in the Supplementary
Material, Section~\ref{sec:B.3}.

Algorithm~\ref{alg1} summarizes the resulting CALVI procedure. At iteration $t$, we update
$q^{(t+1)}(\vecr(\bA))$ via~\eqref{eq:19} and then update the conjugate variational factors via the closed-form
expressions underlying~\eqref{eq:20}, cycling until the relative approximate ELBO improvement falls below a
tolerance.

\begin{algorithm}
\caption{Coordinate ascent Laplace variational inference for the response envelope model}
\label{alg1}
\begin{algorithmic}[1]
    \State Initialize variational factors $q^{(1)}(\btmu)$, $q^{(1)}(\boldsymbol{\tilde{\eta}})$,
           $q^{(1)}(\btOmega)$, $q^{(1)}(\btOmega_0)$, and $q^{(1)}(\vecr(\bA))$
    \State Set tolerance $\epsilon$, maximum iterations $T$
    \State $t \gets 1$
    \While{$t < T$}
        \State Compute $\tilde{\calL}(q)^{(t)}$ using \eqref{eq:21}
        \State Update $q^{(t+1)}(\vecr(\bA))$ using \eqref{eq:19}
        \State Update $q^{(t+1)}(\boldsymbol{\tilde{\eta}})$ using closed-form CAVI updates (Supplementary Material, Section~\ref{sec:B.2})
        \State Update $q^{(t+1)}(\btOmega)$ using closed-form CAVI updates (Supplementary Material, Section~\ref{sec:B.2})
        \State Update $q^{(t+1)}(\btOmega_0)$ using closed-form CAVI updates (Supplementary Material, Section~\ref{sec:B.2})
        \State Update $q^{(t+1)}(\btmu)$ using closed-form CAVI updates (Supplementary Material, Section~\ref{sec:B.2})
        \State Compute $\tilde{\calL}(q)^{(t+1)}$ using \eqref{eq:21}
        \If{$|\tilde{\calL}(q)^{(t+1)} - \tilde{\calL}(q)^{(t)}|
             < \epsilon\,|\tilde{\calL}(q)^{(t+1)}|$}
            \State \textbf{break}
        \EndIf
        \State $t \gets t + 1$
    \EndWhile
    \State \Return $q^{(t)}(\btmu)$, $q^{(t)}(\boldsymbol{\tilde{\eta}})$, $q^{(t)}(\btOmega)$, $q^{(t)}(\btOmega_0)$,
           $q^{(t)}(\vecr(\bA))$, and $\{\tilde{\calL}(q)^{(1)} ,\dots, \tilde{\calL}(q)^{(t)}\}$
\end{algorithmic}
\end{algorithm}

\section{Theoretical justification of the CALVI Laplace update}
\label{sec5}

This section quantifies the accuracy of the Laplace approximation used in CALVI.
We first establish a generic result for the one-step Laplace approximation in the
nonconjugate coordinate update~\eqref{eq:10}, and then apply it to the Bayesian response
envelope model to justify the model-specific update~\eqref{eq:19}.

\paragraph{Laplace approximation error in a general variational setting.}
Recall the notation of Section~\ref{sec3}. The target posterior is $p(\btheta\mid\bbY_n)$, with a block decomposition $\btheta= (\btheta_1,\btheta_2)$.
Under the MFVF assumption $q(\btheta)=q_1(\btheta_1)\,q_2(\btheta_2)$, the block
$\btheta_1\in\bbR^{d_1}$ denotes the conditionally nonconjugate component, while
$\btheta_2\in\bbR^{d_2}$ denotes the conditionally conjugate component.
Since we focus on the one-step Laplace approximation error in~\eqref{eq:10}, we fix the
variational factor for the conjugate block at the current CAVI iterate, denote it by
$q_{2,n}^{\dagger}(\btheta_2)$, and suppress the explicit iteration index $(t)$.
Throughout this section, all quantities with subscript $n$ may depend on the sample size $n$,
whereas the dimensions $d_1$ and $d_2$ are fixed.

The exact CAVI coordinate update for the nonconjugate block $\btheta_1$ is the distribution with density proportional to
\begin{align*}
q_n^{VI}(\btheta_1 \mid q_{2,n}^{\dagger})
\;\propto\;
\exp\!\big\{\tf_n^\dagger(\btheta_1)\big\},
\label{eq:def-qVI} \tag{22}
\end{align*}
which coincides with the $i=1$ instance of~\eqref{eq:8} when the conjugate variational factor
$q_{2,n}^{\dagger}$ is held fixed. The exponent $\tf_n^\dagger(\btheta_1)$, defined in~\eqref{eq:tfdef} below, matches the kernel in~\eqref{eq:def-qVI}. Any term independent of $\btheta_1$ is absorbed into the normalizing constant.

In CALVI, the one-step coordinate update for the nonconjugate block replaces the exact update~\eqref{eq:def-qVI} with its Laplace approximation
\begin{align*}
q_n^{L}(\btheta_1)
\;:=\;
\calN\!\left(\hat{\btheta}_{1_n},
\big[-\tf_n^{\dagger\prime\prime}(\hat{\btheta}_{1_n})\big]^{-1}\right),
\label{eq:def-qL} \tag{23}
\end{align*}
where $\hat{\btheta}_{1_n}:=\arg\max_{\btheta_1}\tf_n^\dagger(\btheta_1)$. The primary goal is to show that, for fixed $q_{2,n}^{\dagger}$ and after rescaling, the Laplace-based update~\eqref{eq:def-qL} is asymptotically equivalent to the exact CAVI update~\eqref{eq:def-qVI} in TV distance.

To formalize this comparison on the local $n^{-1/2}$ scale, define the rescaled variable
\[
\bX := \sqrt n\,(\btheta_1-\hat{\btheta}_{1_n}),
\]
and let $\overline{q_n^{VI}}$ and $\overline{q_n^{L}}$ denote the laws of $\bX$ under $q_n^{VI}(\cdot\mid q_{2,n}^{\dagger})$ and $q_n^{L}$, respectively. Our objective is to establish that
\[
d_{TV}\!\left(\overline{q_n^{L}},\,\overline{q_n^{VI}}\right)\xrightarrow{P}0,
\qquad n\to\infty,
\]
under suitable regularity conditions.

Before stating these conditions, we introduce terminology that will be used throughout the analysis. With $q_{2,n}^\dagger$ fixed, write the variational log-likelihood
\begin{align*}
\ell_n^\dagger(\btheta_1)
:= \int q_{2,n}^\dagger(\btheta_2)\,
\log p(\bbY_n,\btheta_2\mid\btheta_1)\, d\btheta_2,
\tag{24}\label{eq:elldef}
\end{align*}
and write the corresponding variational log-posterior (up to an additive constant in $\btheta_1$) as
\begin{align*}
\tf_n^\dagger(\btheta_1)
= \ell_n^\dagger(\btheta_1) + \log\pi(\btheta_1).
\tag{25}\label{eq:tfdef}
\end{align*}
Equivalently, the exact coordinate-update distribution~\eqref{eq:def-qVI} is the normalized law with density proportional to $\exp\{\tf_n^\dagger(\btheta_1)\}$. All assumptions introduced below are imposed on $\ell_n^\dagger$ and $\tf_n^\dagger$.

We impose Assumptions~\ref{assum1}--\ref{assum5}, adapted from \citet[Theorem~17]{kasprzak2025laplace} to the present variational coordinate-update setting. These conditions coincide in structure with those of \citet{kasprzak2025laplace}, but are stated here in terms of the variational log-likelihood and log-posterior relevant to CALVI.
\begin{assumption}
\label{assum1}
Let $\ell_n^{\dagger}(\btheta_1)$ denote the variational log-likelihood defined in~\eqref{eq:elldef}. Assume that:
\begin{enumerate}
\item[(a)] There exists a unique maximum variational likelihood estimator (MVLE)
\[
\bar{\btheta}_{1_n}
:= \arg\max_{\btheta_1}\,\ell_n^{\dagger}(\btheta_1).
\]

\item[(b)] There exists $\bar{\delta}>0$ such that $\ell_n^{\dagger}(\btheta_1)$ is three-times continuously differentiable on the closed ball
\[
\calB(\bar{\btheta}_{1_n},\bar{\delta})
:= \big\{\,\btheta_1:\ \|\btheta_1-\bar{\btheta}_{1_n}\|\le \bar{\delta}\,\big\}.
\]

\item[(c)] There exists a constant $\bar{M}_1>0$ such that
\[
\sup_{\btheta_1 \in \calB(\bar{\btheta}_{1_n},\bar{\delta})}
\frac{\big\|\,\ell_n^{\dagger\prime\prime\prime}(\btheta_1)\,\big\|^{*}}{n}
\le \bar{M}_1,
\]
where the norm of the third derivative is defined by
\begin{align*}
    g^{\prime \prime \prime}(\theta)[a, b, c] := \sum_{i,j,k=1}^{d_1}
    \frac{\partial^3 g(\theta)}{\partial \theta_{i}\partial \theta_{j}\partial \theta_{k}}
    \, a_{i} b_{j} c_{k}, \qquad
    \| g^{\prime \prime \prime}(\theta) \|^{*}
    := \sup_{\|a\| \leq 1,\ \|b\| \leq 1,\ \|c\| \leq 1}
    \left| g^{\prime \prime \prime}(\theta)[a, b, c] \right|.
\end{align*}
\end{enumerate}
\end{assumption}
Assumption~\ref{assum1} ensures that the variational posterior admits a local quadratic approximation on the neighborhood $\calB(\bar{\btheta}_{1_n}, \bar{\delta})$ around the MVLE $\bar{\btheta}_{1_n}$.

\begin{assumption}
\label{assum2}
For the same $\bar{\delta} > 0$ and $\calB(\bar{\btheta}_{1_n}, \bar{\delta})$ as in Assumption~\ref{assum1},
there exists a constant $\bar{M}_2 > 0$ such that
\[
\sup_{\btheta_1 \in \calB(\bar{\btheta}_{1_n}, \bar{\delta})}
\left|\frac{1}{\pi(\btheta_1)}\right|
\le \bar{M}_2.
\]
\end{assumption}
Assumption~\ref{assum2} guarantees that $\calB(\bar{\btheta}_{1_n}, \bar{\delta})$ has strictly positive prior probability, preventing the local posterior behavior
from being dominated by the prior.

\begin{assumption}
\label{assum3}
Let $\tf_n^{\dagger}(\btheta_1)$ denote the variational log-posterior defined in~\eqref{eq:tfdef}, and let
$\lambda_{\min}(\cdot)$ denote the minimum eigenvalue of a symmetric matrix. Assume that:
\begin{enumerate}
\item[(a)] There exists a unique maximum variational posterior (MAVP),
\[
\hat{\btheta}_{1_n}
:= \arg\max_{\btheta_1}\, \tf_n^{\dagger}(\btheta_1).
\]

\item[(b)] There exists $\hat{\delta}>0$ such that the log prior $\log\pi(\btheta_1)$ is
three-times continuously differentiable on the closed ball
\[
\calB(\hat{\btheta}_{1_n},\hat{\delta})
:= \big\{\btheta_1:\ \|\btheta_1-\hat{\btheta}_{1_n}\|\le \hat{\delta}\big\}.
\]

\item[(c)] There exists a constant $\hat{M}_1>0$ such that
\begin{align*}
\sup_{\btheta_1 \in \calB(\hat{\btheta}_{1_n},\hat{\delta})}
\frac{\big\| \tf_n^{\dagger\prime\prime\prime}(\btheta_1)\big\|^{*}}{n}
\le \hat{M}_1, \qquad
\lambda_{\min}\!\left(-\frac{\tf_n^{\dagger\prime\prime}(\hat{\btheta}_{1_n})}{n}\right)
> \hat{\delta}\,\hat{M}_1 .
\end{align*}
\end{enumerate}
\end{assumption}
Assumption~\ref{assum3} ensures that, after incorporating the prior, the induced
variational log-posterior remains strongly concave in a neighborhood of the MAVP,
thereby justifying a Gaussian approximation at the posterior level.

\begin{assumption}
\label{assum4}
For the same $\bar{\delta} > 0$ as in Assumption~\ref{assum1} and the same $\hat{\delta} > 0$ as in Assumption~\ref{assum3},
\begin{align*}
\max \left\{
\|\bar{\btheta}_{1_n} - \hat{\btheta}_{1_n}\|,
\ \sqrt{ \tr \!\left[ - \tf_n^{\dagger\prime\prime}(\hat{\btheta}_{1_n})^{-1} \right] }
\right\}
< \hat{\delta}, \qquad 
\sqrt{ \tr \!\left[ - \ell_n^{\dagger\prime\prime}(\bar{\btheta}_{1_n})^{-1} \right] }
< \bar{\delta}.
\end{align*}
\end{assumption}
Assumption~\ref{assum4} ensures that, for sufficiently large $n$, the local Gaussian
scales induced by the Hessians fit within the neighborhoods
$\calB(\bar{\btheta}_{1_n}, \bar{\delta})$ and
$\calB(\hat{\btheta}_{1_n}, \hat{\delta})$, and that the MVLE and MAVP are sufficiently
close for the corresponding local quadratic approximations to overlap.

\begin{assumption}
\label{assum5}
For the same $\hat{\delta} > 0$ as in Assumption~\ref{assum3}, there exists a constant $\hat{\kappa} > 0$ such that
\[
\sup_{\btheta_1:\ \| \btheta_1 - \bar{\btheta}_{1_n} \| > \hat{\delta}
- \| \bar{\btheta}_{1_n} - \hat{\btheta}_{1_n}\|}
\frac{ \ell_n^{\dagger}(\btheta_1) - \ell_n^{\dagger}(\bar{\btheta}_{1_n})}{n}
\le -\hat{\kappa}.
\]
\end{assumption}
Assumption~\ref{assum5} guarantees sufficient tail decay of the variational
log-likelihood, ensuring that regions away from the MVLE contribute negligibly to the
variational posterior.

\begin{theorem}
\label{main-theorem}
Suppose that Assumptions~\ref{assum1}--\ref{assum5} hold and adopt the notation above. Then
\[
d_{TV}\!\left(\overline{q_n^L},\overline{q_n^{VI}}\right)\xrightarrow{P}0, \qquad \text{as} \quad  n \rightarrow \infty.
\]
\end{theorem}
This theorem establishes a local one-step equivalence between the Laplace coordinate update and the exact CAVI update; it does not address the global behavior of the full iterative scheme. The proof is based on an explicit nonasymptotic bound on the TV distance and is deferred to Supplementary Material, Section~\ref{sec:D}.

\paragraph{Verification for the Bayesian response envelope model.}
We now introduce a model-specific assumption under which
Assumptions~\ref{assum1}--\ref{assum5} are satisfied for the reparameterized
Bayesian response envelope model.
\begin{assumption}
\label{assum6}
Let $\bA^{\star}$ denote the true parameter under the chosen identification implied by~\eqref{eq:4}.
There exist $\rho>0$ and $\delta>0$ such that, letting
\begin{align*}
\mathcal N_{\rho} &:=\{\bA:\|\bA-\bA^\star\|_F\le\rho\},\\
\mathbf{J}_n(\bA) &:= -\bbE_{q_n^{\dagger}(\bZ)} \big[\nabla^2_{\vecr(\bA)} \log p(\bbY_n, \bZ \mid \bA) \big],
\end{align*}
the following hold:
\begin{enumerate}
\item[(a)]
With probability tending to one, the global maximizer sets are nonempty and satisfy
\begin{align*}
\bbA_n:=\argmax_{\bA}\ell^{\dagger}_n(\bA)\ \subset\ \operatorname{int}(\mathcal N_{\rho}), \qquad
\hbA_n:=\argmax_{\bA}\tf_n^\dagger(\bA)\ \subset\ \operatorname{int}(\mathcal N_{\rho}).
\end{align*}
In particular, on this event the maximizers over $\mathcal N_{\rho}$ coincide with the global maximizers.

\item[(b)]
With probability tending to one,
\[
\inf_{\bA\in\mathcal N_{\rho}}\lambda_{\min}\!\Big(\frac{1}{n}\mathbf J_n(\bA)\Big)\ \ge\ \delta.
\]
\end{enumerate}
\end{assumption}
Assumption~\ref{assum6}(a) is motivated by existing localization results for envelope estimators. In particular, under standard regularity conditions the envelope MLE is consistent for $\bA^\star$ \citep{cook2010envelope, cook2016note}, so it is reasonable to expect that, for sufficiently large $n$, the maximizers of $\ell_n^\dagger(\bA)$ and $\tf_n^\dagger(\bA)$ concentrate in a (model-dependent) neighborhood of $\bA^\star$; Assumption~\ref{assum6}(a) formalizes this by postulating the existence of a compact ball $\mathcal N_\rho$
that contains the global maximizers with high probability.

From a computational standpoint, when this localization event holds, initializing the numerical maximization of $\tf_n^\dagger(\bA)$ at a consistent pilot estimator such as the envelope MLE provides a practical way to target the relevant maximizer inside $\mathcal N_\rho$.

Assumption~\ref{assum6}(b) is a curvature condition that is consistent with standard information-matrix behavior in regular parametric models: the information for $\bA$ is positive definite at $\bA^\star$ and varies continuously in $\bA$ \citep{cook2010envelope}, so for $\rho$ chosen sufficiently small one expects a uniform eigenvalue lower bound on $\mathcal N_\rho$. Since $\mathbf J_n(\bA)$ averages the observed information under the fixed variational distribution $q_n^\dagger(\bZ)$, this assumption ensures that the Laplace covariance in the $\bA$-update is well-defined and uniformly well-conditioned. Consistent with this assumption, our numerical experiments in Section~\ref{sec6} did not exhibit near-singularity of the Hessian along CALVI iterations.

Theorem~\ref{main-theorem} applies to the reparameterized model, yielding the following Theorem~\ref{thm:yenv}.
\begin{theorem}\label{thm:yenv}
Under the reparameterized response envelope model, suppose that Assumption~\ref{assum6} holds.
Define
\[
\bX \;:=\; \sqrt{n}\,\big(\vecr(\bA)-\vecr(\hbA_n)\big),
\]
with $\hbA_n$ denoting the MAVP for $\bA$, corresponding to
$\hat{\btheta}_{1_n}$ under the identification $\btheta_1\equiv\bA$.
Let $\overline{q_n^{L}}$ and $\overline{q_n^{VI}}$ denote the laws of $\bX$ under
$\bA\sim q_n^{L}$ and $\bA\sim q_n^{VI}$, respectively.
Then
\[
d_{TV}\!\left(\overline{q_n^{L}}, \overline{q_n^{VI}}\right)\xrightarrow{P} 0.
\]
\end{theorem}
The proof of Theorem~\ref{thm:yenv} is given in Supplementary Material, Section~\ref{sec:E}.

\section{Numerical Studies} 
\label{sec6}

\subsection{Simulation studies}

We conduct simulation studies of the Bayesian response envelope model. The performance of CALVI is assessed in two aspects: (1) estimation accuracy of $\bbeta$ and (2) computation time. We compare CALVI with three competing methods: (i) the Euclidean MH–within–Gibbs sampler with a likelihood-driven proposal
proposed by \citet{chakraborty2024comprehensive} (hereafter, MH-Gibbs); (ii) the ADVI method implemented by \citet{chakraborty2024comprehensive} in their Supplementary Material; and (iii) the Stiefel manifold Gibbs sampler of \citet{khare2017bayesian}
(hereafter, Manifold-Gibbs).

Unless stated otherwise, CALVI is terminated using the stopping rule in Algorithm~\ref{alg1} with tolerance $\epsilon = 10^{-6}$ and maximum iterations $T = 10{,}000$. For the MCMC-based methods (MH-Gibbs and Manifold-Gibbs), we run each chain for $10{,}000$ iterations, discarding the first $5{,}000$ iterations as burn-in. Although the notion of an iteration differs between optimization-based and sampling-based methods, this choice is intended to provide a roughly comparable computational budget across methods.

Throughout, the superscript ${*}$ denotes the true values of the parameters. We generate 100 replicated datasets according to the response envelope model in~\eqref{eq:4}, with $n=100, 200, 500, 1000$, $r=20$, $p=7$ for each true dimension of the envelope $u^{*}=2, 5$. The elements of $\bmu^{*}$, $\boldsymbol{\eta}^{*}$, and $\bA^{*}$ are independently sampled from Uniform (0, 10), Uniform (0, 10), and Uniform (-1, 1), respectively. The matrices $\bOmega^{*}$ and $\bOmega_0^{*}$ are diagonal, with diagonal elements independently drawn from Uniform(0, 1) and Uniform(5, 10), respectively.

We specify the following vague priors for each parameter: $\boldsymbol{\eta} \sim \mathcal{MN}(\boldsymbol{0}, \bOmega, 10^{6} \bI_p)$, $\bA \sim \mathcal{MN}(\boldsymbol{0}, 10^{6} \bI_{r-u}, 10^6 \bI_u)$ , $\bOmega \sim \mathcal{IW}_u(u,10^{-6} \bI_u)$, and $\bOmega_0 \sim \mathcal{IW}_{r-u}(r-u,10^{-6} \bI_{r-u})$. These prior specifications are the same as those used in the simulation study of \citet{chakraborty2024comprehensive}. For the Manifold-Gibbs sampler, which directly parameterizes the envelope subspace on the Stiefel manifold and does not involve the $\bA$-parameterization, we employ the conventional noninformative prior used in \citet{khare2017bayesian}. In particular, a uniform (Haar) prior is imposed on the orthonormal basis matrix defining the envelope subspace, together with weakly informative diagonal priors on $\bOmega$ and $\bOmega_0$. Due to these fundamental differences in parameterization, the priors cannot be matched exactly across methods; instead, each method is evaluated under its standard vague prior specification.

In this simulation study, the envelope dimension $u$ is treated as unknown and estimated from the data. For the Bayesian methods, we estimate $u$ via model averaging using a BIC-based approximation to the marginal likelihood \citep{Kass01061995}. Specifically, the posterior distribution of $u$ is approximated as:
\begin{align*}
    p(u = \mathcal{M} \mid \bbY) = 
     \quad \frac{\exp \left( - \mathrm{BIC}(\mathcal{M})/2 \right)\,\pi(u=\mathcal{M})}
    {\sum_{\mathcal{M}^{\prime}=0}^{r} \exp \left( - \mathrm{BIC}(\mathcal{M}^{\prime})/2 \right)\,\pi(u=\mathcal{M}^{\prime})}, \quad \mathcal{M} = 0, \dots, r,
\end{align*}
where $\mathrm{BIC}(\mathcal{M}) = -2\,\hat{\mathcal{L}}(\mathcal{M}) + d_{\mathcal{M}} \log(n)$. 
Here, $\hat{\mathcal{L}}(\mathcal{M})$ denotes the log-likelihood evaluated at the variational posterior mean for CALVI, and at the maximum over posterior samples for each MCMC method. The effective number of parameters, $d_{\mathcal{M}}$, for the response envelope model is given by $r + r(r+1)/2 + up$.

The mean posterior probabilities for the envelope dimension $u$ across 100 replicated datasets are reported in Table~\ref{table1}. Overall, all three methods increasingly concentrate on the true envelope dimension as the sample size grows. When $u^{*}=2$, CALVI assigns posterior probability 0.684 to $u=2$ at $n=100$, with most of the remaining mass placed on $u=3$ (0.270). MH--Gibbs yields a slightly higher probability of 0.749 at $n=100$ and similarly places the remaining mass primarily on $u=3$ (0.215). Both CALVI and MH--Gibbs rapidly concentrate on the correct dimension as $n$ increases, reaching 0.905 at $n=200$ and at least 0.987 by $n=500$, and essentially 1.000 at $n=1000$. In contrast, the Manifold--Gibbs sampler produces a much more concentrated dimension posterior, assigning 1.000 mass on $u=2$ across all sample sizes considered. When $u^{*}=5$, CALVI and MH--Gibbs again put most posterior mass on the correct dimension at $n=100$ (0.833 and 0.808, respectively), with the remaining probability largely on $u=6$. The concentration strengthens quickly with $n$, with CALVI reaching 0.993 at $n=200$ (versus 0.938 for MH--Gibbs) and both methods essentially selecting $u=5$ with probability at least 0.996 by $n=500$. The Manifold--Gibbs sampler assigns posterior probability 1.000 to $u=5$ across all sample sizes considered.

To evaluate the estimation accuracy of $\bbeta$, we define $\hat{\bbeta}$ using Bayesian model averaging (BMA) as follows:
\begin{align*} 
    \hat{\boldsymbol{\beta}} = \sum_{\mathcal{M}=0}^{r} \mathbb{E} (\boldsymbol{\beta} \mid u = \mathcal{M}, \mathbb{Y})\, p(u = \mathcal{M} \mid \mathbb{Y}).
\end{align*}
We computed the mean squared error (MSE) for each replicated dataset $k$ as $\text{MSE}_{k} = \| \hat{\boldsymbol{\beta}}_{k} - \boldsymbol{\beta}^{*} \|_{\text{F}}^2$.

In Table~\ref{table2}, we summarize the MSE of $\hat{\bbeta}$ using the average $\overline{\text{MSE}} = \frac{1}{100} \sum_{k=1}^{100} \text{MSE}_k$ and the corresponding standard deviation $\sqrt{\frac{1}{100} \sum_{k=1}^{100}(\text{MSE}_k - \overline{\text{MSE}})^2}$ across 100 replicated datasets. For $u^{*}=2$, CALVI exhibits slightly larger MSE at the smallest sample size ($2.15$ with sd $1.13$ at $n=100$) than MH--Gibbs ($1.74$ with sd $0.69$) and Manifold--Gibbs ($1.51$ with sd $0.42$). However, the gap shrinks rapidly as $n$ increases. 
By $n=500$, all three methods attain essentially identical MSE 
($0.28$ with sd $0.08$), and by $n=1000$ their MSEs are nearly indistinguishable ($0.13$--$0.14$ with comparable standard deviations). For $u^{*}=5$, a similar convergence pattern is observed. At $n=100$, CALVI yields a slightly larger MSE ($3.02$ with sd $0.75$) than MH--Gibbs ($2.87$ with sd $0.54$) and Manifold--Gibbs ($2.72$ with sd $0.60$). From $n=200$ onward, however, the three methods produce nearly identical and 
monotonically decreasing MSE values, reaching approximately $0.24$--$0.25$ at $n=1000$.

In contrast, ADVI displays substantially larger estimation error and markedly higher variability in both envelope settings. For $u^{*}=2$, the mean MSE is extremely inflated at $n=100$ ($389.88$ with sd $119.83$) 
and, although it decreases sharply as $n$ increases, it remains unstable and highly dispersed even at $n=1000$ ($4.91$ with sd $33.43$). The instability is even more pronounced for $u^{*}=5$, where ADVI produces very large and highly variable MSE across all sample sizes (e.g., $617.26$ with sd $267.91$ at $n=100$ and $83.49$ with sd $123.61$ at $n=1000$). These findings suggest that naive gradient-based VI under the standard Euclidean parameterization can be numerically fragile in envelope models, particularly when the envelope dimension is moderate or large.

The average computational times across 100 replicated datasets are presented in Table~\ref{table3}. The reported times include fitting all candidate dimensions $u=1,\ldots,20$ to form the BIC-based weights. Across both $u^{*}\in\{2,5\}$ scenarios, CALVI is orders of magnitude faster than the MCMC-based methods. When $u^{*}=2$, CALVI requires 11.8--33.0 seconds across the considered sample sizes, whereas MH--Gibbs takes 2850.7--3075.7 seconds and Manifold--Gibbs takes 13101.5--26629.7 seconds. When $u^{*}=5$, CALVI takes 14.0--19.9 seconds, compared to 3350.1--3526.1 seconds for MH--Gibbs and 12295.3--26112.1 seconds for Manifold--Gibbs. Relative to MH--Gibbs, CALVI uses about 0.38\%--1.16\% of the runtime (and only about 0.04\%--0.25\% relative to Manifold--Gibbs), and these ratios decrease as $n$ increases, illustrating the scalability of CALVI.

We compare the posterior densities obtained by CALVI and two MCMC approaches (MH--Gibbs and Manifold--Gibbs) in Figure~\ref{fig:beta marginal}. The figure displays marginal posterior densities for each entry of $\bbeta$. Overall, all methods recover the true coefficients satisfactorily. The two MCMC samplers produce nearly indistinguishable marginal distributions in this replicate. Compared with the MCMC benchmarks, CALVI yields slightly more concentrated marginal posteriors for several coefficients, particularly in the top-$u$ rows. This mild underestimation of posterior variability may be explained by the reparameterization through $\bA$: variability associated with the envelope subspace can be partially redistributed to other model components under the variational approximation. Because CALVI adopts the MFVF, residual dependence among these components is not fully captured, which can lead to modest shrinkage of marginal variances while preserving posterior centers.

\begin{table*}[!t]
\centering
\small
\caption{Mean posterior probabilities for the envelope dimension $u$ by method, with $r=20$, $p=7$, and $u^{*}\in\{2,5\}$. Results are averaged over 100 replications.}
\label{table1}
\setlength{\tabcolsep}{4pt} 
\begin{tabular}{l|cccc|cccc|cccc}
\toprule
 & \multicolumn{4}{c|}{CALVI}
 & \multicolumn{4}{c|}{MH--Gibbs}
 & \multicolumn{4}{c}{Manifold--Gibbs} \\
\midrule
$n$ & $u \le 1$ & $u=2$ & $u=3$ & $4 \le u$
    & $u \le 1$ & $u=2$ & $u=3$ & $4 \le u$
    & $u \le 1$ & $u=2$ & $u=3$ & $4 \le u$ \\
\midrule
100  & 0.000 & 0.684 & 0.270 & 0.045 & 0.000 & 0.749 & 0.215 & 0.036 & 0.000 & 1.000 & 0.000 & 0.000 \\
200  & 0.000 & 0.905 & 0.094 & 0.001 & 0.000 & 0.905 & 0.095 & 0.000 & 0.000 & 1.000 & 0.000 & 0.000 \\
500  & 0.000 & 0.997 & 0.003 & 0.000 & 0.000 & 0.987 & 0.013 & 0.000 & 0.000 & 1.000 & 0.000 & 0.000 \\
1000 & 0.000 & 1.000 & 0.000 & 0.000 & 0.000 & 0.999 & 0.001 & 0.000 & 0.000 & 1.000 & 0.000 & 0.000 \\
\midrule
\multicolumn{13}{c}{(a) The true $u$ is 2.} \\
\midrule
$n$ & $u \le 4$ & $u=5$ & $u=6$ & $7 \le u$
    & $u \le 4$ & $u=5$ & $u=6$ & $7 \le u$
    & $u \le 4$ & $u=5$ & $u=6$ & $7 \le u$ \\
\midrule
100  & 0.000 & 0.833 & 0.156 & 0.011 & 0.000 & 0.808 & 0.180 & 0.011 & 0.000 & 1.000 & 0.000 & 0.000 \\
200  & 0.000 & 0.993 & 0.007 & 0.000 & 0.000 & 0.938 & 0.062 & 0.000 & 0.000 & 1.000 & 0.000 & 0.000 \\
500  & 0.000 & 1.000 & 0.000 & 0.000 & 0.000 & 0.996 & 0.004 & 0.000 & 0.000 & 1.000 & 0.000 & 0.000 \\
1000 & 0.000 & 1.000 & 0.000 & 0.000 & 0.000 & 1.000 & 0.000 & 0.000 & 0.000 & 1.000 & 0.000 & 0.000 \\
\midrule
\multicolumn{13}{c}{(b) The true $u$ is 5.} \\
\bottomrule
\end{tabular}
\end{table*}

\begin{table*}[!t]
\centering
\caption{
MSE for $\bbeta$ (sum over coordinates) by method—BMA via CALVI, MH--Gibbs, Manifold--Gibbs, and ADVI—with $r=20$, $p=7$, and $u^{*}\in\{2,5\}$. 
Entries are mean (sd) over 100 replications. 
For ADVI, results are reported only at the true envelope dimension $u=u^{*}$.
}
\label{table2}
\small
\setlength{\tabcolsep}{6pt}
\begin{tabular}{lcccc}
\toprule
Method & $n=100$ & $n=200$ & $n=500$ & $n=1000$ \\
\midrule
CALVI         & 2.15 (1.13) & 0.86 (0.43) & 0.28 (0.08) & 0.14 (0.04) \\
MH--Gibbs     & 1.74 (0.69) & 0.80 (0.31) & 0.28 (0.08) & 0.13 (0.04) \\
Manifold--Gibbs & 1.51 (0.42) & 0.73 (0.22) & 0.28 (0.08) & 0.13 (0.04) \\
ADVI          & 389.88 (119.83) & 10.205 (51.28) & 3.66 (29.95) & 4.91 (33.43) \\
\midrule
\multicolumn{5}{c}{(a) The true $u$ is 2.} \\
\midrule
CALVI         & 3.02 (0.75) & 1.32 (0.25) & 0.50 (0.09) & 0.25 (0.05) \\
MH--Gibbs     & 2.87 (0.54) & 1.33 (0.25) & 0.49 (0.09) & 0.24 (0.04) \\
Manifold--Gibbs & 2.72 (0.60) & 1.31 (0.16) & 0.48 (0.05) & 0.25 (0.03) \\
ADVI          & 617.26 (267.91) & 111.01 (128.82) & 65.11 (102.74) & 83.49 (123.61) \\
\midrule
\multicolumn{5}{c}{(b) The true $u$ is 5.} \\
\bottomrule
\end{tabular}
\end{table*}

\begin{table}[!t]
\centering
\caption{Mean total CPU time (seconds) across $u=1,\ldots,20$ over 100 replications ($r=20$, $p=7$; $u^* \in \{2,5\}$).}
\label{table3}
\small
\setlength{\tabcolsep}{6pt}
\begin{tabular}{lccc}
\toprule
 & CALVI & MH--Gibbs & Manifold--Gibbs \\
\midrule
$n=100$  & 33.0 & 2850.7 & 13101.5 \\
$n=200$  & 17.7 & 2892.6 & 16447.2 \\
$n=500$  & 12.6 & 2970.8 & 17838.3 \\
$n=1000$ & 11.8 & 3075.7 & 26629.7 \\
\midrule
\multicolumn{4}{c}{(a) True $u=2$} \\
\midrule
$n=100$  & 19.9 & 3350.1 & 12295.3 \\
$n=200$  & 14.9 & 3377.3 & 14010.1 \\
$n=500$  & 14.5 & 3431.0 & 19384.1 \\
$n=1000$ & 14.0 & 3526.1 & 26112.1 \\
\midrule
\multicolumn{4}{c}{(b) True $u=5$}\\
\bottomrule
\end{tabular}
\end{table}

\begin{figure}
    \centering
    \includegraphics[width=0.8\linewidth]{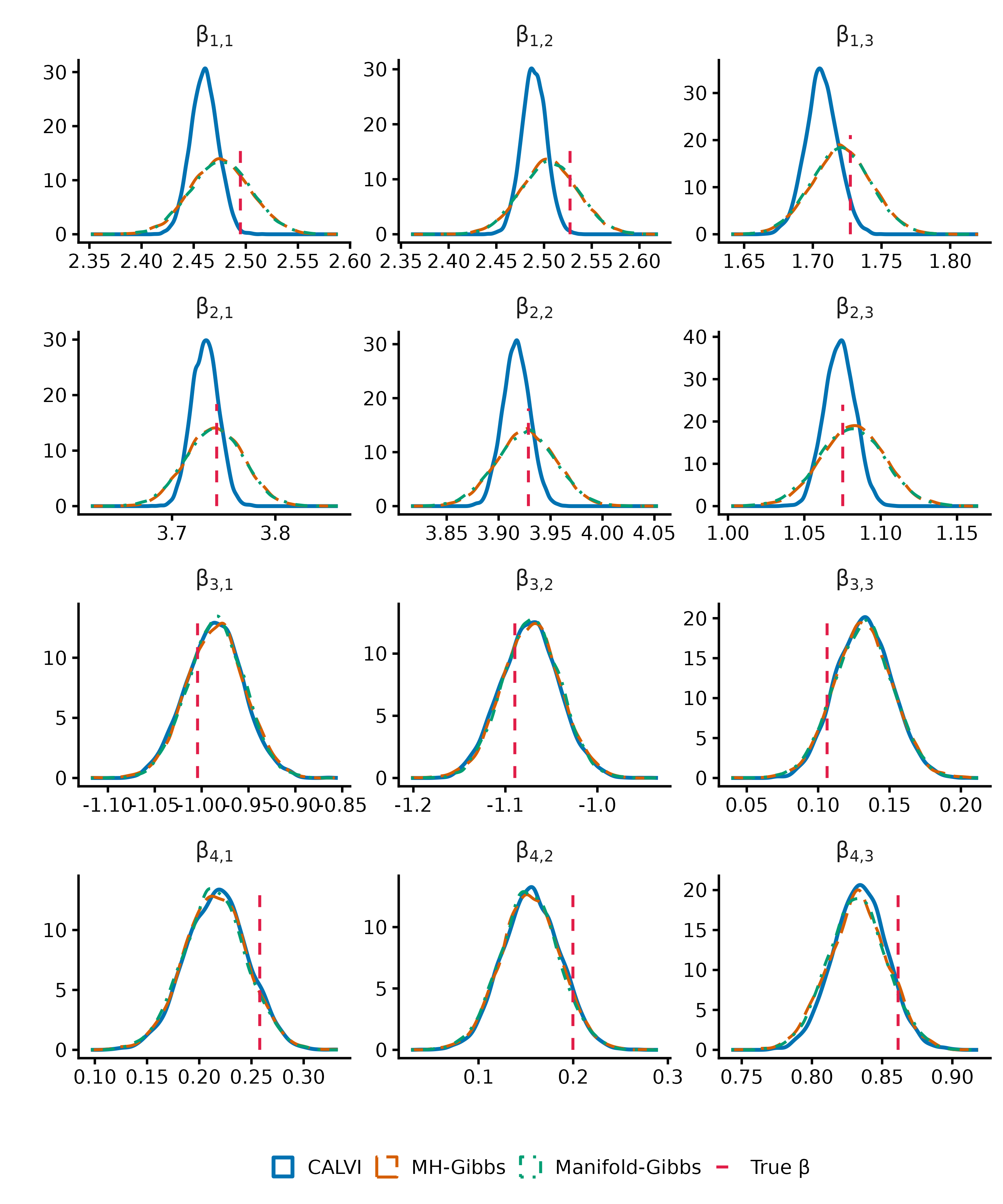}
    \caption{Marginal posterior distributions for $\beta_{ij}, \ i=1,\dots,3 \ , j=1,\dots,3$ based on the first replicate of the simulated data. Red dotted lines denote the true values of the parameters. The marginal densities for the remaining coefficients can be seen in Figure~\ref{beta all1} and Figure~\ref{beta all2}.}
    \label{fig:beta marginal}
\end{figure}

\subsection{Real data analysis} \label{real data}
We applied the response envelope model to women's breast tissue data. We tested whether socioeconomic stress influences metabolism, resulting in altered N-glycosylation associated with breast cancer risk in Black and White women. Normal breast tissue samples from 43 Black women (BW) and 43 White women (WW) at high 5-year risk of breast cancer based on the Gail score were obtained from The Susan G. Komen for the Cure Tissue Bank at the IU Simon Cancer Center. Socioeconomic data, including Gail score, age, body mass index (BMI), household income, education level (university and graduate school), and marital status, were collected. We excluded 13 BW and 13 WW due to missing information on these variables. To understand the N-glycan alterations that may contribute to breast cancer risk, 53 N-glycans were identified by mass spectrometry imaging methods. In addition, we used as responses the 10 (of 53) N-glycans that showed significant discrimination between BW and WW, as measured by the receiver operating characteristic area under the curve (AUC) (p < 0.0004). More detailed information about the data processing can be found in \citet{rujchanarong2022metabolic}.

We fit the response envelope model to the data $(\mathbf{Y}_i, \mathbf{X}_i)$ for $i = 1, \dots, 30$ in each group, where $\mathbf{Y}_i \in \mathbb{R}^{10}$ represents the mass of N‑glycans for the $i$th subject. The predictor vector $\mathbf{X}_i \in \mathbb{R}^{8}$ is partitioned as $\mathbf{X}_i = \begin{pmatrix} (\mathbf{X}_i^{(1)})^\top  & (\mathbf{X}_i^{(2)})^\top \end{pmatrix}^\top,$ where \(\mathbf{X}_i^{(1)} \in \{0,1\}^4\) comprises binary indicators for vocational or technical school status, associate degree status, single status and married status, and \(\mathbf{X}_i^{(2)} \in \mathbb{R}^4\) includes continuous variables: Gail score, age (years), BMI, and household income. All predictor variables are standardized and the responses are log-transformed.

Both CALVI and MH--Gibbs selected the envelope dimension as $u=1$. Table~\ref{tab:rmse-real} reports residual bootstrap root--mean--square errors (RMSEs) for the regression coefficients ($B=1000$). CALVI yields substantially smaller RMSEs than MH--Gibbs in both groups (0.0053 vs.\ 0.0373 for White women; 0.0073 vs.\ 0.0517 for Black women), indicating reduced sensitivity of the CALVI point estimates to resampling variability. The inflated bootstrap RMSE under MH--Gibbs may reflect the additional Monte Carlo variability induced by the likelihood-driven proposal in this small-sample setting. Although the Manifold--Gibbs sampler was included in the simulation study, it was not considered in the bootstrap analysis due to its substantially higher computational cost. In particular, the residual bootstrap requires refitting the model $B=1000$ times, rendering the Manifold--Gibbs approach computationally impractical in this setting. Figure~\ref{fig:heatmap} shows that several N-glycans are strongly associated with age for Black women, whereas associations in White women are generally weak. In terms of computation, we generated 50{,}000 MH--Gibbs samples, discarding the first 50\% as burn-in in order to ensure adequate stabilization of posterior summaries. 
This required 988.0 and 955.7 seconds for White and Black women, respectively. 
In contrast, CALVI converged in 1.5 and 7.7 seconds (with a maximum of 50{,}000 iterations). Reproducible code is available at \url{https://github.com/Seunghyeon-Kim-stat/env-LVI}.

\begin{table}[!t]
\centering
\caption{Bootstrap root--mean--square error (RMSE) of different estimation methods for White and Black women.}
\label{tab:rmse-real}
\small
\setlength{\tabcolsep}{6pt}
\begin{tabular}{lcc}
\toprule
Method & White women & Black women \\
\midrule
CALVI      & 0.0053 (0.0125) & 0.0073 (0.0057) \\
MH--Gibbs  & 0.0373 (0.0168) & 0.0517 (0.0241) \\
\bottomrule
\end{tabular}
\end{table}

\begin{figure}
    \centering
    \includegraphics[width=0.7\linewidth]{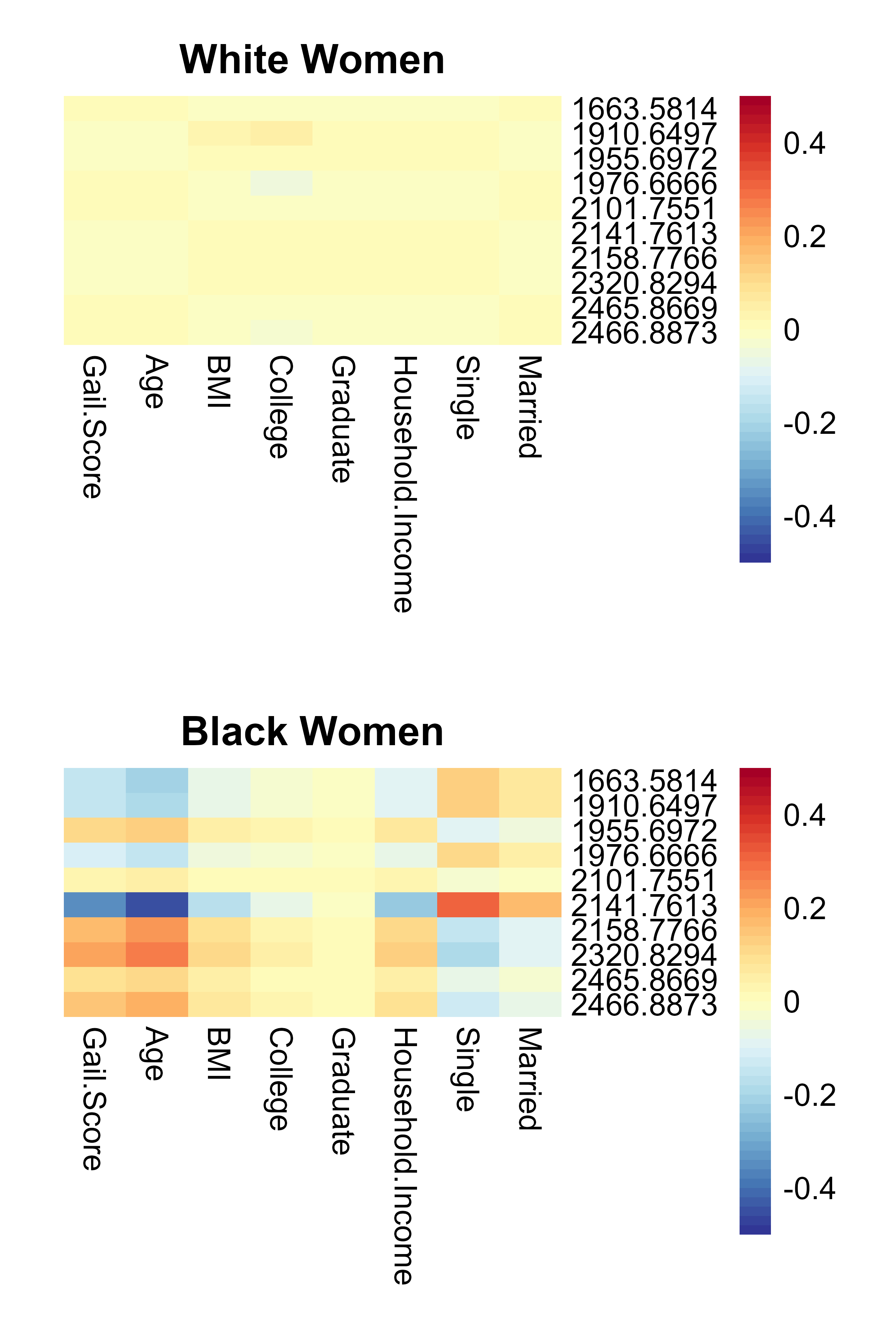}
    \caption{Heatmaps of the regression coefficient matrices for each group, estimated under the response envelope model using the CALVI algorithm. Rows correspond to N-glycans (responses) and columns correspond to predictors. The color bar indicates the magnitude and direction of the coefficient values.}
    \label{fig:heatmap}
\end{figure}

\section{Concluding remarks and discussion}\label{conclude}

We proposed a novel reparameterization and employed CALVI to improve computational efficiency in Bayesian response envelope models. To the best of our knowledge, this is the first attempt to bring the VI framework to envelope models in a way that preserves conjugacy for most parameters while handling the remaining nonconjugate block efficiently. Along with the methodological development, we provided a theoretical justification for the CALVI Laplace update used in the nonconjugate block, showing that it asymptotically coincides (after the usual $\sqrt n$ rescaling) with the corresponding exact CAVI update in TV distance.

Simulation studies and a real-data analysis demonstrated that CALVI can offer substantial computational gains while maintaining competitive inferential performance. Beyond the response envelope model considered here, the same reparameterization-and-update strategy can be adapted to other Bayesian models with envelope-type structure.

A general limitation of MFVF is that it can underestimate posterior uncertainty. In our setting, Figure~\ref{fig:beta marginal} illustrates that this underestimation is more pronounced for the top-$u$ rows of $\bbeta$. This behavior appears to reflect the interaction between the mean-field factorization, the parameterization of $\bGamma$ through $\bA$, and the proposed reparameterization. Mitigating this effect may require a richer variational family (e.g., structured covariance, blockwise dependence, or other extensions beyond MFVF), or alternative parameterizations that reduce the strength of posterior dependence; we leave these directions for future work.

\clearpage
\appendix
\section*{Supplementary Material for ``Laplace Variational Inference for Bayesian Envelope Models''}

\renewcommand{\thesection}{S\arabic{section}}
\renewcommand{\theequation}{S.\arabic{equation}}
\renewcommand{\thefigure}{S\arabic{figure}}
\renewcommand{\thetable}{S\arabic{table}}

\setcounter{section}{0}
\setcounter{equation}{0}
\setcounter{figure}{0}
\setcounter{table}{0}

The supplementary material is organized as follows.
Section~\ref{sec:A} provides detailed derivations of the computational bottleneck identified in the main paper.
Section~\ref{sec:B} presents implementation details of the coordinate-ascent Laplace variational inference (CALVI) algorithm.
Section~\ref{sec:C} discusses the implications of CALVI for the Bayesian predictor envelope model.
Sections~\ref{sec:D} and~\ref{sec:E} prove Theorems~\ref{main-theorem} and~\ref{thm:yenv}, respectively.
Finally, Section~\ref{sec:F} reports additional simulation results for both the Bayesian response and predictor envelope models.

\subsubsection*{Notations}
Throughout, $\|\cdot\|_2$ denotes the Euclidean norm, $\|\cdot\|_F$ the Frobenius norm, and $\|\cdot\|_{op}$ the spectral norm.
The operator $\vecr(\cdot)$ denotes column-wise vectorization.
For $a,b\in\mathbb{N}$, $\mathcal{K}_{a,b}$ denotes the $(ab)\times(ab)$ commutation matrix satisfying
$\vecr(\bA^\top)=\mathcal{K}_{a,b}\vecr(\bA)$ for any $\bA\in\mathbb{R}^{a\times b}$.
We write $\mathcal{MN}(\cdot)$ for the matrix normal distribution and $\mathcal{IW}(\cdot)$ for the inverse--Wishart distribution.
Unless otherwise stated, we suppress the iteration index $(t)$ for readability. $\calC_d$ denotes the set of $d \times d$ symmetric positive definite matrices.

\section{Computational issues}
\label{sec:A}

In this section, we provide the explicit form of the $\bA$-coordinate objective
\eqref{eq:11} together with derivative expressions that make the dominant
computational bottleneck explicit (Section~\ref{sec4.1}).

\subsection{Local objective for the \texorpdfstring{$\bA$}{A}-update}
\label{sec:A.1}

Let the $\bA$-coordinate objective at iteration $t$ be
\begin{align*}
f^{(t)}(\bA)
:=\;& \bbE_{q^{(t)}_{-\bA}} \big[ \log p(\bA \mid \bmu, \boldsymbol{\eta}, \bOmega, \bOmega_0, \bbY) \big] \\
=\;&
\bbE_{q^{(t)}_{-\bA}} \Bigg[-\frac{1}{2}\,
\tr\!\Big[
\underbrace{\bGamma(\bA)^\top
\Big(
\bbY_{\bmu}^\top \bbY_{\bmu}
+ \bB_0 \boldsymbol{M} \bB_0^\top
\Big)
\bGamma(\bA)\,
\bOmega^{-1}}_{(i)}
\Big] \;-\frac{1}{2}\,
\tr\!\Big[
\underbrace{\bGamma_0(\bA)^\top
\bbY_{\bmu}^\top \bbY_{\bmu}
\bGamma_0(\bA)\,
\bOmega_0^{-1}}_{(ii)}
\Big]
\\[6pt]
&\quad \;+\tr\!\Big[
\bGamma(\bA)^\top
\Big(
\bbY_{\bmu}^\top \bbX \boldsymbol{\eta}^\top
+ \bB_0 \boldsymbol{M} \boldsymbol{\eta}^\top
\Big)
\bOmega^{-1}
\Big] \;-\frac{1}{2}\,
\tr\!\Big[
\boldsymbol{V}_0^{(A)-1}
\bA^\top
\boldsymbol{U}_0^{(A)-1}
\bA
\Big]
\\[6pt]
&\quad \;+\tr\!\Big[
\boldsymbol{V}_0^{(A)-1}
\bA_0^\top
\boldsymbol{U}_0^{(A)-1}
\bA 
\Big] \Bigg]
\;+\; \text{\emph{const.}}, \tag{S1} \label{eq:fA}
\end{align*}
where $\bbY_{\bmu}=\bbY-\boldsymbol{1}_n\bmu^\top$ and $\boldsymbol{1}_n$ is the $n$-vector of ones.
The objective $f^{(t)}(\bA)$ is obtained from the unnormalized log-posterior
$\log p(\bmu,\boldsymbol{\eta},\bOmega,\bOmega_0,\bA\mid\bbY)$ in \eqref{eq:7}
by collecting all $\bA$-dependent terms and taking expectation with respect to
$q^{(t)}(\bmu,\boldsymbol{\eta},\bOmega,\bOmega_0)$.

\subsection{Justification for exchanging differentiation and expectation}
\label{sec:A.2}

The dominant bottleneck is algebraic and already appears at the complete-data level.
Proposition~\ref{prop:DCT} justifies interchanging differentiation and expectation so that
the same structure carries over to $f^{(t)}(\bA)$.

\begin{proposition}[Exchange of differentiation and expectation (first order)]
\label{prop:DCT}
Let $\bZ:=(\bmu,\boldsymbol\eta,\bOmega,\bOmega_0)$ and consider the complete-data log-posterior
$\log p(\bZ,\bA\mid \bbY_n)$ in \eqref{eq:7}. Fix a compact neighborhood
$\mathcal N := \{ \bA: \| \bA - \bA^{\star}\|_{F} \le \rho \}$.
Then there exists a measurable function $h_1(\bZ)$, independent of $\bA$, such that for all $\bA\in\mathcal N$,
\[
\big\|\nabla_{\vecr(\bA)} \log p(\bZ,\bA\mid \bbY_n) \big\|_{op}\ \le\ h_1(\bZ),
\qquad
\mathbb{E}_{q(\bZ)}[\,h_1(\bZ)\,]<\infty,
\]
for the variational factors $q(\bZ)=q(\bmu)\,q(\boldsymbol\eta)\,q(\bOmega)\,q(\bOmega_0)$ used in CALVI.
Consequently, for every $\bA\in\mathcal N$,
\[
\nabla_{\vecr(\bA)} \int q(\bZ)\log p(\bZ,\bA\mid \bbY_n)\,d\bZ
\;=\; \int q(\bZ)\,\nabla_{\vecr(\bA)} \log p(\bZ,\bA\mid \bbY_n)\,d\bZ.
\]
\end{proposition}
The proof is given in Section~\ref{sec:A.5}.

\subsubsection*{Common notation for inverse square-root factors}
\label{sec:common_invsqrt_notation}

Throughout Section~\ref{sec:A}, the bottleneck originates from differentiating factors of the form
$(\bU^\top\bU)^{-1/2}$. To avoid repeated expansions, we use the following notation.

Let $\bU\in\mathbb{R}^{r\times k}$ have full column rank and define
\[
\bJ(\bU):=\bU^\top\bU,\qquad
\bF(\bU):=\bJ(\bU)^{-1/2},\qquad
\bP(\bU):=\bJ(\bU)^{-1},
\]
together with the Kronecker--sum matrix
\[
\bQ(\bU):=\bF(\bU)\otimes \bI_k + \bI_k\otimes \bF(\bU)\in\mathbb{R}^{k^2\times k^2},
\]
and
\[
\bR(\bU)
:= \bP(\bU)\otimes \bP(\bU)\bU^\top
\;+\;
(\bP(\bU)\bU^\top\otimes \bP(\bU))\,\mathcal{K}_{r,k}
\in \mathbb{R}^{k^2\times (rk)}.
\]

In the response envelope model, this notation is used twice:
(i) $\bU=\bC_{\bA}$ with $k=u$ and $\bJ(\bC_{\bA})=\bC_{\bA}^\top\bC_{\bA}$,
(ii) $\bU=\bD_{\bA}$ with $k=r-u$ and $\bJ(\bD_{\bA})=\bD_{\bA}^\top\bD_{\bA}$.

\subsection{First-order derivative and the computational bottleneck}
\label{sec:A.3}

Through Proposition~\ref{prop:DCT}, the dominant bottleneck in $\nabla_{\vecr(\bA)} f^{(t)}(\bA)$ arises from
the envelope terms \eqref{eq:fA}(i)--(ii), which depend on $\bA$ through the orthonormal factors
\[
\bGamma(\bA)=\bC_{\bA}\big(\bC_{\bA}^\top\bC_{\bA}\big)^{-1/2},
\qquad
\bGamma_0(\bA)=\bD_{\bA}\big(\bD_{\bA}^\top\bD_{\bA}\big)^{-1/2}.
\]
Accordingly, the derivative takes the generic chain-rule form
\begin{align*}
    \frac{\partial}{\partial \vecr(\bA)^\top}
    \tr\!\left[
        \bS\, \bGamma(\bA)\, \bM\, \bGamma^\top(\bA)
    \right]
    &=
    \frac{\partial}{\partial \vecr\!\left(\bGamma(\bA)\right)^\top}
    \tr\!\left[
        \bS\, \bGamma(\bA)\, \bM\, \bGamma^\top(\bA)
    \right]
    \\
    &\quad\times
    \frac{\partial \vecr\!\left(
        \bC_{\bA}(\bC_{\bA}^\top \bC_{\bA})^{-1/2}
    \right)}{\partial \vecr(\bC_{\bA})^\top}
    \frac{\partial \vecr(\bC_{\bA})}{\partial \vecr(\bA)^\top},
    \tag{S2}\label{eq:dGamma-generic}
    \\[6pt]
    \frac{\partial}{\partial \vecr(\bA)^\top}
    \tr\!\left[
        \bS_0\, \bGamma_0(\bA)\, \bM_0\, \bGamma_0^\top(\bA)
    \right]
    &=
    \frac{\partial}{\partial \vecr\!\left(\bGamma_0(\bA)\right)^\top}
    \tr\!\left[
        \bS_0\, \bGamma_0(\bA)\, \bM_0\, \bGamma_0^\top(\bA)
    \right]
    \\
    &\quad\times
    \frac{\partial \vecr\!\left(
        \bD_{\bA}(\bD_{\bA}^\top \bD_{\bA})^{-1/2}
    \right)}{\partial \vecr(\bD_{\bA})^\top}
    \frac{\partial \vecr(\bD_{\bA})}{\partial \vecr(\bA)^\top}.
    \tag{S3}\label{eq:dGamma0-generic}
\end{align*}
The first factor in \eqref{eq:dGamma-generic} is explicit:
\[
\frac{\partial \tr\!\left[\bS\,\bGamma(\bA)\,\bM\,\bGamma^\top(\bA)\right]}
{\partial \vecr\!\left(\bGamma(\bA)\right)^\top}
=
2\,\bS\,\bGamma(\bA)\,\bM,
\]
(and analogously for the $\bGamma_0$ term).
The bottleneck is therefore driven by the Jacobians of the inverse square-root factors
$(\bC_{\bA}^\top\bC_{\bA})^{-1/2}$ and $(\bD_{\bA}^\top\bD_{\bA})^{-1/2}$.

\begin{lemma}[Compact Jacobian for the inverse square-root factor]
\label{lem:invsqrt_first_compact}
Adopt the common notation in Section~\ref{sec:common_invsqrt_notation}.
For any full-column-rank $\bU\in\mathbb{R}^{r\times k}$,
\[
\frac{\partial\,\vecr\!\big(\bJ(\bU)^{-1/2}\big)}{\partial\,\vecr(\bU)^\top}
=
-\,\bQ(\bU)^{-1}\,\bR(\bU).
\]
\end{lemma}

\noindent
Applying Lemma~\ref{lem:invsqrt_first_compact} with $\bU=\bC_{\bA}$ (so $k=u$) yields the expanded form
\begin{align*}
    \frac{\partial \, \vecr \left( (\boldsymbol{C}_{\bA}^\top \boldsymbol{C}_{\bA})^{-1/2} \right)}{\partial \, \vecr(\boldsymbol{C}_{\bA})}
    =&
    - \Big[
    (\boldsymbol{C}_{\bA}^\top \boldsymbol{C}_{\bA})^{-1/2} \otimes \boldsymbol{I}_{u}
    + \boldsymbol{I}_{u} \otimes (\boldsymbol{C}_{\bA}^\top \boldsymbol{C}_{\bA})^{-1/2}
    \Big]^{-1}
\\
    &\times \Big[
    (\boldsymbol{C}_{\bA}^\top \boldsymbol{C}_{\bA})^{-1} \otimes (\boldsymbol{C}_{\bA}^\top \boldsymbol{C}_{\bA})^{-1} \boldsymbol{C}_{\bA}^\top
\\
    &\qquad\qquad
    +(\boldsymbol{C}_{\bA}^\top \boldsymbol{C}_{\bA})^{-1} \boldsymbol{C}_{\bA}^\top \otimes (\boldsymbol{C}_{\bA}^\top \boldsymbol{C}_{\bA})^{-1} \mathcal{K}_{r,u}
    \Big],
    \tag{S4}\label{eq:bottleneck}
\end{align*}
and applying it with $\bU=\bD_{\bA}$ (so $k=r-u$) yields
\begin{align*}
    \frac{\partial \, \vecr \left( (\boldsymbol{D}_{\bA}^\top \boldsymbol{D}_{\bA})^{-1/2} \right)}{\partial \, \vecr(\boldsymbol{D}_{\bA})}
    =&
    - \Big[
    (\boldsymbol{D}_{\bA}^\top \boldsymbol{D}_{\bA})^{-1/2} \otimes \boldsymbol{I}_{r-u}
    + \boldsymbol{I}_{r-u} \otimes (\boldsymbol{D}_{\bA}^\top \boldsymbol{D}_{\bA})^{-1/2}
    \Big]^{-1}
\\
    &\times \Big[
    (\boldsymbol{D}_{\bA}^\top \boldsymbol{D}_{\bA})^{-1} \otimes (\boldsymbol{D}_{\bA}^\top \boldsymbol{D}_{\bA})^{-1} \boldsymbol{D}_{\bA}^\top
\\
    &\qquad\qquad
    +(\boldsymbol{D}_{\bA}^\top \boldsymbol{D}_{\bA})^{-1} \boldsymbol{D}_{\bA}^\top \otimes (\boldsymbol{D}_{\bA}^\top \boldsymbol{D}_{\bA})^{-1} \mathcal{K}_{r,(r-u)}
    \Big],
    \tag{S5}\label{eq:bottleneck0}
\end{align*}
which matches the bottleneck expression \eqref{eq:12} in the main paper.

Finally, the remaining linear Jacobians are
\[
\frac{\partial \vecr\left(\boldsymbol{C}_{\bA}\right)}{\partial \vecr\left(\bA\right)^\top}
=
\bI_u \otimes 
\begin{pmatrix}
\boldsymbol{0}_{u \times (r-u)} \\
\boldsymbol{I}_{r-u}
\end{pmatrix},
\qquad
\frac{\partial \vecr\left(\boldsymbol{D}_{\bA}\right)}{\partial \vecr\left(\bA\right)^\top}
=
\mathcal{K}_{(r-u),r} 
\begin{pmatrix}
-\boldsymbol{I}_{u(r-u)} \\
\boldsymbol{0}_{(r-u)^2 \times u(r-u)}
\end{pmatrix}.
\]

\paragraph{Complexity of the first-order derivative.}
The dominant cost stems from solving linear systems involving the Kronecker--sum matrices
$\bQ(\bC_{\bA})\in\mathbb{R}^{u^2\times u^2}$ and $\bQ(\bD_{\bA})\in\mathbb{R}^{(r-u)^2\times (r-u)^2}$
appearing in \eqref{eq:bottleneck}--\eqref{eq:bottleneck0}.
Using generic fast matrix inversion with exponent $\omega\approx 2.376$ \citep{10.1145/28395.28396}, these solves scale as $\mathcal{O}(u^{2\omega})=\mathcal{O}(u^{4.752})$ and
$\mathcal{O}((r-u)^{2\omega})=\mathcal{O}((r-u)^{4.752})$, respectively.
Therefore, the overall complexity of evaluating $\nabla_{\vecr(\bA)} f^{(t)}(\bA)$ is
\[
\mathcal{O}\!\big(\max\{u,r-u\}^{4.752}\big),
\]
up to lower-order matrix multiplications.

\subsection{Second-order derivative and cost of assembling the full Hessian}
\label{sec:A.4}

Differentiating \eqref{eq:bottleneck}--\eqref{eq:bottleneck0} once more is feasible and yields an explicit
second-order derivative. The resulting expressions remain tractable in closed form, but involve additional
Kronecker--sum solves and Kronecker products, which makes the full Hessian substantially more expensive
to assemble than the gradient.

\begin{lemma}[Directional second derivative of the inverse square-root Jacobian]
\label{lem:invsqrt_second}
Adopt the common notation in Section~\ref{sec:common_invsqrt_notation}.
Let $\bU\in\mathbb{R}^{r\times k}$ be full column rank and let $\Delta\bU$ be a perturbation.
Define
\[
\Delta\bJ(\bU):=\bU^\top\Delta\bU+(\Delta\bU)^\top\bU,
\qquad
\Delta\bP(\bU):=-\bP(\bU)\,\Delta\bJ(\bU)\,\bP(\bU).
\]
Let $\Delta\bF(\bU)$ be defined via its vectorization
\[
\vecr(\Delta\bF(\bU))
=
-\,\bQ(\bU)^{-1}\,\vecr\!\big(\bP(\bU)\,\Delta\bJ(\bU)\,\bP(\bU)\big).
\]
Then the directional derivative of the Jacobian in Lemma~\ref{lem:invsqrt_first_compact} satisfies
\[
\Delta\!\left(
\frac{\partial\,\vecr\!\big(\bJ(\bU)^{-1/2}\big)}{\partial\,\vecr(\bU)^\top}
\right)
=
\bQ(\bU)^{-1}\Big\{\Delta\bQ(\bU)\,\bQ(\bU)^{-1}\bR(\bU)-\Delta\bR(\bU)\Big\},
\]
where
\[
\Delta\bQ(\bU)
:=
\Delta\bF(\bU)\otimes \bI_k + \bI_k\otimes \Delta\bF(\bU),
\]
and
\begin{align*}
\Delta\bR(\bU)
&=
(\Delta\bP(\bU))\otimes (\bP(\bU)\bU^\top)
+\bP(\bU)\otimes(\Delta\bP(\bU)\,\bU^\top+\bP(\bU)\,\Delta\bU^\top) \\
&\quad
+\big(\Delta\bP(\bU)\,\bU^\top+\bP(\bU)\,\Delta\bU^\top\big)\otimes \bP(\bU)\,\mathcal{K}_{r,k}
+(\bP(\bU)\bU^\top)\otimes(\Delta\bP(\bU))\,\mathcal{K}_{r,k}.
\end{align*}
\end{lemma}

\paragraph{Implication for assembling $\nabla^2_{\vecr(\bA)} f^{(t)}(\bA)$.}
Let $d:=u(r-u)$ be the dimension of $\vecr(\bA)$.
In the Euclidean envelope parameterization, assembling the full Hessian matrix
$\nabla^2_{\vecr(\bA)} f^{(t)}(\bA)\in\mathbb{R}^{d\times d}$ requires applying
Lemma~\ref{lem:invsqrt_second} to $\bU=\bC_{\bA}$ (with $k=u$) and $\bU=\bD_{\bA}$ (with $k=r-u$),
combined with the linear Jacobians $\partial\vecr(\bC_{\bA})/\partial\vecr(\bA)^\top$ and
$\partial\vecr(\bD_{\bA})/\partial\vecr(\bA)^\top$.

Each Hessian column (corresponding to a coordinate direction in $\vecr(\bA)$) requires a constant number of
solves with $\bQ(\bC_{\bA})$ and $\bQ(\bD_{\bA})$. Using generic fast inversion with exponent
$\omega\approx 2.376$, the resulting cost of assembling the full Hessian is
\[
\mathcal{O}\!\left(
d\cdot\big(u^{2\omega}+(r-u)^{2\omega}\big)
\right)
=
\mathcal{O}\!\left(
u(r-u)\cdot\big(u^{4.752}+(r-u)^{4.752}\big)
\right),
\]
up to lower-order matrix multiplications and bookkeeping. In the balanced regime $u\asymp r-u$,
this scales as $\mathcal{O}(u^{6.752})$.

\paragraph{Remark (link to the reparameterized objective).}
The above shows that the \citet{cook2016note}'s parameterization admits explicit first- and second-order derivatives,
but the Hessian assembly involves repeated Kronecker--sum solves of dimension $u^2$ and $(r-u)^2$.
This motivates the reparameterized objective \eqref{eq:16}, whose Hessian admits an explicit
closed form (see \eqref{eq:tAobj2}) and only requires inverting a $(r-u)u\times (r-u)u$ matrix to form
the Laplace covariance, which is substantially less burdensome in the dimensions considered here.

\subsection{Proof of Proposition~\ref{prop:DCT}}
\label{sec:A.5}

\begin{proof}[Proof of Proposition~\ref{prop:DCT}]
In the complete-data log-posterior \eqref{eq:7}, the parameter $\bA$ appears only through
$\bGamma(\bA)$, $\bGamma_0(\bA)$, and the prior term; all other parts are $\bA$--free.

\paragraph{Step 1: Uniform bounds on $\mathcal N$ and on $\nabla_{\vecr(\bA)}\bGamma,\ \nabla_{\vecr(\bA)}\bGamma_0$.}
As in the local analysis, take $\mathcal N := \{ \bA: \| \bA - \bA^{\star}\|_{F} \le \rho \}$ and define
\[
B_{\mathcal N}:=\|\bA^\star\|_F+\rho,
\qquad\text{so that}\qquad
\sup_{\bA\in\mathcal N}\|\bA\|_F\le B_{\mathcal N}.
\]
Recall that $\bC_{\bA}^\top\bC_{\bA}=\bI_u+\bA^\top\bA\succeq \bI_u$ (and analogously for $\bD_{\bA}$),
hence $\lambda_{\min}(\bC_{\bA}^\top\bC_{\bA})\ge 1$ for all $\bA$.
Restricting $\bA$ to $\mathcal N$ yields
\[
\lambda_{\max}(\bC_{\bA}^\top\bC_{\bA})
\le 1+\|\bA\|_{op}^2
\le 1+\|\bA\|_F^2
\le 1+B_{\mathcal N}^2.
\]
Therefore $(\bC_{\bA}^\top\bC_{\bA})^{-1}$ and $(\bC_{\bA}^\top\bC_{\bA})^{-1/2}$ are uniformly bounded on
$\mathcal N$ (and likewise with $\bD_{\bA}$ in place of $\bC_{\bA}$).
Using the parameterization~\eqref{eq:3} together with Lemma~\ref{lem:invsqrt_first_compact} (and standard matrix-calculus rules), we obtain a finite constant $K_1=K_1(B_{\mathcal N})<\infty$ such that
\[
\sup_{\bA\in\mathcal N}\|\nabla_{\vecr(\bA)}\bGamma(\bA)\|_{op} \le K_1,
\qquad
\sup_{\bA\in\mathcal N}\|\nabla_{\vecr(\bA)}\bGamma_0(\bA)\|_{op} \le K_1 .
\]

\paragraph{Step 2: Construct an integrable dominating function for the gradient.}
\emph{(a) $\bA$--uniform gradient bound.}
Write $\log p(\bZ,\bA\mid\bbY_n)$ as in \eqref{eq:7}.
Only the terms involving $\bGamma(\bA)$, $\bGamma_0(\bA)$ and the prior contribute to
$\nabla_{\vecr(\bA)}$.

Consider first the term
\[
T_1(\bA):=-\frac12\tr\!\Big[(\bbY_{\bmu}\bGamma(\bA)-\bbX\boldsymbol\eta^\top)\,\bOmega^{-1}\,
(\bbY_{\bmu}\bGamma(\bA)-\bbX\boldsymbol\eta^\top)^\top\Big].
\]
Let $\bE(\bA):=\bbY_{\bmu}\bGamma(\bA)-\bbX\boldsymbol\eta^\top$.
By standard matrix calculus, the derivative of $T_1$ with respect to $\bGamma$ satisfies
\[
\nabla_{\bGamma}T_1(\bA)= -\,\bbY_{\bmu}^\top \bE(\bA)\,\bOmega^{-1}.
\]
Using the chain rule and the operator bound on $\nabla_{\vecr(\bA)}\bGamma(\bA)$,
\[
\|\nabla_{\vecr(\bA)}T_1(\bA)\|_2
\le \|\nabla_{\vecr(\bA)}\bGamma(\bA)\|_{op}\,\|\nabla_{\bGamma}T_1(\bA)\|_F
\le K_1\,\|\bbY_{\bmu}^\top \bE(\bA)\bOmega^{-1}\|_F.
\]
Moreover,
\[
\|\bbY_{\bmu}^\top \bE(\bA)\bOmega^{-1}\|_F
\le \|\bbY_{\bmu}\|_F\,\|\bE(\bA)\|_F\,\|\bOmega^{-1}\|_{op},
\qquad
\|\bE(\bA)\|_F\le \|\bbY_{\bmu}\|_F + \|\bbX\|_{op}\,\|\boldsymbol\eta\|_F,
\]
(where we used $\|\bGamma(\bA)\|_{op}=1$).
Hence, for a universal constant $C$,
\[
\sup_{\bA\in\mathcal N}\|\nabla_{\vecr(\bA)}T_1(\bA)\|_2
\le C K_1\,\|\bOmega^{-1}\|_{op}\Big(\|\bbY_{\bmu}\|_F^2+\|\bbX\|_{op}^2\|\boldsymbol\eta\|_F^2\Big).
\]

Next consider the $\bGamma_0$--term
\[
T_2(\bA):=-\frac12\tr\!\Big[ \bGamma_0^\top(\bA)\,\bbY_{\bmu}^\top\bbY_{\bmu}\,\bGamma_0(\bA)\,\bOmega_0^{-1}\Big].
\]
Again by the chain rule and $\sup_{\bA\in\mathcal N}\|\nabla_{\vecr(\bA)}\bGamma_0(\bA)\|_{op}\le K_1$,
\[
\sup_{\bA\in\mathcal N}\|\nabla_{\vecr(\bA)}T_2(\bA)\|_2
\le C K_1\,\|\bOmega_0^{-1}\|_{op}\,\|\bbY_{\bmu}\|_F^2,
\]
for a universal constant $C$.

All remaining $\bA$--dependent likelihood terms in \eqref{eq:7} enter through $\bGamma(\bA)$ in an
affine/quadratic way and can be bounded similarly by $CK_1\|\bOmega^{-1}\|_{op}(1+\|\boldsymbol\eta\|_F^2)$,
and the prior gradient is $\bZ$--independent and uniformly bounded on $\mathcal N$.
Absorbing all fixed constants into a single constant, we obtain the domination
\[
\sup_{\bA\in\mathcal N}\big\|\nabla_{\vecr(\bA)}\log p(\bZ,\bA\mid \bbY_n)\big\|_{2}
\ \le\ h_1(\bZ),
\]
where one may take
\[
h_1(\bZ)
:= C_1(B_{\mathcal N})\Big\{
\|\bOmega^{-1}\|_{op}\big(\|\bbY_{\bmu}\|_F^2+\|\bbX\|_{op}^2\|\boldsymbol\eta\|_F^2+1\big)
+ \|\bOmega_0^{-1}\|_{op}\,\|\bbY_{\bmu}\|_F^2 + 1\Big\},
\]
for a finite constant $C_1(B_{\mathcal N})$ depending only on $\mathcal N$ through $B_{\mathcal N}$ and fixed
hyperparameters.

\emph{(b) Integrability under $q(\bZ)$.}
Because $q(\bmu),q(\boldsymbol\eta)$ are Gaussian, the moments
$\E_q\|\bbY_{\bmu}\|_F^2<\infty$ and $\E_q\|\boldsymbol\eta\|_F^2<\infty$ are finite.
Because $q(\bOmega),q(\bOmega_0)$ are inverse-Wishart, we have
$\E_q\tr(\bOmega^{-1})<\infty$ and $\E_q\tr(\bOmega_0^{-1})<\infty$, hence
$\E_q\|\bOmega^{-1}\|_{op}<\infty$ and $\E_q\|\bOmega_0^{-1}\|_{op}<\infty$.
Therefore $\E_q[h_1(\bZ)]<\infty$.

\paragraph{Step 3: Dominated convergence in a coordinate direction.}
Fix $\bA\in\mathcal N$ and consider the $j$th coordinate direction $e_j$ of $\vecr(\bA)$.
For $t$ small enough so that $\bA_t:=\bA+t\,\vecr^{-1}(e_j)\in\mathcal N$, define
\[
\Delta_t(\bZ):=\frac{\log p(\bZ,\bA_t\mid \bbY_n)-\log p(\bZ,\bA\mid \bbY_n)}{t}.
\]
By the mean value theorem, for each $\bZ$ there exists $\btA=\btA(\bZ,t)$ on the line segment between
$\bA$ and $\bA_t$ such that
\[
\Delta_t(\bZ)=\big\langle \nabla_{\vecr(\bA)}\log p(\bZ,\btA\mid \bbY_n),\, e_j\big\rangle,
\]
hence $|\Delta_t(\bZ)|\le \|\nabla_{\vecr(\bA)}\log p(\bZ,\btA\mid \bbY_n)\|_{2}\le h_1(\bZ)$.
Moreover, $\Delta_t(\bZ)\to \partial_{\vecr(\bA)_j}\log p(\bZ,\bA\mid \bbY_n)$ pointwise as $t\to0$.
Since $h_1$ is $q$--integrable, dominated convergence yields
\[
\partial_{\vecr(\bA)_j}\int q(\bZ)\log p(\bZ,\bA\mid \bbY_n)\,d\bZ
=
\int q(\bZ)\,\partial_{\vecr(\bA)_j}\log p(\bZ,\bA\mid \bbY_n)\,d\bZ.
\]
Doing this for all coordinates gives the vector identity
\[
\nabla_{\vecr(\bA)} \int q(\bZ)\log p(\bZ,\bA\mid \bbY_n)\,d\bZ
=
\int q(\bZ)\,\nabla_{\vecr(\bA)} \log p(\bZ,\bA\mid \bbY_n)\,d\bZ,
\]
which completes the proof.
\end{proof}

\section{Details of CALVI algorithm}
\label{sec:B}
This section collects the explicit expressions required to implement the reparameterized CALVI procedure in Algorithm~\ref{alg1}, with particular emphasis on the Laplace update for the nonconjugate envelope block $\bA$ and the closed-form CAVI updates for the conjugate blocks.

In Section~\ref{sec:B.1}, we provide a moment-expanded expression of the reparameterized $\bA$-coordinate objective $\tf^{(t)}(\bA)$ in~\eqref{eq:16} (up to an additive constant independent of $\bA$), written in terms of the centered design $\bbX_c$ and variational parameters under $q(\btmu)\,q(\boldsymbol{\tilde\eta})\,q(\btOmega)\,q(\btOmega_0)$.

In Section~\ref{sec:B.2}, we derive the coefficient matrices and scalars appearing in the gradient and Hessian representations~\eqref{eq:17}--\eqref{eq:18} (i.e., $\kappa$, $\boldsymbol{C}_1^{(t)}$, $\boldsymbol{C}_2^{(t)}$, and $\boldsymbol{C}_3^{(t)}$), and we record the resulting closed-form coordinate updates for the conjugate variational factors. Since the Laplace factor $q(\vecr(\bA))$ is Gaussian, these updates repeatedly involve Gaussian expectations of trace-quadratic forms in $\bC_{\bA}$ and $\bD_{\bA}$; to streamline the presentation, we introduce a Kronecker--trace operator and related identities used throughout the derivations.

Finally, Section~\ref{sec:B.3} provides the complete expression (including constants) of the approximate ELBO $\tilde{\mathcal L}(q)$ in~\eqref{eq:21} used for convergence monitoring in Algorithm~\ref{alg1}.

\subsection{Explicit form of the  \texorpdfstring{$\tilde f(\bA)$}{f~tilde(A)}.}
\label{sec:B.1}
We provide an explicit form of the local objective
$\tilde f^{(t)}(\bA)$ in~\eqref{eq:16} under the Bayesian response
envelope model with reparameterization~\eqref{eq:13}--\eqref{eq:14} and the MFVF scheme~\eqref{eq:15}.
Recall that, at iteration $t$, the MFVF factorizes as
\[
q(\btmu,\tilde{\boldsymbol{\eta}},\btOmega,\btOmega_0,\bA)
=
q(\btmu)\,q(\tilde{\boldsymbol{\eta}})\,q(\btOmega)\,q(\btOmega_0)\,q(\bA),
\]
with the conjugate variational factors 
\begin{align*}
    & q(\btmu) \sim \calN_{r}(\bmu_{q}, \bSigma_{q}), \nonumber\\
    & q(\boldsymbol{\tilde{\eta}}) \sim \mathcal{MN}_{u,p}(\boldsymbol{\eta}_{q}, \bU^{(\boldsymbol{\eta})}_{q}, \bV^{(\boldsymbol{\eta})}_{q}), \nonumber\\
    & q(\btOmega) \sim \mathcal{IW}_{u}(\bPsi^{(1)}_{q}, \nu^{(1)}_{q}), \nonumber\\
    & q(\btOmega_0) \sim \mathcal{IW}_{r-u}(\bPsi^{(0)}_{q}, \nu^{(0)}_{q}). \tag{S6} \label{eq:VF}
\end{align*}
The $\bA$--coordinate objective is defined by
\[
\tilde f(\bA)
:= \E_{ q(\btmu)\,q(\tilde{\boldsymbol{\eta}})\,q(\btOmega)\,q(\btOmega_0)}\!\left[\log p\!\left(\bA \mid \btmu,\tilde{\boldsymbol{\eta}},\btOmega,\btOmega_0,\bbY\right)\right].
\]
Since the normalizing constant in $p(\bA \mid \btmu,\tilde{\boldsymbol{\eta}},\btOmega,\btOmega_0,\bbY)$
does not depend on $\bA$, $\tilde f^{(t)}(\bA)$ can be written equivalently (up to an additive constant
independent of $\bA$) using the expected log-joint density.

We use the linear representations
\[
\bC_{\bA} = \bK + \bL\bA \in \mathbb{R}^{r\times u},
\qquad
\bD_{\bA} = \bL - \bK\bA^\top \in \mathbb{R}^{r\times (r-u)},
\]
where the fixed selection matrices $\bK\in\mathbb{R}^{r\times u}$ and $\bL\in\mathbb{R}^{r\times(r-u)}$ are
\[
\bK :=
\begin{pmatrix}
\boldsymbol{I}_u\\
\boldsymbol{0}_{(r-u)\times u}
\end{pmatrix},
\qquad
\bL :=
\begin{pmatrix}
\boldsymbol{0}_{u\times (r-u)}\\
\boldsymbol{I}_{r-u}
\end{pmatrix}.
\]
Then, using the variational factors implied by~\eqref{eq:VF} as follows:
\begin{align*}
    \E_{q(\btOmega)}[\btOmega^{-1}] &= \nu^{(1)}_{q}\,\bPsi_{q}^{(1)-1}, \\
    \E_{q(\btOmega_0)}[\btOmega_0^{-1}] &= \nu^{(0)}_{q}\,\bPsi_{q}^{(0)-1}, \\
    \E_{q(\boldsymbol{\tilde{\eta}})}[\boldsymbol{\tilde{\eta}}] &= \boldsymbol{\eta}_q, \\
    \E_{q(\btmu)}[\bbY_{\btmu}^\top \bbY_{\btmu}]
&= (\bbY-\boldsymbol{1}_n\bmu_q^\top)^\top(\bbY-\boldsymbol{1}_n\bmu_q^\top)+n\bSigma_q,
\end{align*}
we obtain
\begin{align}
\tf(\bA)
&=
\frac{2n+\nu^{(1)}+\nu^{(0)}}{2}\,\log\big|\bI_{r-u}+\bA\bA^\top\big|
\notag\\
&\quad
-\frac{\nu^{(1)}_q}{2}\,
\tr\!\Big( \bPsi_q^{(1)-1}\,\bA^\top \bL^\top \bG_{r,1}\,\bL\,\bA \Big)
-\frac{\nu^{(0)}_q}{2}\,
\tr\!\Big( \bK^\top \bG_{r,2}\,\bK\,\bA^\top \bPsi_q^{(0)-1}\,\bA \Big)
\notag\\
&\quad
-\frac{1}{2}\,\tr\!\Big( \bV_0^{(A)-1}\,\bA^\top \bU_0^{(A)-1}\,\bA \Big)
\notag\\
&\quad
-\tr\!\Big( \nu^{(1)}_q\,\bPsi_q^{(1)-1}\,(\bK^\top \bG_{r,1}\,\bL - \boldsymbol{\eta}_q \bH\,\bL)\,\bA \Big)
+\tr\!\Big( \nu^{(0)}_q\,\bK^\top \bG_{r,2}\,\bL\,\bPsi_q^{(0)-1}\,\bA \Big)
\notag\\
&\quad
+\tr\!\Big( \bV_0^{(A)-1}\,\bA_0^\top \bU_0^{(A)-1}\,\bA \Big)
+\textit{const.},
\tag{S7}\label{eq:tAobj}
\end{align}
where 
\begin{align*}
\bG_{r,1}
&:=
(\bbY-\boldsymbol{1}_n\bmu_q^\top)^\top(\bbY-\boldsymbol{1}_n\bmu_q^\top)
+ n\bSigma_q
+ \psi^{(1)}\bI_r
+ \bB_0 \bM \bB_0^\top,\\
\bG_{r,2}
&:=
(\bbY-\boldsymbol{1}_n\bmu_q^\top)^\top(\bbY-\boldsymbol{1}_n\bmu_q^\top)
+ n\bSigma_q
+ \psi^{(0)}\bI_r,\\
\bH
&:=
\bbX_c^\top \bbY + \bM \bB_0^\top, \tag{S8} \label{eq:notations}
\end{align*}
with $\boldsymbol{1}_n$ the $n$-vector of ones and $\bbX_c=\bbX - \boldsymbol{1}_n \bar{X}^\top$ the centered design matrix. Note that $\bbX_c^\top \boldsymbol{1}_n=\boldsymbol{0}$, hence
$\bbX_c^\top(\bbY-\boldsymbol{1}_n\btmu^\top)=\bbX_c^\top\bbY$. \textit{const.} denotes terms that do not depend on $\bA$ (but may depend on $t$).

\subsubsection*{Explicit form of the $\nabla_{\bA} \tf(\bA)$ and $\nabla_{\vecr(\bA)} \tf(\bA)$ }
Let $\bJ_0(\bA) := \bD_{\bA}^\top \bD_{\bA} = \bI_{r-u} + \bA \bA^\top$. The gradient of $\tf(\bA)$ with respect to $\bA$ is as follows: 
\begin{align*}
\nabla_{\bA}\tf(\bA)
&=
(2n+\nu^{(1)}+\nu^{(0)})\,\bJ_0(\bA)^{-1} \bA
\\
&\quad
-\nu^{(1)}_q\,\bL^\top \bG_{r,1}\,\bL\,\bA \bPsi_q^{(1)-1}\,
-\nu^{(0)}_q\, \bPsi_q^{(0)-1}\,\bA \bK^\top \bG_{r,2}\,\bK\,
-\bU_0^{(A)-1}\,\bA\ \bV_0^{(A)-1}
\\
&\quad
-\nu^{(1)}_q\,(\bK^\top \bG_{r,1}\,\bL - \boldsymbol{\eta}_q \bH\,\bL)^\top\,\bPsi_q^{(1)-1}
+ \nu^{(0)}_q \,\bPsi_q^{(0)-1} \bL^\top\, \bG_{r,2}^\top \,\bK
\\
&\quad
+ \bU_0^{(A)-1}\,\bA_0\, \bV_0^{(A)-1}.
\tag{S9}\label{eq:tAobj1}
\end{align*}
The second-order derivative characterized as a Hessian with respect to the vectorized parameter $\vecr(\bA)$ is as follows:
\begin{align*}
\nabla^2_{\vecr(\bA)}\tf(\bA)
&= (2n+\nu^{(1)}+\nu^{(0)}) \\
&\qquad\times\Big[(\bI_u - \bA^\top \bJ_0(\bA)^{-1} \bA) \otimes \bJ_0(\bA)^{-1} -\big((\bJ_0(\bA)^{-1}\bA)^\top\otimes (\bJ_0(\bA)^{-1}\bA)\big)\mathcal{K}_{(r-u),u} \Big] \\
&\quad
-\nu^{(1)}_q\,
\bPsi_q^{(1)-1} \otimes (\bL^\top \bG_{r,1}\,\bL) - \nu^{(0)}_q\,
(\bK^\top \bG_{r,2}\,\bK) \otimes \bPsi_q^{(0)-1} - \bV_0^{(A)-1} \otimes \bU_0^{(A)-1}.
\tag{S10}\label{eq:tAobj2}
\end{align*}

\subsection{Update of variational factors}
\label{sec:B.2}
This subsection provides the explicit closed-form coordinate updates for the conjugate
variational factors in Algorithm~\ref{alg1}.
These updates are obtained by applying the generic CAVI formula~\eqref{eq:8} to the
reparameterized Bayesian response envelope model under the MFVF assumption.

\subsubsection*{An operator for Gaussian expectations over $q(\vecr(\bA))$.}
In CALVI, the nonconjugate factor $q(\vecr(\bA))$ is updated via a Laplace approximation and is therefore
Gaussian. When deriving the conjugate coordinate updates for $q(\btOmega)$, $q(\btOmega_0)$, and $q(\btmu)$,
we repeatedly encounter expectations of trace-quadratic terms involving
$\bC_{\bA}$ and $\bD_{\bA}$ with respect to this Gaussian factor.
To preserve the standard inverse-Wishart and normal forms after taking these expectations,
we introduce the Kronecker--trace operator $\Ktr[\cdot]$ and provide
two identities (Proposition~\ref{Ktr proposition} and Corollary~\ref{Ktr corollary})
that are used throughout the subsequent updates.

\subsubsection*{Kronecker--trace operator and its properties.}
Let $d,n\in\mathbb{N}$. For a matrix $\boldsymbol{R}\in\mathbb{R}^{d\times d}$ and a block matrix
$\boldsymbol{B}\in\mathbb{R}^{dn\times dn}$ partitioned into $n\times n$ blocks
$\boldsymbol{B}_{ij}\in\mathbb{R}^{d\times d}$, define the Kronecker--trace operator
$\Ktr[\boldsymbol{R}\odot \bB]\in\mathbb{R}^{n\times n}$ by
\[
\big[\Ktr[\boldsymbol{R}\odot \bB]\big]_{ij}:=\tr(\boldsymbol{R}\,\bB_{ij}),\qquad i,j=1,\dots,n.
\]
(Here $\odot$ is a notational separator indicating blockwise products inside the trace; it is not the Hadamard product.)

\begin{proposition}[Kronecker--trace identity]
\label{Ktr proposition}
Let $\bR\in\mathbb{R}^{d\times d}$, $\bS\in\mathbb{R}^{n\times n}$, and
$\bH\in\mathbb{R}^{dn\times dn}$, where $\bH$ is partitioned into $n\times n$ blocks
$\bH_{ij}\in\mathbb{R}^{d\times d}$. Then
\[
\tr\!\big[(\bS\otimes \bR)\bH\big] \;=\; \tr\!\big(\bS\,\Ktr[\bR\odot \bH]\big).
\]
Equivalently,
\[
\tr\!\big[(\bS^\top\otimes \bR)\bH\big]
\;=\;
\tr\!\big(\Ktr[\bR\odot \bH]\,\bS^\top\big)
\;=\;
\tr\!\big(\Ktr[\bR\odot \bH]^\top \bS\big).
\]
In our applications, $\bS$ (e.g., $\btOmega^{-1}$ or $\btOmega_0^{-1}$) is symmetric, so the transpose can be dropped.
\end{proposition}

Corollary~\ref{Ktr corollary} provides closed-form expressions for the Gaussian expectations of
trace-quadratic terms in $\bC_{\bA}$ and $\bD_{\bA}$ that appear in the conjugate coordinate updates.
These identities allow us to rewrite the expected quadratic forms in the standard inverse-Wishart/normal
update forms, using only the Laplace mean $\hbA$ and covariance $\bSigma$.
\begin{corollary}
\label{Ktr corollary}
Let $m:=r-u$ and suppose
\[
\vecr(\bA)\sim \calN\!\big(\vecr(\hbA),\,\bSigma\big),
\qquad
\bSigma\in\mathbb{R}^{mu\times mu}.
\]
Define
\[
\bSigma_T := \mathcal{K}_{m,u}\bSigma \mathcal{K}_{m,u}^\top,
\]
which is the covariance of $\vecr(\bA^\top)$.

Assume $\bC_1\in\mathbb{S}^r$ and $\bC_2$ is symmetric of conformable size.
Then:
\begin{enumerate}
\item[(1)]
\[
\E\!\left[\tr\!\{\bC_{\bA}^\top \bC_1 \bC_{\bA}\bC_2\}\right]
=
\tr\!\left[\left(
\bC_{\hbA}^\top \bC_1 \bC_{\hbA}
+
\Ktr[\bL^\top \bC_1 \bL \odot \bSigma]
\right)\bC_2\right].
\]

\item[(2)]
\[
\E\!\left[\tr\!\{\bD_{\bA}^\top \bC_1 \bD_{\bA}\bC_2\}\right]
=
\tr\!\left[\left(
\bD_{\hbA}^\top \bC_1 \bD_{\hbA}
+
\Ktr[\bK^\top \bC_1 \bK \odot \bSigma_T]
\right)\bC_2\right].
\]
\end{enumerate}
\end{corollary}
\begin{proof}
We prove (1); the proof of (2) follows analogously.

Recall that $\bC_{\bA} = \bK + \bL \bA$ and that
$\vecr(\bA)\sim \calN(\vecr(\hbA),\,\bSigma)$.
Write $\bA = \hbA + (\bA-\hbA)$ and expand:
\begin{align*}
\E\!\left[\tr\!\{\bC_{\bA}^\top \bC_1 \bC_{\bA}\bC_2\}\right]
&= \E\!\left[\tr\!\{(\bK+\bL\bA)^\top \bC_1 (\bK+\bL\bA)\bC_2\}\right] \\
&= \E\!\Big[
\tr\!\{(\bC_{\hbA}+\bL(\bA-\hbA))^\top \bC_1
(\bC_{\hbA}+\bL(\bA-\hbA))\bC_2\}
\Big].
\end{align*}
The cross terms vanish by centering, yielding
\begin{align*}
\E\!\left[\tr\!\{\bC_{\bA}^\top \bC_1 \bC_{\bA}\bC_2\}\right]
&= \tr\!\{\bC_{\hbA}^\top \bC_1 \bC_{\hbA}\bC_2\}
+ \E\!\left[\tr\!\{(\bA-\hbA)^\top \bL^\top \bC_1 \bL (\bA-\hbA)\bC_2\}\right].
\end{align*}
Using the identity
\[
\tr\!\{ \bX^\top \bM \bX \bN \}
= \vecr(\bX)^\top (\bN^\top \otimes \bM)\,\vecr(\bX),
\]
the second term becomes
\begin{align*}
\E\!\left[
\vecr(\bA-\hbA)^\top
(\bC_2^\top \otimes \bL^\top \bC_1 \bL)
\vecr(\bA-\hbA)
\right]
= \tr\!\left[
(\bC_2^\top \otimes \bL^\top \bC_1 \bL)\,\bSigma
\right].
\end{align*}
By Proposition~\ref{Ktr proposition}, this trace can be written as
\[
\tr\!\left[
\Ktr[\bL^\top \bC_1 \bL \odot \bSigma]\;\bC_2
\right].
\]
Combining terms yields
\[
\E\!\left[\tr\!\{\bC_{\bA}^\top \bC_1 \bC_{\bA}\bC_2\}\right]
=
\tr\!\left[
\left(
\bC_{\hbA}^\top \bC_1 \bC_{\hbA}
+
\Ktr[\bL^\top \bC_1 \bL \odot \bSigma]
\right)\bC_2
\right],
\]
which proves (1).

For (2), note that $\bD_{\bA} = \bL - \bK\bA^\top$ and
$\vecr(\bA^\top)\sim \calN(\vecr(\hbA^\top),\,\bSigma_T)$ with
$\bSigma_T=\mathcal{K}_{m,u}\bSigma\mathcal{K}_{m,u}^\top$.
Repeating the same argument with $\bA^\top$ in place of $\bA$
yields the stated result.
\end{proof}

\subsubsection*{Update for $q(\boldsymbol{\tilde{\eta}})$.}
\begin{align*}
    \log q(\boldsymbol{\tilde{\eta}}) &\propto \bbE_{-q(\boldsymbol{\tilde \eta})} \left[ \log p(\btmu, \boldsymbol{\tilde \eta}, \btOmega, \btOmega_0, \bA \mid \mathbb{X}, \mathbb{Y}) \right], \\
    & \propto -\frac{1}{2} \bbE_{q(\btmu) q(\btOmega) q(\vecr(\bA))} \left[ \tr \left[ \left( \boldsymbol{\tilde \eta}(\mathbb{X}_c^\top \mathbb{X}_c + \boldsymbol{M}) \boldsymbol{\tilde \eta}^\top -2  \boldsymbol{\tilde \eta}(\mathbb{X}_c^\top \mathbb{\tilde Y}_{\bmu} + \boldsymbol{M} \boldsymbol{B}^\top_0) \boldsymbol{C}_{\bA} \right) \btOmega^{-1} \right] \right] \\
    & = -\frac{1}{2} \bbE_{ q(\btOmega) q(\vecr(\bA))} \left[ \tr \left[ \left( \boldsymbol{\tilde \eta}(\mathbb{X}_c^\top \mathbb{X}_c + \boldsymbol{M}) \boldsymbol{\tilde \eta}^\top -2  \boldsymbol{\tilde \eta}(\mathbb{X}_c^\top \bbY + \boldsymbol{M} \boldsymbol{B}^\top_0) \boldsymbol{C}_{\bA} \right) \btOmega^{-1} \right] \right] \\
    & = -\frac{1}{2} \bbE_{ q(\vecr(\bA))} \left[ \tr \left[ \nu^{(1)}_q \left( \boldsymbol{\tilde \eta}(\mathbb{X}_c^\top \mathbb{X}_c + \boldsymbol{M}) \boldsymbol{\tilde \eta}^\top -2  \boldsymbol{\tilde \eta}(\mathbb{X}_c^\top \bbY + \boldsymbol{M} \boldsymbol{B}^\top_0) \boldsymbol{C}_{\bA} \right) \bPsi^{(1)-1}_q \right] \right] \\
    & = -\frac{1}{2} \tr \left[ \nu^{(1)}_q \left( \boldsymbol{\tilde \eta}(\mathbb{X}_c^\top \mathbb{X}_c + \boldsymbol{M}) \boldsymbol{\tilde \eta}^\top -2  \boldsymbol{\tilde \eta}(\mathbb{X}_c^\top \bbY + \boldsymbol{M} \boldsymbol{B}^\top_0) \bC_{\hbA} \right) \bPsi^{(1)-1}_q \right] \\
    & \propto
    \text{const.} -\frac{1}{2} \tr \left[\boldsymbol{V}^{(\boldsymbol{\eta}) -1}_q (\boldsymbol{\tilde{\eta}} - \boldsymbol{\eta}_q)^\top \boldsymbol{U}^{(\boldsymbol{\eta}) -1}_q (\boldsymbol{\tilde{\eta}} - \boldsymbol{\eta}_q) \right],
\end{align*}
where $\mathbb{\tilde Y}_{\btmu} = (\bbY-\boldsymbol{1}_n\btmu^\top)$.
\begin{align*}
    q(\boldsymbol{\tilde{\eta}}) & \propto \mathcal{MN}_{u,p}(\boldsymbol{\eta}_q, \boldsymbol{U}^{(\boldsymbol{\eta})}_q, \boldsymbol{V}^{(\boldsymbol{\eta})}_q),\\
    \boldsymbol{\eta}_q &= \bC_{\hbA}^\top (\mathbb{X}_c^\top \mathbb{Y} + \boldsymbol{M} \boldsymbol{B}_0^\top)^\top (\mathbb{X}_c^\top \mathbb{X}_c + \boldsymbol{M})^{-1}, \\
    \boldsymbol{U}^{(\boldsymbol{\eta})}_q &= \bPsi^{(1)}_q/\nu^{(1)}_q, \\
    \boldsymbol{V}^{(\boldsymbol{\eta})}_q &= (\mathbb{X}_c^\top \mathbb{X}_c + \boldsymbol{M})^{-1}.
\end{align*}

\subsubsection*{Update for $q(\btOmega)$.}
\begin{align*}
    \log q(\btOmega) & \propto  \bbE_{-q(\btOmega)} \left[ \log p(\btmu, \boldsymbol{\tilde \eta}, \btOmega, \btOmega_0, \bA \mid \mathbb{X}, \mathbb{Y}) \right] \\
    & \propto \text{const.} -\frac{n + \nu^{(1)} + u + p + 1}{2} \log \left| \mathbf{\tilde{\Omega}} \right| \\
    &\quad -\frac{1}{2} \bbE_{q(\btmu)q(\vecr(\bA))} \left[ \tr \left[ \boldsymbol{C}^\top_{\bA} (\mathbb{\tilde Y}_{\bmu}^\top \mathbb{\tilde Y}_{\bmu} + \boldsymbol{B}_0 \boldsymbol{M} \boldsymbol{B}^\top_0 + \psi^{(1)} \boldsymbol{I}_r) \boldsymbol{C}_{\bA} \btOmega^{-1} \right] \right] \\
    &\quad -\frac{1}{2} \bbE_{q(\btmu)q(\boldsymbol{\tilde{\eta}})q(\vecr(\bA))} \left[ -2 \tr \left[ \boldsymbol{\tilde \eta}(\mathbb{X}_c^\top \mathbb{\tilde Y}_{\bmu} + \boldsymbol{M} \boldsymbol{B}_0^\top) \boldsymbol{C}_{\bA} \btOmega^{-1} \right] \right] \\
    &\quad -\frac{1}{2} \bbE_{q(\boldsymbol{\tilde{\eta}})} \left[ \tr\left[ \boldsymbol{\tilde \eta}(\mathbb{X}_c^\top \mathbb{X}_c + \boldsymbol{M}) \boldsymbol{\tilde \eta}^\top \btOmega^{-1} \right] \right] \\
    & = \text{const.} -\frac{n + \nu^{(1)} + u + p + 1}{2} \log \left| \mathbf{\tilde{\Omega}} \right| \\
    &\quad -\frac{1}{2} \bbE_{q(\vecr(\bA))} \left[ \tr \left[ \boldsymbol{C}^\top_{\bA} \boldsymbol{G}_{r,1} \boldsymbol{C}_{\bA} \btOmega^{-1} \right] \right] \\
    &\quad -\frac{1}{2}  \bbE_{q(\boldsymbol{\tilde{\eta}})q(\vecr(\bA))} \left[ -2 \tr \left[ \boldsymbol{\tilde \eta}(\mathbb{X}_c^\top \bbY + \boldsymbol{M} \boldsymbol{B}_0^\top) \boldsymbol{C}_{\bA} \btOmega^{-1} \right] \right] \\
    &\quad -\frac{1}{2} \bbE_{q(\boldsymbol{\tilde{\eta}})} \left[ \tr\left[ \boldsymbol{\tilde \eta}(\mathbb{X}_c^\top \mathbb{X}_c + \boldsymbol{M}) \boldsymbol{\tilde \eta}^\top \btOmega^{-1} \right] \right] \\
    & = \text{const.} -\frac{n + \nu^{(1)} + u + p + 1}{2} \log \left| \mathbf{\tilde{\Omega}} \right| \\
    &\quad -\frac{1}{2} \bbE_{q(\vecr(\bA))} \left[ \tr \left[ \boldsymbol{C}^\top_{\bA} \boldsymbol{G}_{r,1} \boldsymbol{C}_{\bA} \btOmega^{-1} \right] \right] \\
    &\quad -\frac{1}{2} \bbE_{q(\vecr(\bA))} \left[ -2 \tr \left[ \boldsymbol{\eta}_q(\mathbb{X}_c^\top \bbY + \boldsymbol{M} \boldsymbol{B}_0^\top) \boldsymbol{C}_{\bA} \btOmega^{-1} \right] \right] \\
    &\quad -\frac{1}{2} \tr\left[ \left( \tr\left[(\mathbb{X}_c^\top \mathbb{X}_c + M) \boldsymbol{V}_q^{(\eta)}\right]\boldsymbol{U}_q^{(\eta)} + \boldsymbol{\eta}_q(\mathbb{X}_c^\top \mathbb{X}_c + \boldsymbol{M}) \boldsymbol{\eta}_q^\top \right) \btOmega^{-1} \right] \\
    & = \text{const.} -\frac{n + \nu^{(1)} + u + p + 1}{2} \log \left| \mathbf{\tilde{\Omega}} \right| \\
    &\quad -\frac{1}{2} \tr \Big[ \Big( \bC^\top_{\hbA} \boldsymbol{G}_{r,1} \bC_{\hbA} \\
    &\qquad + \Ktr\left[\bL^\top \boldsymbol{G}_{r,1} \bL \odot H(\vecr(\hbA))\right] \Big) \btOmega^{-1} \Big] \tag{S11} \label{eq:upOmega} \\
    &\quad -\frac{1}{2} \tr \left[ -2 \boldsymbol{\eta}_q(\mathbb{X}_c^\top \bbY + \boldsymbol{M} \boldsymbol{B}_0^\top) \bC_{\hbA} \btOmega^{-1} \right] \\
    &\quad -\frac{1}{2} \tr\left[ \left( \tr\left[(\mathbb{X}_c^\top \mathbb{X}_c + M) \boldsymbol{V}_q^{(\eta)}\right]\boldsymbol{U}_q^{(\eta)} + \boldsymbol{\eta}_q(\mathbb{X}_c^\top \mathbb{X}_c + \boldsymbol{M}) \boldsymbol{\eta}_q^\top \right) \btOmega^{-1} \right] \\
    & \propto \text{const.} - \frac{\nu_{q}^{(1)} + u + 1}{2} \log |\btOmega| 
    - \frac{1}{2} \tr \left[ \boldsymbol{\Psi}_{q}^{(1)} \btOmega^{-1} \right], \\
    q(\btOmega) & \propto \mathcal{IW}_u(\boldsymbol{\Psi}_{q}^{(1)}, \nu_{q}^{(1)}),\\
    \boldsymbol{\Psi}_{q}^{(1)} &=   \bC^\top_{\hbA} \boldsymbol{G}_{r,1} \bC_{\hbA} + \Ktr\left[\bL^\top \boldsymbol{G}_{r,1} \bL \odot H(\vecr(\hbA))\right] \\
    & \quad \quad \quad -2 \boldsymbol{\eta}_q \left( \mathbb{X}_c^\top \mathbb{Y} + \boldsymbol{M} \boldsymbol{B}_0^\top \right) \bC_{\hbA} + \tr\left[(\mathbb{X}_c^\top \mathbb{X}_c + \boldsymbol{M}) \boldsymbol{V}_q^{(\eta)}\right]\boldsymbol{U}_q^{(\eta)} + \boldsymbol{\eta}_q (\mathbb{X}_c^\top \mathbb{X}_c + \boldsymbol{M}) \boldsymbol{\eta}_q^\top, \\
    H(\vecr(\hbA)) &:= -\Big[\nabla^2_{\vecr(\bA)}\tf(\hbA)\Big]^{-1}, \\
    \nu_q^{(1)} &=  n + \nu^{(1)} + p.
\end{align*}
Here $\boldsymbol{G}_{r,1}$ in given in~\eqref{eq:notations}. Equation~\eqref{eq:upOmega} is the result of Corollary~\ref{Ktr corollary} (1). 

\subsubsection*{Update for $q(\btOmega_0)$.}
\begin{align*}
    \log q(\btOmega_0) &\propto \bbE_{-q(\btOmega_0)} \left[ \log p(\btmu, \boldsymbol{\tilde \eta}, \btOmega, \btOmega_0, \bA \mid \mathbb{X}, \mathbb{Y}) \right] \\
    &\propto \textit{const.} -\frac{n + \nu^{(0)} + r - u + 1}{2} \log \left| \mathbf{\tilde{\Omega}}_0 \right| \\
    &\quad -\frac{1}{2} \bbE_{q(\btmu) q(\vecr(\bA))} \left[ \tr \left[ \boldsymbol{D}^\top_{\bA} (\mathbb{\tilde Y}_{\bmu}^\top \mathbb{\tilde Y}_{\bmu} + \psi^{(0)} \boldsymbol{I}_{r}) \boldsymbol{D}_{\bA} \btOmega_0^{-1} \right] \right] \\
    &= \textit{const.} -\frac{n + \nu^{(0)} + r - u + 1}{2} \log \left| \mathbf{\tilde{\Omega}}_0 \right| \\
    &\quad -\frac{1}{2} \bbE_{q(\vecr(\bA))} \left[ \tr \left[ \boldsymbol{D}^\top_{\bA} \boldsymbol{G}_{r,2} \boldsymbol{D}_{\bA} \btOmega_0^{-1} \right] \right] \\
    &= \textit{const.} -\frac{n + \nu^{(0)} + r - u + 1}{2} \log \left| \btOmega_0 \right| \\
    &\quad -\frac{1}{2} \tr \left[ \left( \bD^\top_{\hbA} \boldsymbol{G}_{r,2} \bD_{\hbA} + \Ktr[\bK^\top \boldsymbol{G}_{r,2} \bK \odot \tilde{H}(\vecr(\hbA))] \right)\btOmega_0^{-1} \right]  \tag{S12} \label{eq:upOmega0} \\
    &\propto  
    \text{const.} - \frac{\nu^{(0)}_q + r - u + 1}{2} \log |\btOmega_0| 
    - \frac{1}{2} \tr \left[ \boldsymbol{\Psi}_{q}^{(0)} \btOmega_0^{-1} \right], \\
    q(\btOmega_0) & \propto \mathcal{IW}_{r-u}(\boldsymbol{\Psi}_{q}^{(0)}, \nu_{q}^{(0)}),\\
    \boldsymbol{\Psi}_{q}^{(0)} &= \bD^\top_{\hbA} \boldsymbol{G}_{r,2} \bD_{\hbA} + \Ktr[\bK^\top \boldsymbol{G}_{r,2} \bK \odot \tilde{H}(\vecr(\hbA))], \\
    \tilde{H}(\vecr(\hbA)) &= \mathcal{K}_{(r-u), u}H(\vecr(\hbA))\mathcal{K}_{u, (r-u)}, \\
    \nu_q^{(0)} &= n + \nu^{(0)}.
\end{align*}
Here $\boldsymbol{G}_{r,2}$ is given in~\eqref{eq:notations}. Equation~\eqref{eq:upOmega0} is the result of Corollary~\ref{Ktr corollary} (2). 

\subsubsection*{Update for $q(\btmu)$.}
\begin{align*}
    \log q(\btmu) &\propto \bbE_{-q(\btmu)} \left[ \log p(\btmu, \boldsymbol{\tilde \eta}, \btOmega, \btOmega_0, \bA \mid \mathbb{X}, \mathbb{Y}) \right] , \\
    &\propto -\frac{1}{2} \bbE_{q(\btOmega) q(\btOmega_0) q(\vecr(\bA))} \left[ \tr \left[ (\mathbb{\tilde Y}_{\bmu}^\top \mathbb{\tilde Y}_{\bmu}) \left( \bC_{\bA} \btOmega^{-1} \bC^\top_{\bA} + \bD_{\bA} \btOmega_0^{-1} \bD^\top_{\bA} \right) \right] \right] \\
    &= -\frac{1}{2} \bbE_{q(\vecr(\bA))} \left[ \tr \left[ (\btmu \btmu^\top - 2 \bar{\boldsymbol{Y}} \btmu^\top) n \left( \bC_{\bA} \nu^{(1)}_q \bPsi^{(1) -1}_q \bC^\top_{\bA} + \bD_{\bA} \nu^{(0)}_q \bPsi^{(0) -1}_q \bD^\top_{\bA} \right) \right] \right] \\
    &= -\frac{1}{2} \tr \left[ (\btmu \btmu^\top - 2 \bar{\boldsymbol{Y}} \btmu^\top) ( \boldsymbol{S}_1 + \boldsymbol{S}_2 ) \right] \\
    &\propto  
    \text{const.} - \frac{1}{2} \tr\left[ (\btmu-\bmu_q)^\top \bSigma_q^{-1} (\btmu-\bmu_q) \right], \\
    q(\btmu) & \propto \mathcal{N}_{r}(\bmu_q, \bSigma_q),\\
    \bmu_q &= \boldsymbol{\bar{Y}}, \quad \bSigma_q = (\boldsymbol{S}_1 + \boldsymbol{S}_2)^{-1}, \\
    \boldsymbol{S}_1 &= n \nu_q^{(1)} \left(\bC_{\hbA} \boldsymbol{\Psi}_q^{(1) -1} \bC_{\hbA}^\top + \bL \Ktr\left[\boldsymbol{\Psi}_q^{(1) -1} \odot \tilde{H}(\vecr(\hbA)) \right] \bL^\top \right), \\
    \boldsymbol{S}_2 &= n\nu_q^{(0)} \left(\bD_{\hbA} \boldsymbol{\Psi}_q^{(0) -1} \bD_{\hbA}^\top + \bK \Ktr\left[\boldsymbol{\Psi}_q^{(0) -1} \odot H(\vecr(\hbA)) \right] \bK^\top \right).
\end{align*}

\subsection{Approximated ELBO}
\label{sec:B.3}
\begin{align*}
    \tilde\calL(q) = & \mathbb{E}_{q(\cdot)}\Big[\log p(\btmu, \boldsymbol{\tilde \eta}, \btOmega, \btOmega_0, \bA \mid \mathbb{X}, \mathbb{Y}) - \log q(\btmu, \boldsymbol{\tilde \eta}, \btOmega, \btOmega_0, \bA) \Big] \\
    \approx & \ const. + \mathbb{E}_{q(\bA)}\Big[\tf(\hbA) + \frac{1}{2}\vecr(\bA-\hbA)^\top  \tf^{\prime\prime}(\hbA) \vecr(\bA-\hbA) \Big] -\mathbb{E}_{q(\btmu, \boldsymbol{\tilde \eta}, \btOmega, \btOmega_0, \bA)}\Big[\log q(\btmu, \boldsymbol{\tilde \eta}, \btOmega, \btOmega_0, \bA) \Big] \\
    = &-\frac{up}{2} \log (2 \pi) -\frac{u}{2} \log |\boldsymbol{M}^{-1}|  -\frac{\nu^{(1)} u + \nu^{(0)} (r-u)}{2}\log(2) \\ 
    & - \log (\Gamma_u (\frac{\nu^{(1)}}{2})) +\frac{\nu^{(1)}}{2} \log \left| \boldsymbol{\Psi} \right| - \log (\Gamma_{r-u} (\frac{\nu^{(0)}}{2})) +\frac{\nu^{(0)}}{2} \log \left| \boldsymbol{\Psi}_0 \right|  \\
    & -\frac{((r-u)u)}{2} \log (2 \pi) - \frac{u}{2} \log |\boldsymbol{U}_0^{(\bA)}| - \frac{(r-u)}{2} \log |\boldsymbol{V}_0^{(\bA)}| \\
    & + \frac{n + \nu^{(1)} + u + p + 1}{2} \left( \sum_{s=1}^{u}\psi (\frac{\nu^{(1)}_q + (1 - s)}{2}) + u \log(2) + \log \left| \boldsymbol{\Psi}^{(1) -1}_q \right|\right) \\
    & +\frac{n + \nu^{(0)} + r - u + 1}{2}  \left( \sum_{s=1}^{r-u}\psi (\frac{\nu^{(0)}_q + (1 - s)}{2}) + (r-u) \log(2) + \log \left| \boldsymbol{\Psi}^{(0) -1}_q \right|\right) \\ 
    & -\frac{\nu^{(1)}_q}{2} \tr \left[ \left( \tr[\left( \mathbb{X}_c^\top \mathbb{X}_c+ \boldsymbol{M}\right) \boldsymbol{V}^{(\boldsymbol{\eta})}_q] \boldsymbol{U}^{(\boldsymbol{\eta})}_q  + \boldsymbol{\eta}_q \left( \mathbb{X}_c^\top \mathbb{X}_c+ \boldsymbol{M}\right) \boldsymbol{\eta}_q^\top \right) \mathbf{\Psi}_q^{(1)-1} \right] \\
    & + \tf(\hbA) - \frac{(r-u)u}{2} + \frac{r}{2} \log (2 \pi e) + \frac{1}{2} \log |\bSigma_q| + \frac{up}{2} \log (2 \pi e) + \frac{p}{2} \log |\boldsymbol{U}^{(\eta)}_q| + \frac{u}{2} \log |\boldsymbol{V}^{(\eta)}_q| \\ 
    & + \frac{\nu^{(1)}_q}{2} \log |\boldsymbol{\Psi}^{(1) -1}_q| + \frac{\nu^{(1)}_q u}{2} \log(2e) + \log \Gamma_u (\frac{\nu^{(1)}_q}{2}) \\
    & - \frac{(\nu^{(1)}_q+u+1)}{2} \left( \psi_u (\frac{\nu^{(1)}_q}{2}) + u \log(2) + \log |\boldsymbol{\Psi}^{(1) -1}_q| \right) \\ 
    & + \frac{\nu^{(0)}_q}{2} \log |\boldsymbol{\Psi}^{(0) -1}_q| + \frac{\nu^{(0)}_q (r-u)}{2} \log(2e) + \log \Gamma_{r-u} (\frac{\nu^{(0)}_q}{2}) \\
    & - \frac{(\nu^{(0)}_q+(r-u)+1)}{2} \left( \psi_{r-u} (\frac{\nu^{(0)}_q}{2}) + (r-u) \log(2) + \log |\boldsymbol{\Psi}^{(0) -1}_q|  \right) \\
    & + \frac{(r-u)u}{2} \log (2 \pi e) + \frac{1}{2} \log | H(\vecr(\hbA)) |, \label{eq:ELBO} \tag{S13}
\end{align*}
where $\psi_d$ denotes the standard multivariate digamma function and $\Gamma_d$ denotes the standard multivariate gamma function. 

\section{CALVI algorithm for the Bayesian predictor envelope model}
\label{sec:C}

This section extends CALVI to the Bayesian predictor envelope model.
Let $\bX\in\mathbb{R}^p$ denote the predictor and $\bY\in\mathbb{R}^r$ the response, and let
$m<p$ be the envelope dimension. Under the Euclidean envelope parameterization
$\{\bGamma(\bA),\bGamma_0(\bA)\}$, the envelope matrix $\bA\in\mathbb{R}^{(p-m)\times m}$ enters the
likelihood through inverse square-root factors, so that direct Laplace updates for $\bA$
inherit the same Kronecker--sum derivative bottlenecks discussed in Section~\ref{sec4.1}.
We therefore adopt a reparameterization analogous to Section~\ref{sec4.2} to remove the dominant
inverse square-root terms from the $\bA$-coordinate objective, yielding a tractable and numerically
stable Laplace update for $\bA$.

\begin{align*}
    \boldsymbol{Y} &= \bmu_{\bY} + \boldsymbol{\eta}^\top \bGamma^\top(\bA) (\boldsymbol{X} - \bmu_{\boldsymbol{X}}) + \bepsilon, \quad \bepsilon \sim \mathcal{N}_r \left(\boldsymbol{0}, \bSigma_{\boldsymbol{Y} \mid \boldsymbol{X}}\right), \\
    \boldsymbol{X} &\sim \mathcal{N}_p\left(\bmu_{\boldsymbol{X}}, \bGamma(\bA) \boldsymbol{\Omega}^{(1)}_{\boldsymbol{X}} \bGamma^\top(\bA) + \bGamma_0(\bA) \boldsymbol{\Omega}^{(0)}_{\boldsymbol{X}} \bGamma_0^\top(\bA)\right) \label{eq:penv} \tag{S14},
\end{align*}
where $\bmu_{\bY} \in \bbR^{r}, \bmu_{\bX} \in \bbR^{p}, \boldsymbol{\eta} \in \bbR^{m \times r}, \bSigma_{\bY \mid \bX} \in \calC_r, \bOmega^{(1)}_{\bX} \in \calC_{m},$ and $\bOmega^{(0)}_{\bX} \in \calC_{p-m}$. Then 
\begin{align*}
    \bGamma(\bA) &:= \boldsymbol{C}_{\bA}(\boldsymbol{C}_{\bA}^\top \boldsymbol{C}_{\bA})^{-1/2}, \qquad \boldsymbol{C}_{\bA} \equiv \begin{pmatrix} \boldsymbol{I}_m \\ \bA \end{pmatrix} \in \bbR^{p \times m}, \\
    \bGamma_0(\bA) &:= \boldsymbol{D}_{\bA}(\boldsymbol{D}_{\bA}^\top \boldsymbol{D}_{\bA})^{-1/2}, \qquad \boldsymbol{D}_{\bA} \equiv \begin{pmatrix} -\bA^\top \\ \boldsymbol{I}_{p-m} \end{pmatrix} \in \bbR^{p \times (p-m)}, 
\end{align*}
and $\bbeta = \boldsymbol{\eta}^\top \bGamma^\top(\bA)$.

We consider the following prior specifications as follows:
\begin{enumerate}
    \item[(i)] $\pi(\bmu_{\boldsymbol{X}}) \propto 1$ and $\pi(\bmu_{\boldsymbol{Y}}) \propto 1$ are an improper flat prior on $\bbR^{p}$ and $\bbR^{r}$, respectively. 
    \item[(ii)] $\bSigma_{\boldsymbol{Y} \mid \boldsymbol{X}} \sim \mathcal{IW}_{r}(\boldsymbol{\Psi}_{\boldsymbol{Y}}, \nu_{\boldsymbol{Y}})$, where $\boldsymbol{\Psi}_{\boldsymbol{Y}} \in \mathcal{C}_{r}$ and $\nu_{\boldsymbol{Y}} > r-1$. 
    \item[(iii)] $\boldsymbol{\Omega}_{\boldsymbol{X}}^{(1)} \sim \mathcal{IW}_m(\boldsymbol{\Psi}_{\boldsymbol{X}}^{(1)}, \nu_{\boldsymbol{X}}^{(1)})$, where $\boldsymbol{\Psi}_{\boldsymbol{X}}^{(1)} \in \mathcal{C}_{m}$ and $\nu_{\boldsymbol{X}}^{(1)} > m-1$. 
    \item[(iv)] $\boldsymbol{\Omega}_{\boldsymbol{X}}^{(0)} \sim \mathcal{IW}_{p-m}(\boldsymbol{\Psi}_{\boldsymbol{X}}^{(0)}, \nu_{\boldsymbol{X}}^{(0)})$, where $\boldsymbol{\Psi}_{\boldsymbol{X}}^{(0)} \in \mathcal{C}_{p-m}$ and $\nu_{\boldsymbol{X}}^{(0)} > p-m-1$. 
    \item[(v)]   $\boldsymbol{\eta} \mid \bA, \boldsymbol{\Omega}_{\boldsymbol{X}}^{(1)} \sim \mathcal{MN}_{m,r}((\boldsymbol{C}^\top_{\bA} \boldsymbol{C}_{\bA})^{1/2}\boldsymbol{C}^\top_{\bA}\boldsymbol{B}_0, \psi^{(\boldsymbol{\eta})}_0 (\boldsymbol{C}^\top_{\bA} \boldsymbol{C}_{\bA}) \boldsymbol{\Omega}_{\boldsymbol{X}}^{(1)} (\boldsymbol{C}^\top_{\bA} \boldsymbol{C}_{\bA}), \bSigma_{\bY \mid \bX})$, where $\psi_0^{(\boldsymbol{\eta})} \in \bbR$, and $\boldsymbol{B}_0 \in \bbR^{p \times r}$. 
    \item[(vi)] $\bA \sim \mathcal{MN}_{(p-m),m}(\bA_0, \boldsymbol{U}^{(A)}_0, \boldsymbol{V}^{(A)}_0)$, where $\bA_0 \in \bbR^{(p-m) \times m}$,  $\boldsymbol{U}^{(A)}_0 \in \mathcal{C}_{p-m}$, and  $\boldsymbol{V}^{(A)}_0 \in \mathcal{C}_{m}$. 
\end{enumerate}
The prior specifications follow \citet{chakraborty2024comprehensive}, except that we use a slightly different prior for $\boldsymbol{\eta}$ due to the reparameterization.

Let $(\boldsymbol{X}_1, \boldsymbol{Y}_1), \dots, (\boldsymbol{X}_n, \boldsymbol{Y}_n)$ be the $n$ independent observations from the predictor envelope model. Define $\bbY = (\boldsymbol{Y}_1, \dots, \boldsymbol{Y}_n)^\top \in \bbR^{n \times r}$ and $\bbX = (\boldsymbol{X}_1, \dots, \boldsymbol{X}_n)^\top \in \bbR^{n \times p}$. Then unnormalized posterior distribution is:
\begin{align*}
    & \log p(\bmu_Y, \bmu_X, \bSigma_{Y|X}, \boldsymbol{\eta}, \boldsymbol{\Omega}_{\boldsymbol{X}}^{(1)}, \boldsymbol{\Omega}_{\boldsymbol{X}}^{(0)}, \bA \mid \mathbb{X}, \mathbb{Y}) = \\
    &\quad \textit{const.}-\frac{\nu_{\boldsymbol{Y}} + n + m + r + 1}{2} \log |\bSigma_{\boldsymbol{Y|X}}| 
    - \frac{\nu^{(1)}_{\boldsymbol{X}} + n + m + 1}{2} \log |\boldsymbol{\Omega}_{\boldsymbol{X}}^{(1)}| - \frac{\nu^{(0)}_{\boldsymbol{X}} + n + (p - m) + 1}{2} \log |\boldsymbol{\Omega}_{\boldsymbol{X}}^{(0)}| \\
    &\quad - \frac{r}{2} \log | \psi_0^{(\boldsymbol{\eta})} (\boldsymbol{C}^\top_{\bA} \boldsymbol{C}_{\bA})\boldsymbol{\Omega}_{\boldsymbol{X}}^{(1)}(\boldsymbol{C}^\top_{\bA} \boldsymbol{C}_{\bA})| \\
    &\quad - \frac{1}{2} \tr \left[ (\mathbb{Y}_{\bmu} - \mathbb{X}_{\bmu} \bGamma(\bA) \boldsymbol{\eta}) \bSigma_{Y|X}^{-1} (\mathbb{Y}_{\bmu} - \mathbb{X}_{\bmu}\bGamma(\bA) \boldsymbol{\eta})^T \right] - \frac{1}{2} \tr \left[ \bGamma^\top(\bA)\mathbb{X}_{\bmu}^\top \mathbb{X}_{\bmu} \bGamma(\bA)\boldsymbol{\Omega}_{\boldsymbol{X}}^{(1) -1} \right] \label{eq:penvGamma} \tag{S15} \\
    &\quad - \frac{1}{2} \tr \left[ \bGamma_0^\top(\bA)\mathbb{X}_{\bmu}^\top \mathbb{X}_{\bmu} \bGamma_0(\bA) \boldsymbol{\Omega}_{\boldsymbol{X}}^{(0) -1} \right] \label{eq:penvGamma0} \tag{S16} \\
    &\quad - \frac{1}{2} \tr \left[ \bSigma_{\bY|\bX}^{-1} (\boldsymbol{\eta} - (\boldsymbol{C}^\top_{\bA}\boldsymbol{C}_{\bA})^{1/2}\boldsymbol{C}^\top_{\bA} \boldsymbol{B}_0)^\top \psi_0^{(\boldsymbol{\eta}) -1} (\boldsymbol{C}^\top_{\bA}\boldsymbol{C}_{\bA})^{-1} \boldsymbol{\Omega}_{\boldsymbol{X}}^{(1) -1} (\boldsymbol{C}^\top_{\bA}\boldsymbol{C}_{\bA})^{-1} (\boldsymbol{\eta} - (\boldsymbol{C}^\top_{\bA}\boldsymbol{C}_{\bA})^{1/2}\boldsymbol{C}^\top_{\bA}\boldsymbol{B}_0) \right] \\
    &\quad - \frac{1}{2} \tr \left[ \bSigma_{Y|X}^{-1} \boldsymbol{\Psi}_Y \right]
    - \frac{1}{2} \tr \left[ \boldsymbol{\Omega}_{\boldsymbol{X}}^{(1) -1} \boldsymbol{\Psi}^{(1)}_{\boldsymbol{X}} \right] 
    - \frac{1}{2} \tr \left[ \boldsymbol{\Omega}_{\boldsymbol{X}}^{(0) -1} \boldsymbol{\Psi}^{(0)}_{\boldsymbol{X}} \right] - \frac{1}{2} \tr\left[\boldsymbol{V}_0^{(A) -1} (\bA - \bA_0)^\top \boldsymbol{U}_0^{(A) -1} (\bA - \bA_0) \right], \label{eq:penvpost} \tag{S17}
\end{align*}
where $\mathbb{Y}_{\bmu} = (\mathbb{Y} - \boldsymbol{1}_n \bmu_{\boldsymbol{Y}}^\top)$ and $\mathbb{X}_{\bmu} = (\mathbb{X} - \boldsymbol{1}_n \bmu_{\boldsymbol{X}}^\top)$.

As mentioned in Section~\ref{sec4.1}, the derivatives of Equations~\eqref{eq:penvGamma} and~\eqref{eq:penvGamma0} with respect to $\bA$ in the predictor envelope model incur substantial computational cost. Now consider the reparameterization to the posterior distribution. Especially, reparameterization for $\boldsymbol{\eta}$ is different compared to Bayesian response envelope model.
\begin{align*}
    \boldsymbol{\tilde{\eta}} &= (\boldsymbol{C}^\top_{\bA} \boldsymbol{C}_{\bA})^{-1/2} \boldsymbol{\eta}, \\
    \btOmega^{(1)}_{\boldsymbol{X}} &= (\boldsymbol{C}^\top_{\bA} \boldsymbol{C}_{\bA})^{1/2} \boldsymbol{\Omega}^{(1)}_{\boldsymbol{X}} (\boldsymbol{C}^\top_{\bA} \boldsymbol{C}_{\bA})^{1/2} , \\ 
    \btOmega^{(0)}_{\boldsymbol{X}} &=(\boldsymbol{D}^\top_{\bA} \boldsymbol{D}_{\bA})^{1/2} \boldsymbol{\Omega}^{(0)}_{\boldsymbol{X}} 
    (\boldsymbol{D}^\top_{\bA} \boldsymbol{D}_{\bA})^{1/2}.  \label{eq:penvrep} \tag{S18}
\end{align*}
Consider $\boldsymbol{\Psi}^{(1)}_{\boldsymbol{X}} = \psi^{(1)}_{\boldsymbol{X}} \boldsymbol{I}_m$ and $\boldsymbol{\Psi}^{(0)}_{\boldsymbol{X}} = \psi^{(0)}_{\boldsymbol{X}} \boldsymbol{I}_{p-m}$ for some positive constants $\psi^{(1)}_{\boldsymbol{X}}$ and $\psi^{(0)}_{\boldsymbol{X}}$, then the reparameterized posterior distribution of the predictor envelope model is as follows.
\begin{align*}
    & \log p(\bmu_Y, \bmu_X, \bSigma_{Y|X}, \boldsymbol{\tilde \eta}, \btOmega_{\boldsymbol{X}}^{(1)}, \btOmega_{\boldsymbol{X}}^{(0)}, \bA \mid \mathbb{X}, \mathbb{Y}) = \\
    &\quad \text{const.} - \frac{\nu_{\boldsymbol{Y}} + n + m + r + 1}{2} \log |\bSigma_{Y|X}| 
    - \frac{\nu^{(1)}_{\boldsymbol{X}} + n + m + r + 1}{2} \log |\btOmega_{\boldsymbol{X}}^{(1)}| - \frac{\nu^{(0)}_{\boldsymbol{X}} + n + (p - m) + 1}{2} \log |\btOmega_{\boldsymbol{X}}^{(0)}| \\
    &\quad + \frac{2n + \nu^{(1)}_{\boldsymbol{X}} + \nu^{(0)}_{\boldsymbol{X}}}{2} \log |\boldsymbol{I}_{m} + \bA^\top \bA| \\
    &\quad - \frac{1}{2} \tr \left[ \boldsymbol{C}^\top_{\bA}(\mathbb{X}_{\bmu}^\top \mathbb{X}_{\bmu} + \psi_0^{(\boldsymbol{\eta}) -1} \bB_0 \bSigma_{\bY|\bX}^{-1} \bB_0^\top + \psi_{\boldsymbol{X}}^{(1)} \boldsymbol{I}_p) \boldsymbol{C}_{\bA} \btOmega_{\boldsymbol{X}}^{(1) -1} \right] - \frac{1}{2} \tr \left[ \boldsymbol{C}^\top_{\bA} \mathbb{X}_{\bmu}^\top \mathbb{X}_{\bmu} \boldsymbol{C}_{\bA} \boldsymbol{\tilde \eta} \bSigma_{\boldsymbol{Y}|\boldsymbol{X}}^{-1} \boldsymbol{\tilde \eta}^\top \right] \\
    &\quad - \frac{1}{2} \tr \left[ -2 (\psi_0^{(\boldsymbol{\eta}) -1} \btOmega_{\boldsymbol{X}}^{(1) -1} \boldsymbol{\tilde \eta} \bSigma_{\bY|\bX}^{-1} \boldsymbol{B}_0^\top + \boldsymbol{\tilde \eta} \bSigma_{\boldsymbol{Y}|\boldsymbol{X}}^{-1}\mathbb{Y}_{\bmu}^\top \mathbb{X}_{\bmu}) \boldsymbol{C}_{\bA} \right] \\
    &\quad - \frac{1}{2} \tr \left[ \boldsymbol{D}^\top_{\bA}(\mathbb{X}_{\bmu}^\top \mathbb{X}_{\bmu} + \psi_{\boldsymbol{X}}^{(0)} \boldsymbol{I}_p) \boldsymbol{D}_{\bA} \btOmega_{\boldsymbol{X}}^{(0) -1} \right] \\
    &\quad - \frac{1}{2} \tr \left[ (\mathbb{Y}_{\bmu}^\top \mathbb{Y}_{\bmu} + \boldsymbol{\Psi}_Y)\bSigma_{Y|X}^{-1} \right] - \frac{1}{2} \tr \left[ \psi_0^{(\boldsymbol{\eta}) -1} \boldsymbol{\tilde{\eta}} \bSigma_{\bY|\bX}^{-1} \boldsymbol{\tilde{\eta}}^\top \btOmega_{\bX}^{(1) -1} \right] - \frac{1}{2} \tr\left[\boldsymbol{V}_0^{(A) -1} (\bA - \bA_0)^\top \boldsymbol{U}_0^{(A) -1} (\bA - \bA_0) \right], \label{eq:penvrepost} \tag{S19}
\end{align*}
Let MFVF to the reparameterized parameters and choose families that match the full conditional forms except $q(\bA)$: 
\begin{align*}
    q(\bmu_{\bX}, \bmu_{\bY}, \boldsymbol{\tilde{\eta}}, \bSigma_{\bY \mid \bX}, \btOmega_{\bX}^{(1)}, \btOmega_{\bX}^{(0)}, \bA) &= q(\bmu_{\bX}) q(\bmu_{\bY}) q(\boldsymbol{\tilde{\eta}}) q(\bSigma_{\bY \mid \bX}) q(\btOmega_{\bX}^{(1)}) q(\btOmega_{\bX}^{(0)}) q(\bA) \\
    q(\bmu_{\bX}) &\sim \mathcal{N}_{p}(\bmu_{\bX_q}, \bSigma_{\bX_q}) \\
    q(\bmu_{\bY}) &\sim \mathcal{N}_{r}(\bmu_{\bY_q}, \bSigma_{\bY_q}) \\
    q(\vecr (\boldsymbol{\tilde \eta})) &\sim \mathcal{N}_{mr}(\vecr (\boldsymbol{\eta}_{q}), \bSigma_q^{(\boldsymbol{\eta})}) \\
    q(\bSigma_{\bY \mid \bX}) &\sim \mathcal{IW}_{r}(\bPsi_{\bY_q}, \nu_{\bY_q}) \\
    q(\btOmega^{(1)}_{\bX}) &\sim \mathcal{IW}_{m}(\bPsi^{(1)}_{\bX_q}, \nu^{(1)}_{\bX_q}) \\
    q(\btOmega^{(0)}_{\bX}) &\sim \mathcal{IW}_{p-m}(\bPsi^{(0)}_{\bX_q}, \nu^{(0)}_{\bX_q}) \label{eq:penvMFVF} \tag{S20} 
\end{align*}
Let $\boldsymbol{C}_{\bA} = \bK + \bL\bA \in \bbR^{p \times m}$ and $\boldsymbol{D}_{\bA} = \bL - \bK\bA^\top \in \bbR^{p \times (p-m)}$ with $\bK = \begin{pmatrix}
         \boldsymbol{I}_m  \\
         \boldsymbol{0}
    \end{pmatrix} \in \bbR^{p \times m}$ and $\bL = \begin{pmatrix}
          \boldsymbol{0} \\
          \boldsymbol{I}_{p-m}
    \end{pmatrix} \in \bbR^{p \times (p-m)}$. From the reparameterized posterior distribution in~\eqref{eq:penvrepost}, we can define $\tf_{\bX}(\bA)$ as follows.
\begin{align*}
    \tf_{\bX}(\bA) \propto&  \frac{(2n+\nu_{\bX}^{(1)}+\nu_{\bX}^{(0)})}{2} \log \left| \boldsymbol{I}_{p-m} + \bA \bA^\top \right| \\
    &-\frac{1}{2} \tr \left[ \left( \boldsymbol{\eta}_q \nu_{\bY_q} \bPsi_{\bY_q}^{-1} \boldsymbol{\eta}_q^\top + \Ktr \left[ \nu_{\bY_q} \bPsi_{\bY_q}^{-1} \odot \bSigma_q^{(\boldsymbol{\eta})} \right] \right) \bA^\top \bL^\top \boldsymbol{G}_{pX} \bL \bA \right] \\
    &- \frac{1}{2} \tr \left[ \nu_{\bX_q}^{(1)} \bPsi_{\bX_q}^{(1) -1} \bA^\top \bL^\top \left( \boldsymbol{G}_{pX} + \psi^{(\boldsymbol{\eta}) -1}_0 \bB_0 \nu_{\bY_q} \bPsi_{\bY_q}^{-1} \bB_0^\top + \psi_{\bX}^{(1)} \bI_p \right) \bL \bA \right] \\
    &- \frac{1}{2} \tr \left[ \bK^\top \left( \boldsymbol{G}_{pX} + \psi_{\bX}^{(0)} \bI_p \right) \bK \bA^\top \nu_{\bX_q}^{(0)} \bPsi_{\bX_q}^{(0) -1} \bA \right] \\
    &+ \tr \left[ \left( \boldsymbol{\eta}_q \nu_{\bY_q} \bPsi_{\bY_q}^{-1} \boldsymbol{\eta}_q^\top + \Ktr \left[ \nu_{\bY_q} \bPsi_{\bY_q}^{-1} \odot \tilde{\bSigma}_q^{(\boldsymbol{\eta})} \right] \right) \bK^\top \boldsymbol{G}_{pX} \bL \bA \right] \\
    &+ \tr \left[ \nu_{\bX_q}^{(1)} \bPsi_{\bX_q}^{(1) -1} \bK^\top \left( \boldsymbol{G}_{pX} + \psi^{(\boldsymbol{\eta}) -1}_0 \bB_0 \nu_{\bY_q} \bPsi_{\bY_q}^{-1} \bB_0^\top + \psi_{\bX}^{(1)} \bI_p \right) \bL \bA \right] \\
    &- \tr \left[ \bK^\top \left( \boldsymbol{G}_{pX} + \psi_{\bX}^{(0)} \bI_p \right) \bL \nu_{\bX_q}^{(0)} \bPsi_{\bX_q}^{(0) -1} \bA \right] \\
    &+ \tr \left[ \left( \boldsymbol{\eta}_q \nu_{\bY_q} \bPsi_{\bY_q}^{-1} \bbY_{\bmu_q}^\top \bbX_{\bmu_q} + \psi^{(\boldsymbol{\eta}) -1}_0 \nu_{\bX_q}^{(1)} \bPsi_{\bX_q}^{(1) -1} \boldsymbol{\eta}_q \nu_{\bY_q} \bPsi_{\bY_q}^{-1} \bB_0^\top \right) \bL \bA \right] \\
    &- \frac{1}{2} \tr \left[ \bV^{(\bA) -1}_0 (\bA - \bA_0)^\top \bU^{(\bA) -1}_0 (\bA - \bA_0) \right], \label{eq:tfpenv} \tag{S21}
\end{align*}
where $\bbX_{\bmu_q} = (\bbX - \boldsymbol{1}_n \bmu_{\bX_q}^\top)$, $\bbY_{\bmu_q} = (\bbY - \boldsymbol{1}_n \bmu_{\bY_q}^\top)$, $\boldsymbol{G}_{pX} = \bbX_{\bmu_q}^\top \bbX_{\bmu_q} + n \bSigma_{\bX_q}$, and $\Ktr \left[\cdot \odot \cdot \right]$ denote the operator defined in Section~\ref{sec:B.2}.

The process of deriving the CALVI update for the Bayesian predictor envelope model is similar to that of the Bayesian response envelope model. Therefore, only the results are listed here. 

\subsubsection*{Update for $q(\vecr(\bA))$.}
\begin{align*}
        q(\vecr(\bA)) &\propto \calN_{(p-m)m}\left( \vecr(\hbA), - \tf_{\bX}^{\prime \prime}(\vecr(\hbA))^{-1} \right), \label{eq:qA} \tag{S22}\\
        \vecr(\hbA) &= \argmax_{\bA} \tf_{\bX}(\bA),\\
        H(\vecr(\hbA)) &= -\big[\nabla^2_{\vecr(\bA)}\tf_{\bX}(\vecr(\hbA))\big]^{-1}, 
\end{align*}
where $\nabla^2_{\vecr(\bA)}\tf_{\bX}(\vecr(\hbA))$ denotes the Hessian of $\tf_{\bX}(\bA)$ with respect to $\vecr(\bA)$, evaluated at $\vecr(\hbA)$:
\begin{align*}
    \nabla^2_{\vecr(\bA)}\tf_{\bX}(\vecr(\hbA)) &= (2n + \nu_{\bX}^{(1)} + \nu_{\bX}^{(0)}) \Bigg[ (\bI_m - \hbA^\top \bJ_0(\hbA)^{-1} \hbA) \otimes \bJ_0(\hbA)^{-1} \\
    & \qquad \qquad \qquad \qquad \qquad - \Big( \big( \bJ_0(\hbA)^{-1} \hbA \big)^\top \otimes \big( \bJ_0(\hbA)^{-1} \hbA \big) \Big) \mathcal{K}_{(p-m),m} \Bigg] \\
    & \quad - \left( \boldsymbol{\eta}_q \nu_{\bY_q} \bPsi_{\bY_q}^{-1} \boldsymbol{\eta}_q^\top + \Ktr \left[ \nu_{\bY_q} \bPsi_{\bY_q}^{-1} \odot \tilde{\bSigma}_q^{(\boldsymbol{\eta})} \right] \right) \otimes \bL^\top \boldsymbol{G}_{pX} \bL \\
    & \quad - \bK^\top \left( \boldsymbol{G}_{pX} + \psi_{\bX}^{(0)} \bI_p \right) \bK \otimes \nu_{\bX_q}^{(0)} \bPsi_{\bX_q}^{(0) -1} \\
    & \quad - \bV_0^{(\bA) -1} \otimes \bU_0^{(\bA) -1},
\end{align*}
with $\bJ_0(\hbA) := \bI_{p-m} + \hbA \hbA^\top$.

\subsubsection*{Update for $q(\bSigma_{\bY \mid \bX})$.}
\begin{align*}
    \log q(\bSigma_{\bY \mid \bX}) & \propto  
    \text{const.} - \frac{\nu_{\bY_q} + r + 1}{2} \log |\bSigma_{\bY \mid \bX}| 
    - \frac{1}{2} \tr \left[ \bPsi_{\bY_q} \bSigma_{\bY \mid \bX}^{-1} \right] \\
    & \propto \mathcal{IW}_{r}(\bPsi_{\bY_q}, \nu_{\bY_q}), \label{eq:qSigma} \tag{S23}
\end{align*}
where 
\begin{align*}
    & \nu_{\bY_q} = n + m + \nu_{\bY},\\
    & \bPsi_{\bY_q} = \boldsymbol{G}_{pY} - 2 \bbY_{\bmu_q}^\top \bbX_{\bmu_q} \bC_{\hbA}\boldsymbol{\eta}_q + \boldsymbol{\eta}_q^\top S_3 \boldsymbol{\eta}_q + \Ktr \left[ S_3 \odot \bSigma_q^{(\boldsymbol{\eta})} \right] + \bPsi_{\bY} \\
    & \qquad \quad + \frac{\nu_{\bX_q}^{(1)}}{\psi_0^{(\boldsymbol{\eta})}} \Bigg( \left(\bC_{\hbA}^\top \bB_0 - \boldsymbol{\eta}_q \right)^\top \bPsi_{\bX_q}^{(1) -1} \left(\bC_{\hbA}^\top \bB_0 - \boldsymbol{\eta}_q \right)  \\
    & \qquad \qquad \qquad \qquad + \bB_0^\top \Ktr \left[\bPsi_{\bX_q}^{(1) -1} \odot \tilde{H}(\vecr(\hbA)) \right] \bB_0 + \Ktr \left[\bPsi_{\bX_q}^{(1) -1} \odot \bSigma^{(\boldsymbol{\eta})}_{q} \right] \Bigg), \\
    & \boldsymbol{G}_{pY} = \bbY_{\bmu_q}^\top \bbY_{\bmu_q} + n \bSigma_{\bY_q}, \\
    & \tilde{H}(\vecr(\hbA)) = \mathcal{K}_{(p-m),m} H(\vecr(\hbA)) \mathcal{K}_{m,(p-m)} ,\\ 
    & S_3 = \bC_{\hbA}^\top \boldsymbol{G}_{pX} \bC_{\hbA} + \Ktr \left[ \bL^\top \boldsymbol{G}_{pX} \bL \odot H(\vecr(\hbA)) \right]
\end{align*}

\subsubsection*{Update for $q(\btOmega^{(1)}_{\bX})$.}
\begin{align*}
    \log q(\btOmega^{(1)}_{\bX}) & \propto  
    \text{const.} - \frac{\nu^{(1)}_{\bX_q} + m + 1}{2} \log |\btOmega^{(1)}_{\bX}| 
    - \frac{1}{2} \tr \left[ \bPsi^{(1)}_{\bX_q} \btOmega^{(1) -1}_{\bX} \right] \\
    & \propto \mathcal{IW}_{m}(\bPsi^{(1)}_{\bX_q}, \nu^{(1)}_{\bX_q}), \label{eq:qOmega} \tag{S24}
\end{align*}
where 
\begin{align*}
    & \nu^{(1)}_{\bX_q} = n + \nu^{(1)}_{\bX} + r,\\
    & \bPsi^{(1)}_{\bX_q} = \bC_{\hbA}^\top S_4 \bC_{\hbA} + \Ktr \left[ \bL^\top S_4 \bL \odot H(\vecr(\hbA)) \right] + \boldsymbol{\eta}_q \psi^{(\boldsymbol{\eta}) -1}_0 \nu_{\bY_q} \bPsi_{\bY_q}^{-1} \boldsymbol{\eta}_q^\top + \Ktr \left[ \psi^{(\boldsymbol{\eta}) -1}_0 \nu_{\bY_q} \bPsi_{\bY_q}^{-1} \odot \tilde{\bSigma}^{(\boldsymbol{\eta})}_q \right], \\
    & \tilde{\bSigma}_q^{(\boldsymbol{\eta})} = \mathcal{K}_{m,r} \bSigma_q^{(\boldsymbol{\eta})} \mathcal{K}_{r,m} ,\\ 
    & S_4 = \boldsymbol{G}_{pX} + \psi_0^{(\boldsymbol{\eta}) -1} \bB_0 \nu_{\bY_q} \bPsi_{\bY_q}^{-1} \bB_0^\top + \psi^{(1)}_{\bX} \bI_p.
\end{align*}

\subsubsection*{Update for $q(\btOmega^{(0)}_{\bX})$.}
\begin{align*}
    \log q(\btOmega^{(0)}_{\bX}) & \propto  
    \text{const.} - \frac{\nu^{(0)}_{\bX_q} + p - m + 1}{2} \log |\btOmega^{(0)}_{\bX}| 
    - \frac{1}{2} \tr \left[ \bPsi^{(0)}_{\bX_q} \btOmega^{(0) -1}_{\bX} \right] \\
    & \propto \mathcal{IW}_{p-m}(\bPsi^{(0)}_{\bX_q}, \nu^{(0)}_{\bX_q}), \label{eq:qOmega0} \tag{S25}
\end{align*}
where 
\begin{align*}
    & \nu^{(0)}_{\bX_q} = n + \nu^{(0)}_{\bX},\\
    & \bPsi^{(0)}_{\bX_q} = \bD_{\hbA}^\top S_5 \bD_{\hbA} + \Ktr \left[ \bK^\top S_5 \bK \odot \tilde{H}(\vecr(\hbA)) \right], \\
    & S_5 = \boldsymbol{G}_{pX} + \psi^{(0)}_{\bX} \bI_p.
\end{align*}

\subsubsection*{Update for $q(\vecr\boldsymbol{(\tilde{\eta}}))$.} 
\begin{align*}
    \log q(\vecr (\boldsymbol{\tilde{\eta}})) &\propto \text{const.} - \frac{1}{2} \tr \left[ (\vecr (\boldsymbol{\tilde{\eta}}) - \vecr (\boldsymbol{\eta}_q))^\top \bSigma_{q}^{(\boldsymbol{\eta})-1} (\vecr (\boldsymbol{\tilde{\eta}}) - \vecr (\boldsymbol{\eta}_q)) \right], \\
    &\propto \mathcal{N}_{mr}(\vecr (\boldsymbol{\eta}_q), \bSigma_{q}^{(\boldsymbol{\eta})}), \label{eq:qeta} \tag{S26}
\end{align*}
where
\begin{align*}
    \vecr (\boldsymbol{\eta}_q) &= \bSigma_{q}^{(\boldsymbol{\eta})} \vecr \left( (\nu_{\bY_q} \bPsi_{\bY_q}^{-1} \bbY_{\bmu_q}^\top \bbX_{\bmu_q} \bC_{\hbA} + \nu_{\bY_q} \bPsi_{\bY_q}^{-1} \bB_0^\top \bC_{\hbA}  \psi^{(\boldsymbol{\eta}) -1}_0 \nu_{\bX_q}^{(1)} \bPsi_{\bX_q}^{(1) -1})^\top \right), \\ 
    \bSigma_{q}^{(\boldsymbol{\eta})} &= \left( \nu_{\bY_q} \bPsi_{\bY_q}^{-1} \otimes ( S_3 + \psi^{(\boldsymbol{\eta}) -1}_0 \nu_{\bX_q}^{(1)} \bPsi_{\bX_q}^{(1) -1}) \right)^{-1}
\end{align*}

\subsubsection*{Update for $q(\bmu_{\bX})$.}
\begin{align*}
    \log q(\bmu_{\bX}) &\propto \text{const.} - \frac{1}{2} \tr \left[ (\bmu_{\bX} - \bmu_{\bX_q})^\top \bSigma_{\bX_q}^{-1} (\bmu_{\bX} - \bmu_{\bX_q}) \right], \\
    &\propto \mathcal{N}_p(\bmu_{\bX_q}, \bSigma_{\bX_q}), \label{eq:qmux} \tag{S27}
\end{align*}
where
\begin{align*}
    \bmu_{\bX_q} &= \bar{\bX}, \\ 
    \bSigma_{\bX_q} &= \left( \bC_{\hbA} S_6 \bC_{\hbA}^\top + \bL \Ktr \left[ S_6 \odot \tilde{H}(\vecr(\hbA)) \right] \bL^\top + \bD_{\hbA} S_7 \bD_{\hbA}^\top + \bK \Ktr \left[ S_7 \odot H(\vecr(\hbA)) \right] \bK^\top \right)^{-1}, \\
    S_6 &= n \nu_{\bX_q}^{(1)} \bPsi_{\bX_q}^{(1) -1} + \boldsymbol{\eta}_q \nu_{\bY_q} \bPsi_{\bY_q}^{-1} \boldsymbol{\eta}_q^\top + \Ktr \left[ \nu_{\bY_q} \bPsi_{\bY_q}^{-1} \odot \tilde{\bSigma}^{(\boldsymbol{\eta})}_{q} \right], \\
    S_7 &= n \nu_{\bX_q}^{(0)} \bPsi_{\bX_q}^{(0) -1}. 
\end{align*}

\subsubsection*{Update for $q(\bmu_{\bY})$.}
\begin{align*}
    \log q^*(\bmu_{\bY}) &\propto \text{const.} - \frac{1}{2} \tr \left[ (\bmu_{\bY} - \bmu_{\bY_q})^\top \bSigma_{\bY_q}^{-1} (\bmu_{\bY} - \bmu_{\bY_q}) \right], \\
    &\propto \mathcal{N}_r(\bmu_{\bY_q}, \bSigma_{\bY_q}), \label{eq:qmuy} \tag{S28}
\end{align*}
where
\begin{align*}
    \bmu_{\bY_q} = \bar{\bY}, \quad 
    \bSigma_{\bY_q} = \frac{\bPsi_{\bY_q}}{n \nu_{\bY_q}} .
\end{align*}
    
\subsubsection*{Approximated ELBO for convergence criterion}
\begin{align*}
    \calL(q) = &\mathbb{E}_{q(\cdot)}\left[\log \pi(\bmu_{\bX}, \bmu_{\bY}, \bSigma_{\bY \mid \bX}, \boldsymbol{\tilde{\eta}}, \btOmega, \btOmega_0, \bA \mid \mathbb{X}, \mathbb{Y}) - \log q(\bmu_{\bX}, \bmu_{\bY}, \bSigma_{\bY \mid \bX}, \boldsymbol{\tilde{\eta}}, \btOmega, \btOmega_0, \bA) \right] \\
    \approx & \boldsymbol{\tilde{C}} + \tf_{\bX}(\hbA) - \frac{(p-m)m}{2} - \mathbb{E}_{q(\cdot)}\left[\log q(\bmu_{\bX}, \bmu_{\bY}, \bSigma_{\bY \mid \bX}, \boldsymbol{\tilde{\eta}}, \btOmega, \btOmega_0, \bA) \right]. \label{eq:penvELBO} \tag{S29}
\end{align*}      

Algorithm~\ref{penvCALVI} summarizes the CALVI procedure for the predictor envelope model. To obtain the optimal variational distribution, at iteration $t$ and given $q^{t}_{-i}$, we update $q^{t+1}_i$ cycling over $i$ and repeating the sweep until convergence.

\begin{algorithm}
\caption{Coordinate ascent Laplace variational inference for the predictor envelope model}
\label{penvCALVI}
\begin{algorithmic}[1]
    \State Initialize variational factors
    $q^{(1)}(\bmu_{\bX})$, $q^{(1)}(\bmu_{\bY})$,
    $q^{(1)}(\boldsymbol{\tilde{\eta}})$,
    $q^{(1)}(\bSigma_{\bY\mid\bX})$,
    $q^{(1)}(\btOmega^{(1)}_{\bX})$,
    $q^{(1)}(\btOmega^{(0)}_{\bX})$,
    and $q^{(1)}(\vecr(\bA))$
    
    \State Set tolerance $\epsilon$, maximum iterations $T$
    \State $t \gets 1$
    
    \While{$t < T$}
        \State Compute $\tilde{\mathcal L}(q)^{(t)}$ using~\eqref{eq:penvELBO}
        
        \State Update $q^{(t+1)}(\vecr(\bA))$ using the Laplace step~\eqref{eq:qA}
        
        \State Update $q^{(t+1)}(\bSigma_{\bY\mid\bX})$ using closed-form CAVI updates~\eqref{eq:qSigma}
        
        \State Update $q^{(t+1)}(\btOmega^{(1)}_{\bX})$ using closed-form CAVI updates~\eqref{eq:qOmega}
        
        \State Update $q^{(t+1)}(\btOmega^{(0)}_{\bX})$ using closed-form CAVI updates~\eqref{eq:qOmega0}
        
        \State Update $q^{(t+1)}(\boldsymbol{\tilde{\eta}})$ using closed-form CAVI updates~\eqref{eq:qeta}
        
        \State Update $q^{(t+1)}(\bmu_{\bX})$ using closed-form CAVI updates~\eqref{eq:qmux}
        
        \State Update $q^{(t+1)}(\bmu_{\bY})$ using closed-form CAVI updates~\eqref{eq:qmuy}
        
        \State Compute $\tilde{\mathcal L}(q)^{(t+1)}$ using~\eqref{eq:penvELBO}
        
        \If{$\big|\tilde{\mathcal L}(q)^{(t+1)} - \tilde{\mathcal L}(q)^{(t)}\big|
             < \epsilon\,\big|\tilde{\mathcal L}(q)^{(t+1)}\big|$}
            \State \textbf{break}
        \EndIf
        
        \State $t \gets t + 1$
    \EndWhile
    
    \State \Return
    $q^{(t)}(\bmu_{\bX})$,
    $q^{(t)}(\bmu_{\bY})$,
    $q^{(t)}(\boldsymbol{\tilde{\eta}})$,
    $q^{(t)}(\bSigma_{\bY\mid\bX})$,
    $q^{(t)}(\btOmega^{(1)}_{\bX})$,
    $q^{(t)}(\btOmega^{(0)}_{\bX})$,
    $q^{(t)}(\vecr(\bA))$,
    and $\{\tilde{\mathcal L}(q)^{(1)},\dots,\tilde{\mathcal L}(q)^{(t)}\}$
\end{algorithmic}
\end{algorithm}

\subsection{Simulation studies of predictor envelope model}

We conducted simulation studies on the Bayesian predictor envelope model. 
The performance of CALVI is assessed in terms of estimation accuracy of $\bbeta$. 
We compare CALVI with the Metropolis--Hastings within Gibbs method with likelihood-driven proposal \citep{chakraborty2024comprehensive} (hereafter, MH--Gibbs).

Throughout, the superscript ${*}$ denotes the true values of the parameters. 
We generated 100 replicated data sets according to the predictor envelope model in~\eqref{eq:penv}, with $n=100, 200, 500, 1000$, $r=3$, $p=15$, and true envelope dimension $m^{*}\in\{2,5\}$. 
The elements of $\bmu_{\bX}^{*}$ and $\bmu_{\bY}^{*}$ were independently sampled from Uniform$(-10,10)$, while the elements of $\boldsymbol{\eta}^{*}$ and $\bA^{*}$ were independently drawn from Uniform$(5,10)$ and Uniform$(0,5)$, respectively. 
The matrices $\bOmega_{\bX}^{(1) *}$, $\bOmega_{\bX}^{(0) *}$, and $\bSigma_{\bY \mid \bX}^{*}$ were simulated independently from $5\mathcal{IW}_m(m+2,5\bI_m)$, $\mathcal{IW}_{p-m}(p-m+2,0.1\bI_{p-m})$, and $\mathcal{IW}_r(r+1,5\mathbf{I}_r)$.

We adopted vague priors for all parameters, setting zero means and covariance matrices equal to $10^6 \mathbf{I}$ for multivariate and matrix normal priors, and $\mathcal{IW}_p(p,10^{-6}\mathbf{I}_p)$ for inverse-Wishart priors.

Table~\ref{S-table1} reports the BIC-based posterior probabilities of the envelope dimension $m$ under CALVI and MH--Gibbs. When $m^{*}=2$, CALVI concentrates more rapidly on the true dimension at small sample sizes (0.901 at $n=100$), whereas MH--Gibbs places comparatively more mass on neighboring dimensions, particularly at $m=3$. As $n$ increases, both methods increasingly concentrate on the true dimension, with CALVI exhibiting slightly faster concentration at moderate sample sizes.

When $m^{*}=5$, both methods assign the majority of posterior mass to $m=5$ across all sample sizes. MH--Gibbs tends to concentrate more sharply at smaller $n$, while CALVI also achieves near-degenerate concentration as $n$ grows.

Table~\ref{S-table2} summarizes the MSE of $\hat{\bbeta}$ obtained via BMA. For $m^{*}=2$, both methods exhibit decreasing MSE as $n$ increases. MH--Gibbs attains smaller MSE at small and moderate sample sizes, although the gap narrows as $n$ becomes large. For $m^{*}=5$, MH--Gibbs generally achieves lower MSE across sample sizes, while CALVI displays stable and monotone improvement as $n$ increases.

Table~\ref{S-table3} compares the total computational time. 
Across all configurations, CALVI is substantially faster than MH--Gibbs, requiring only a small fraction of the runtime (typically below 4\% and decreasing further as $n$ increases). This computational advantage becomes particularly pronounced at larger sample sizes.

\begin{table*}[!t]
    \centering
    \caption{Mean posterior probabilities for the envelope dimension $m$ by method, with $r=3$, $p=15$, and $m^{*}\in\{2,5\}$. Results are averaged over 100 replications.}
    \label{S-table1}
    \setlength{\tabcolsep}{4pt} 
        \begin{tabular}{@{} l | cccc | cccc @{}}
            \toprule
            & \multicolumn{4}{c|}{CALVI}
            & \multicolumn{4}{c}{MH--Gibbs} \\ 
            \midrule
            $n$ & $m \leq 1$ & $m = 2$ & $m = 3$ & $4 \leq m$ 
            & $m \leq 1$ & $m = 2$ & $m = 3$ & $4 \leq m$ \\
            \midrule
            100  & 0.000 & 0.901 & 0.095 & 0.004
            & 0.000 & 0.714 & 0.267 & 0.019 \\
            200  & 0.000 & 0.975 & 0.024 & 0.000
            & 0.000 & 0.920 & 0.076 & 0.004 \\
            500  & 0.000 & 0.979 & 0.021 & 0.000
            & 0.000 & 0.889 & 0.111 & 0.000 \\
            1000 & 0.000 & 0.997 & 0.003 & 0.000
            & 0.000 & 0.975 & 0.025 & 0.000 \\
            \midrule
            \multicolumn{9}{c}{(a) The true $m$ is 2.} \\
            \midrule
            $n$ & $m \leq 4$ & $m = 5$ & $m = 6$ & $7 \leq m$ 
            & $m \leq 4$ & $m = 5$ & $m = 6$ & $7 \leq m$ \\
            \midrule
            100  & 0.000 & 0.860 & 0.123 & 0.017
            & 0.000 & 0.919 & 0.080 & 0.001 \\
            200  & 0.000 & 0.914 & 0.083 & 0.004
            & 0.000 & 0.980 & 0.019 & 0.001 \\
            500  & 0.000 & 0.960 & 0.039 & 0.001
            & 0.000 & 0.937 & 0.06 & 0.000 \\
            1000 & 0.000 & 0.971 & 0.029 & 0.000
            & 0.000 & 0.992 & 0.008 & 0.000 \\
            \midrule
            \multicolumn{9}{c}{(b) The true $m$ is 5.} \\
            \bottomrule
        \end{tabular}
\end{table*}

\begin{table*}[!t]
    \centering
    \caption{MSE for $\bbeta$ (sum over coordinates) by method—BMA via CALVI and MH--Gibbs—with $r=3$, $p=15$, and $m^{*}\in\{2,5\}$. Entries are mean (sd) over 100 replications.}
    \label{S-table2}
    \setlength{\tabcolsep}{6pt}
    \begin{tabular}{@{}lcccc@{}}
        \toprule
        Method & $n = 100$ & $n = 200$ & $n = 500$ & $n = 1000$ \\
        \midrule
        CALVI & 2.578 (6.708) & 0.327 (1.157) & 0.106 (0.429) & 0.005 (0.012) \\
        MH--Gibbs  & 1.003 (2.920) & 0.102 (0.506) & 0.030 (0.085) & 0.005 (0.015) \\
        \midrule
        \multicolumn{5}{c}{(a) The true $m$ is 2.} \\
        \midrule
        CALVI & 4.447 (18.556) & 1.200 (5.649) & 0.164 (0.924) & 0.026 (0.046) \\
        MH--Gibbs & 0.344 (0.858) & 0.078 (0.076) & 0.385 (0.027) & 0.028 (0.125) \\
        \midrule
        \multicolumn{5}{c}{(b) The true $m$ is 5.} \\
        \bottomrule
    \end{tabular}
\end{table*}

\begin{table}[!t]
\centering
\caption{Mean total CPU time (seconds) across $m=1,\ldots,15$ over 100 replications ($r=3$, $p=15$; $m^* \in \{2,5\}$).}
\label{S-table3}
\setlength{\tabcolsep}{6pt}
\begin{tabular}{@{}lccc@{}}
\toprule
 & CALVI & MH--Gibbs & ratio (\%) \\
\midrule
$n=100$  & 91.7 & 2705.1 & 3.39 \\
$n=200$  & 100.3 & 2836.0 & 3.54 \\
$n=500$  & 23.5 & 3208.2 & 0.73 \\
$n=1000$ & 8.1  & 3752.1 & 0.22 \\
\midrule
\multicolumn{4}{c}{(a) True $u=2$} \\
\midrule
$n=100$  & 54.5 & 2695.9 & 2.02 \\
$n=200$  & 37.8 & 2829.3 & 1.34 \\
$n=500$  & 8.0  & 3221.7 & 0.25 \\
$n=1000$ & 7.5  & 3760.9 & 0.20 \\
\midrule
\multicolumn{4}{c}{(b) True $u=5$}\\
\bottomrule
\end{tabular}
\end{table}

\newpage
\section{Proof of Theorem~\ref{main-theorem}}
\label{sec:D}

This section provides the proof of Theorem~\ref{main-theorem} in the main text. Rather than proving the asymptotic statement directly, we establish a stronger, nonasymptotic bound on the total variation (TV) distance between the Laplace-based coordinate update and the exact CAVI update, from which the asymptotic convergence in Theorem~\ref{main-theorem} follows as a direct corollary.

For clarity and completeness, we first restate the result with all constants
and auxiliary quantities made explicit.
\begin{theorem}
\label{theorem1}
Suppose that Assumptions~\ref{assum1}--\ref{assum5} and the associated notation hold. Then
\[
d_{\mathrm{TV}}\!\left( \overline{q_n^{L}},\, \overline{q_n^{VI}} \right) 
\;\le\; C_1 n^{-1/2} \;+\; 2 e^{\hat{\calT}_n}
\;+\; C_2\, n^{d_1/2} e^{-n \hat{\kappa}},
\]
where
\[
\begin{aligned}
C_1 &= 
\frac{\sqrt{3}\,\tr\!\big[H_f(\hat{\btheta}_{1_n})^{-1}\big]\;\hat{M}_1}{
4\sqrt{\big(\lambda_{\min}(H_f(\hat{\btheta}_{1_n})) - \hat{\delta}\hat{M}_1\big)\,\big(1 - e^{\hat{\calT}_n}\big)}} , 
\qquad
C_2 = \frac{2\,\big|\det\!\big(J(\bar{\btheta}_{1_n})^{-1}\big)\big|^{-1/2}\,\bar{M}_2}{
(2\pi)^{d_1/2}\,\big(1 - e^{\bar{\calT}_J}\big)},\\[3pt]
H_f(\hat{\btheta}_{1_n}) &:= -\frac{\tf_n^{\dagger\,\prime\prime}(\hat{\btheta}_{1_n})}{n}, 
\qquad
\hat{\calT}_n := -\frac{1}{2}\Big(\hat{\delta}\sqrt{n} - \sqrt{\tr\big[H_f(\hat{\btheta}_{1_n})^{-1}\big]}\Big)^{\!2}\,\big\|H_f(\hat{\btheta}_{1_n})^{-1}\big\|_{\mathrm{op}}^{-1},\\[3pt]
H_{\ell}(\bar{\btheta}_{1_n}) &:= -\frac{\ell_n^{\prime\prime}(\bar{\btheta}_{1_n})}{n},
\qquad
J(\bar{\btheta}_{1_n}) := H_{\ell}(\bar{\btheta}_{1_n}) + \frac{\bar{M}_1 \bar{\delta}}{3}\,\bI_{d_1},\\[3pt]
\bar{\calT}_J &:= -\frac{1}{2}\Big(\bar{\delta}\sqrt{n} - \sqrt{\tr\big[J(\bar{\btheta}_{1_n})^{-1}\big]}\Big)^{\!2}\,\big\|J(\bar{\btheta}_{1_n})^{-1}\big\|_{\mathrm{op}}^{-1}.
\end{aligned}
\]
\end{theorem}
Theorem~\ref{theorem1} extends Theorem~17 of \citet{kasprzak2025laplace} to the framework of VI. Proof details from \citet{kasprzak2025laplace}, Appendices D and E, carry over; for completeness we reproduce the argument in our setting. We adopt the notation of Section~\ref{sec5}. Any additional notation and the auxiliary results needed for the proof (Lemma~\ref{std normal lemma}--Proposition~\ref{Bakry-Emery criterion}) are collected below.

\begin{lemma}
    \label{std normal lemma}
    Let $Z \sim \calN_d(0, \Sigma)$, $\|\cdot\|$ denotes the Euclidean norm (for vectors) and $\|\cdot\|_{\mathrm{op}}$ the operator norm (for matrices). Then, for any $\epsilon > \sqrt{\tr[\Sigma]}$,
    \begin{align*}
        \int_{\|Z\| > \epsilon} \frac{| \det (\Sigma)|^{-1/2} e^{-\tfrac{1}{2}z^\top \Sigma^{-1} z}}{(2\pi)^{d/2}} \ dz \leq \exp\left( -\frac{1}{2}\left(\epsilon - \sqrt{\tr[\Sigma]} \right)^2 \|\Sigma\|^{-1}_{op} \right).
    \end{align*}
\end{lemma}
\begin{proof}
    \begin{align*}
        \bbP \Big[ \|Z\| \geq \epsilon \Big] &= \bbP \Big[ \|Z \| \geq \sqrt{\tr(\Sigma)} + \sqrt{2\|\Sigma\|_{op} t} \Big] \\
        &= \bbP \Big[ \|Z \|^2 \geq \tr(\Sigma) + 2\sqrt{2\|\Sigma\|_{op} \tr(\Sigma) t} + 2\|\Sigma\|_{op} t \Big] \\
        &\leq \bbP \Big[ \|Z \|^2 \geq \tr(\Sigma) + 2\sqrt{ \tr(\Sigma^2) t} + 2\|\Sigma\|_{op} t \Big] \\
        &= \bbP \Big[ \|Z \|^2 - \tr(\Sigma) \geq 2\sqrt{ \tr(\Sigma^2) t} + 2\|\Sigma\|_{op} t \Big] \\
        &\leq e^{-t} = e^{-\frac{1}{2} \big(\epsilon - \sqrt{\tr(\Sigma)} \big)^2 \|\Sigma\|^{-1}_{op} } \tag{\cite{10.1214/ECP.v17-2079}, Proposition 1}.
    \end{align*}
\end{proof}

\begin{lemma}
    \label{std normal lemma 2}
    Let $z \sim \calN_d(0, \Sigma)$, $\|\cdot\|$ denotes the Euclidean norm (for vectors) and $\|\cdot\|_{\mathrm{op}}$ the operator norm (for matrices). Then,
    \begin{align*}
        \bbE[ \|z\|^4] \leq 3 \tr \left[ \Sigma \right]^2.
    \end{align*}
\end{lemma}
\begin{proof}
Let $z = U\Lambda^{1/2}\upsilon$ with $\upsilon\sim\calN_d(0,I_d)$ and $\Lambda=\mathrm{diag}(\lambda_1,\dots,\lambda_d)$.
Then
\[
\|z\|^2=\sum_{i=1}^d \lambda_i \upsilon_i^2,
\qquad
\|z\|^4=\Big(\sum_{i=1}^d \lambda_i \upsilon_i^2\Big)^2.
\]
Using $\bbE[\upsilon_i^4]=3$ and $\bbE[\upsilon_i^2\upsilon_j^2]=1$ for $i\neq j$, we obtain
\[
\bbE[\|z\|^4]
= \sum_{i=1}^d \lambda_i^2 \bbE[\upsilon_i^4] + 2\sum_{i<j}\lambda_i\lambda_j \bbE[\upsilon_i^2\upsilon_j^2]
= 3\sum_{i=1}^d \lambda_i^2 + 2\sum_{i<j}\lambda_i\lambda_j
= (\tr\Sigma)^2 + 2\,\tr(\Sigma^2).
\]
Since $\tr(\Sigma^2)\le (\tr\Sigma)^2$, it follows that $\bbE[\|z\|^4]\le 3(\tr\Sigma)^2$.
\end{proof}

\begin{proposition}[Pinsker's inequality, (\cite{kasprzak2025laplace}, Proposition D.3)]
    \label{Pinsker's inequality}
    For any two probability measures $\nu$ and $\omega$ such that $\nu \ll \omega$, we have that 
    \begin{align*}
        d_{TV}(\nu, \omega) \leq \sqrt{\frac{1}{2} \text{KL}(\nu || \omega)}.
    \end{align*}
\end{proposition}

\textit{Fisher divergence} (or relative Fisher information in Section 2.2 of \cite{vempala2019rapid}) is a measure of the difference between two probability distributions in terms of their gradients as follows.
\begin{align*}
    \text{Fisher}(p \mid\mid q) = \int p(x) \left\| \left(\log \frac{p(x)}{q(x)}\right)^{'} \right\|^2 \ dx.
\end{align*}

\begin{definition}[Log-Sobolev inequality (LSI), (\cite{vempala2019rapid}, Section 2.2), (\cite{kasprzak2025laplace}, Definition D.2)]
    \label{LSI definition}
    Let $\omega$ be a probability measure on $\Omega \subset \bbR^{d}$ and let $\kappa > 0$. We say that $\omega$ satisfies the log-Sobolev inequality with constant $\kappa$ (LSI($\kappa$)) if for any probability measure $\nu$ on $\Omega$, 
    \begin{align*}
        \text{KL}(\nu || \omega) \leq \frac{1}{2\kappa} \text{Fisher}(\nu || \omega).
    \end{align*}
\end{definition}

\begin{proposition}[Bakry-Émery log-Sobolev criterion (\cite{kolesnikov2016riemannian}, Theorem 2.1), (\cite{kasprzak2022good}, Proposition D.1)]
    \label{Bakry-Emery criterion}
    Let $\Omega \subset \bbR^{d}$ be convex, let $f: \Omega \rightarrow \bbR$ and consider a probability measure $\omega (dx) = Z_{\omega}^{-1} e^{-f(x)} \mathbb{1}_{\Omega}(x)\ dx$. Assume that $f^{''}(x) \succeq \kappa \bI_{d}$ for some $\kappa > 0$. Then $\omega$ satisfies $LSI(\kappa)$ with $\kappa > 0$.
\end{proposition}

The proof proceeds by decomposing the TV distance into a local Laplace
approximation error and a tail contribution, and bounding each term separately using Assumptions~\ref{assum1}--\ref{assum5}.
\begin{proof}[Proof of Theorem~\ref{theorem1}]
    For any measurable set $\mathcal{A} \subset \mathbb{R}^{d_1}$, let $g=\mathbb{1}_{\mathcal{A}}$.
Let
\[
Z \sim \mathcal{N}\big(0,\, H_f(\hat{\btheta}_{1_n})^{-1}\big),
\]
and let $X \sim \overline{\Pi}_n$, where $\overline{\Pi}_n$ denotes the law of
\[
\sqrt{n}\,(\btheta_1-\hat{\btheta}_{1_n})
\quad\text{under}\quad
\btheta_1 \sim \Pi_n^\dagger.
\]
 Define
    \begin{align*}
        D(g)\;:=\;\bigl|\;\mathbb{E}[g(Z)]-\mathbb{E}[g(X)]\;\bigr|.
    \end{align*}  
    Then the TV distance satisfies
    \begin{align*}
        d_{\mathrm{TV}}\!\left(\mathcal{N}\!\big(0, H_f(\hat{\btheta}_{1_n})^{-1}\big), \,\overline{\Pi}_n(X)\right)
        \;=\;\sup_{A\in\mathcal{B}(\mathbb{R}^d)} 
        \bigl|\;\mathbb{E}[g(Z)]-\mathbb{E}[g(X)]\;\bigr|
        \;=\;\sup_{g} D(g).
    \end{align*}    
    where $\mathcal{B}(\mathbb{R}^d)$ denotes the Borel $\sigma$-algebra on $\mathbb{R}^d$.
    
    Now we decompose $D(g)$ into the sum of inner and outer integrals with respect to $\hat{\delta} \sqrt{n}$.
    
\subsection{\texorpdfstring{Decomposition of $D_g$}{Decomposition of D_g}}
\begin{align*}
    D_g :=& \left| \bbE_{Z \sim \calN(0, H_f(\hat{\btheta}_{1_n})^{-1})} \left[ g(Z_n) \right] - \bbE_{\overline{\Pi}_n(\sqrt{n}(\btheta_1 - \hat{\btheta}_{1_n}))} \left[ g\left( \sqrt{n} ({\btheta}_1 - \hat{\btheta}_{1_n} ) \right) \right]  \right| \\
    =&  \left| \int_{\bbR^{d_1}} g(u) \frac{|\text{det}(H_f(\hat{\btheta}_{1_n})^{-1})|^{-1/2} e^{-u^\top H_f(\hat{\btheta}_{1_n})u/2} }{(2\pi)^{d_{1}/2}} \ du - \int_{\bbR^{d_1}} g(u) n^{-d_{1}/2}\overline{\Pi}_n(n^{-1/2}u + \hat{\btheta}_{1_n}) \ du \right| \\
    =&  \left| \int_{\bbR^{d_1}} g(u) \frac{|\text{det}(H_f(\hat{\btheta}_{1_n})^{-1})|^{-1/2} e^{-u^\top H_f(\hat{\btheta}_{1_n})u/2} }{(2\pi)^{d_{1}/2}} \ du \right. \\
    &\left. - n^{-d_{1}/2} \int_{\|u\| > \hat{\delta} \sqrt{n}} g(u) \overline{\Pi}_n(n^{-1/2}u + \hat{\btheta}_{1_n}) \ du - n^{-d_{1}/2}\int_{\|u\| \leq \hat{\delta} \sqrt{n}} g(u) \overline{\Pi}_n(n^{-1/2}u + \hat{\btheta}_{1_n}) \ du \right| \\
    =&  \left| \int_{\bbR^{d_1}} g(u) \frac{|\text{det}(H_f(\hat{\btheta}_{1_n})^{-1})|^{-1/2} e^{-u^\top H_f(\hat{\btheta}_{1_n})u/2} }{(2\pi)^{d_{1}/2}} \ du \right. \\
    &\left. - n^{-d_{1}/2} \int_{\|u\| > \hat{\delta} \sqrt{n}} g(u) \overline{\Pi}_n(n^{-1/2}u + \hat{\btheta}_{1_n}) \ du -\frac{n^{-d_{1}/2} \int_{\|u\| \leq \hat{\delta}\sqrt{n}} g(u) \overline{\Pi}_n(n^{-1/2}u + \hat{\btheta}_{1_n}) \ du}{n^{-d_{1}/2} \int_{\|u\| \leq \hat{\delta}\sqrt{n}} \overline{\Pi}_n(n^{-1/2}u + \hat{\btheta}_{1_n}) \ du}  \right. \\
    &\left. + \frac{n^{-d_{1}/2} \int_{\|u\| \leq \hat{\delta}\sqrt{n}} g(u) \overline{\Pi}_n(n^{-1/2}t + \hat{\btheta}_{1_n}) \ du}{n^{-d_{1}/2} \int_{\|u\| \leq \hat{\delta}\sqrt{n}} \overline{\Pi}_n(n^{-1/2}u + \hat{\btheta}_{1_n}) \ du} - n^{-d_{1}/2} \int_{\|u\| \leq \hat{\delta} \sqrt{n}} g(u) \overline{\Pi}_n(n^{-1/2}u + \hat{\btheta}_{1_n}) \ du \right| \\
    =&  \left| \int_{\bbR^{d_1}} \left( g(u) -\frac{n^{-d_{1}/2} \int_{\|u\| \leq \hat{\delta}\sqrt{n}} g(u) \overline{\Pi}_n(n^{-1/2}u + \hat{\btheta}_{1_n}) \ du}{n^{-d_{1}/2} \int_{\|u\| \leq \hat{\delta}\sqrt{n}} \overline{\Pi}_n(n^{-1/2}u + \hat{\btheta}_{1_n}) \ du}  \right) \frac{|\text{det}(H_f(\hat{\btheta}_{1_n})^{-1})|^{-1/2} e^{-u^\top H_f(\hat{\btheta}_{1_n})u/2} }{(2\pi)^{d_{1}/2}} \ du \right. \\
    &\left. - n^{-d_{1}/2} \int_{\|u\| > \hat{\delta} \sqrt{n}} g(u) \overline{\Pi}_n(n^{-1/2}u + \hat{\btheta}_{1_n}) \ du \right. \\
    &\left. + \frac{n^{-d_{1}/2} \int_{\|u\| \leq \hat{\delta}\sqrt{n}} g(u) \overline{\Pi}_n(n^{-1/2}t + \hat{\btheta}_{1_n}) \ du}{n^{-d_{1}/2} \int_{\|u\| \leq \hat{\delta}\sqrt{n}} \overline{\Pi}_n(n^{-1/2}u + \hat{\btheta}_{1_n}) \ du} \left(  1 - n^{-d_{1}/2} \int_{\|u\| \leq \hat{\delta}\sqrt{n}} \overline{\Pi}_n(n^{-1/2}u + \hat{\btheta}_{1_n}) \ du \right) \right| \\
    =&  \left| \int_{\bbR^{d_1}} \left( g(u) -\frac{n^{-d_{1}/2} \int_{\|u\| \leq \hat{\delta}\sqrt{n}} g(u) \overline{\Pi}_n(n^{-1/2}u + \hat{\btheta}_{1_n}) \ du}{n^{-d_{1}/2} \int_{\|u\| \leq \hat{\delta}\sqrt{n}} \overline{\Pi}_n(n^{-1/2}u + \hat{\btheta}_{1_n}) \ du} \right) \frac{|\text{det}(H_f(\hat{\btheta}_{1_n})^{-1})|^{-1/2} e^{-u^\top H_f(\hat{\btheta}_{1_n})u/2} }{(2\pi)^{d_{1}/2}} \ du \right. \\
    &\left. - n^{-d_{1}/2} \int_{\|u\| > \hat{\delta} \sqrt{n}} \left( g(u) -\frac{n^{-d_{1}/2} \int_{\|u\| \leq \hat{\delta}\sqrt{n}} g(u) \overline{\Pi}_n(n^{-1/2}u + \hat{\btheta}_{1_n}) \ du}{n^{-d_{1}/2} \int_{\|u\| \leq \hat{\delta}\sqrt{n}} \overline{\Pi}_n(n^{-1/2}u + \hat{\btheta}_{1_n}) \ du} \right) \overline{\Pi}_n(n^{-1/2}u + \hat{\btheta}_{1_n}) \ du \right|. \label{eq:Dg} \tag{S30}
\end{align*}

For a more readable expression, we define a set of quantities as follows:
\begin{align*}
    C_{\calP} &:= n^{-d_1/2} \int_{\|u\| \leq \hat{\delta}\sqrt{n}} \overline{\Pi}_n (n^{-1/2}u + \hat{\btheta}_{1_n})\ du, \\
    C_{\calN} &:= \int_{\|u\| \leq \hat{\delta} \sqrt{n}} \frac{|\text{det}(H_f(\hat{\btheta}_{1_n})^{-1})|^{-1/2} e^{-u^\top H_f(\hat{\btheta}_{1_n})u/2} }{(2\pi)^{d_{1}/2}} \ du, \\
    h_g(u) &:= g(u) - \frac{n^{-d_1/2}}{C_{\calP}} \int_{\|u\| \leq \hat{\delta}\sqrt{n}} g(u) \overline{\Pi}_n (n^{-1/2}u + \hat{\btheta}_{1_n})\ du.
\end{align*}
Then we clarify~\eqref{eq:Dg} and continue the decomposition.
\begin{align*}
    D_g :=& \left| \int_{\bbR^{d_1}} h_g(u) \frac{|\text{det}(H_f(\hat{\btheta}_{1_n})^{-1})|^{-1/2} e^{-u^\top H_f(\hat{\btheta}_{1_n})u/2} }{(2\pi)^{d_{1}/2}} \ du  - n^{-d_{1}/2} \int_{\|u\| > \hat{\delta} \sqrt{n}} h_g(u) \overline{\Pi}_n(n^{-1/2}u + \hat{\btheta}_{1_n}) \ du \right| \\
    =& \left| \int_{\|u\| \leq \hat{\delta} \sqrt{n}} h_g(u) \frac{|\text{det}(H_f(\hat{\btheta}_{1_n})^{-1})|^{-1/2} e^{-u^\top H_f(\hat{\btheta}_{1_n})u/2} }{(2\pi)^{d_{1}/2}} \ du \right. \\
    &\left. + \int_{\|u\| > \hat{\delta} \sqrt{n}} h_g(u) \frac{|\text{det}(H_f(\hat{\btheta}_{1_n})^{-1})|^{-1/2} e^{-u^\top H_f(\hat{\btheta}_{1_n})u/2} }{(2\pi)^{d_{1}/2}} \ du - n^{-d_{1}/2} \int_{\|u\| > \hat{\delta} \sqrt{n}} h_g(u) \overline{\Pi}_n(n^{-1/2}u + \hat{\btheta}_{1_n}) \ du \right| \\
    \leq & \frac{1}{C_{\calN}} \left| \int_{\|u\| \leq \hat{\delta} \sqrt{n}} h_g(u) \frac{|\text{det}(H_f(\hat{\btheta}_{1_n})^{-1})|^{-1/2} e^{-u^\top H_f(\hat{\btheta}_{1_n})u/2} }{(2\pi)^{d_{1}/2}} \ du \right| \\
    & + \left| \int_{\|u\| > \hat{\delta} \sqrt{n}} h_g(u) \left[ \frac{|\text{det}(H_f(\hat{\btheta}_{1_n})^{-1})|^{-1/2} e^{-u^\top H_f(\hat{\btheta}_{1_n})u/2} }{(2\pi)^{d_{1}/2}} - n^{-d_{1}/2} \overline{\Pi}_n(n^{-1/2}u + \hat{\btheta}_{1_n}) \right] \ du  \right| \\
    = & \left| \int_{\|u\| \leq \hat{\delta} \sqrt{n}} g(u) \frac{|\text{det}(H_f(\hat{\btheta}_{1_n})^{-1})|^{-1/2} e^{-u^\top H_f(\hat{\btheta}_{1_n})u/2} }{(2\pi)^{d_{1}/2}} \ du - \frac{n^{-d_1/2}}{C_{\calP}} \int_{\|u\| \leq \hat{\delta}\sqrt{n}} g(u) \overline{\Pi}_n (n^{-1/2}u + \hat{\btheta}_{1_n})\ du \right| \\
    & + \left| \int_{\|u\| > \hat{\delta} \sqrt{n}} h_g(u) \left[ \frac{|\text{det}(H_f(\hat{\btheta}_{1_n})^{-1})|^{-1/2} e^{-u^\top H_f(\hat{\btheta}_{1_n})u/2} }{(2\pi)^{d_{1}/2}} - n^{-d_{1}/2} \overline{\Pi}_n(n^{-1/2}u + \hat{\btheta}_{1_n}) \right] \ du  \right| \\
    & := F_1 + F_2
\end{align*}

\subsection{\texorpdfstring{Controlling term $F_1$}{Controlling term F_1}}

\begin{align*}
    F_1 :=& \left| \int_{\|u\| \leq \hat{\delta} \sqrt{n}} g(u) \frac{|\text{det}(H_f(\hat{\btheta}_{1_n})^{-1})|^{-1/2} e^{-u^\top H_f(\hat{\btheta}_{1_n})u/2} }{(2\pi)^{d_{1}/2}} \ du - \frac{n^{-d_1/2}}{C_{\calP}} \int_{\|u\| \leq \hat{\delta}\sqrt{n}} g(u) \overline{\Pi}_n (n^{-1/2}u + \hat{\btheta}_{1_n})\ du \right| \\
    & \leq d_{TV} \left(\left[ \calN(0, H_f(\hat{\btheta}_{1_n})^{-1}) \right]_{B_0(\hat{\delta} \sqrt{n})}, \left[ \overline{\Pi}_n (X)\right]_{B_0(\hat{\delta} \sqrt{n})} \right) \tag{$g$ : indicator function} \\
    & \leq \sqrt{\frac{1}{2}\text{KL} \left(\left[ \calN(0, H_f(\hat{\btheta}_{1_n})^{-1}) \right]_{B_0(\hat{\delta} \sqrt{n})}, \left[ \overline{\Pi}_n (X) \right]_{B_0(\hat{\delta} \sqrt{n})} \right)} \tag{Pinsker's inequality} \\
    & \leq \frac{\sqrt{\text{Fisher}\left(\left[ \calN(0, H_f(\hat{\btheta}_{1_n})^{-1}) \right]_{B_0(\hat{\delta} \sqrt{n})}, \left[ \overline{\Pi}_n (X) \right]_{B_0(\hat{\delta} \sqrt{n})} \right)}}{2\sqrt{\lambda_{min}(H_f(\hat{\btheta}_{1_n})) - \hat{\delta} \hat{M}_1}}  
\label{eq:F1_1} \tag{S31} \\
    & = \frac{1}{2\sqrt{\lambda_{min}(H_f(\hat{\btheta}_{1_n})) - \hat{\delta}\hat{M}_1}} \\
    &\quad \cdot \sqrt{ \int_{\|u\| \leq \hat{\delta}\sqrt{n}} \frac{|\text{det}(H_f(\hat{\btheta}_{1_n})^{-1})|^{-1/2} e^{-u^\top H_f(\hat{\btheta}_{1_n})u/2}}{C_{\calN} (2\pi)^{d_{1}/2}} \left\| H_f(\hat{\btheta}_{1_n})u + \frac{f_n^{\dagger '}(n^{-1/2}u + \hat{\btheta}_{1_n})}{\sqrt{n}} \right\|^2 \ du} \\
    & \leq \frac{\hat{M}_1}{4\sqrt{n \left(\lambda_{min}(H_f(\hat{\btheta}_{1_n})) - \hat{\delta}\hat{M}_1\right)} } \sqrt{ \int_{\|u\| \leq \hat{\delta}\sqrt{n}} \frac{|\text{det}(H_f(\hat{\btheta}_{1_n})^{-1})|^{-1/2} e^{-u^\top H_f(\hat{\btheta}_{1_n})u/2}}{C_{\calN} (2\pi)^{d_{1}/2}} \left\| u \right\|^4 \ du} \label{eq:F1_2} \tag{S32} \\
    & \leq \frac{\sqrt{3} \tr \left[H_f(\hat{\btheta}_{1_n})^{-1}\right] \hat{M}_1}{4\sqrt{n \left(\lambda_{min}(H_f(\hat{\btheta}_{1_n})) - \hat{\delta}\hat{M}_1  \right) \left(1- e^{\hat{\calT}_n} \right) }}. \label{eq:F1_3} \tag{S33}
\end{align*}

Let $\xi \in (\hat{\btheta}_{1_n},  n^{-1/2}u + \hat{\btheta}_{1_n})$, then from the Taylor's theorem, Assumption~\ref{assum3}, and $\|u\| \leq \hat{\delta}\sqrt{n}$, we have 
\begin{align*}
    -\frac{f_n^{\dagger ''}(n^{-1/2}u + \hat{\btheta}_{1_n})}{n} &= -\frac{f_n^{\dagger ''}(\hat{\btheta}_{1_n})}{n} - \frac{f_n^{\dagger '''}(\xi) \| n^{-1/2} u \|}{n} \succeq \left( \lambda_{min} \left(H_f(\hat{\btheta}_{1_n})  \right) - \hat{\delta} \hat{M}_1 \right) \bI_{d_{1}}.
\end{align*}
Then it satisfied the assumption of Proposition~\ref{Bakry-Emery criterion}, which implies that $\sqrt{n}(\btheta_1 - \hat{\btheta}_{1_n})$ satisfies $LSI(\kappa)$ with $\kappa = \left( \lambda_{min} \left(H(\hat{\btheta}_{1_n})  \right) - \hat{\delta} \hat{M}_1 \right)$. Therefore, we can show the inequality \eqref{eq:F1_1} satisfied by Definition~\ref{LSI definition} log-Sobolev inequality.

Inequality \eqref{eq:F1_2} follows from the third-order Taylor's theorem and Assumption~\ref{assum3} for $\|u\| \leq \hat{\delta}\sqrt{n}$ and for some $\xi \in (\hat{\btheta}_{1_n},  n^{-1/2}u + \hat{\btheta}_{1_n})$,  
\begin{align*}
    & f_n^{\dagger}(n^{-1/2}u + \hat{\btheta}_{1_n}) = f_n^{\dagger}(\hat{\btheta}_{1_n}) + f_n^{\dagger '}(\hat{\btheta}_{1_n}) n^{-1/2}u + \frac{1}{2} f_n^{\dagger ''}(\hat{\btheta}_{1_n})(\|n^{-1/2}u\|^2) + \frac{1}{6} f_n^{\dagger '''}(\xi)(\|n^{-1/2}u\|^3), \\
    & \frac{f_n^{\dagger '}(n^{-1/2}u + \hat{\btheta}_{1_n})}{\sqrt{n}} = f_n^{\dagger '}(\hat{\btheta}_{1_n}) n^{-1/2} + \frac{f_n^{\dagger ''}(\hat{\btheta}_{1_n})u}{n} + \frac{1}{2n\sqrt{n}} f_n^{\dagger '''}(\xi)(\|u\|^2), \\
    & \left\| - \frac{f_n^{\dagger ''}(\hat{\btheta}_{1_n})u}{n} + \frac{f_n^{\dagger '}(n^{-1/2}u + \hat{\btheta}_{1_n})}{\sqrt{n}} \right\|^2 = \left\| \frac{\|u\|^2}{2n\sqrt{n}} f_n^{\dagger '''}(\xi) \right\|^2 \leq \left\| \frac{\hat{M}_1}{2\sqrt{n}}\right\|^2 \|u\|^4.
\end{align*}
Inequality \eqref{eq:F1_3} follows from the Lemma~\ref{std normal lemma} and Lemma~\ref{std normal lemma 2}. 

\subsection{\texorpdfstring{Controlling term $F_2$}{Controlling term F_2}}
We decompose $F_2$ once again.
\begin{align*}
    F_2 :=& \left| \int_{\|u\| > \hat{\delta} \sqrt{n}} h_g(u) \left[ \frac{|\text{det}(H_f(\hat{\btheta}_{1_n})^{-1})|^{-1/2} e^{-u^\top H_f(\hat{\btheta}_{1_n})u/2} }{(2\pi)^{d_{1}/2}} - n^{-d_{1}/2} \overline{\Pi}_n(n^{-1/2}u + \hat{\btheta}_{1_n}) \right] \ du  \right| \\
    \leq & \left| \int_{\|u\| > \hat{\delta} \sqrt{n}} g(u) \left[ \frac{|\text{det}(H_f(\hat{\btheta}_{1_n})^{-1})|^{-1/2} e^{-u^\top H_f(\hat{\btheta}_{1_n})u/2} }{(2\pi)^{d_{1}/2}} - n^{-d_{1}/2} \overline{\Pi}_n(n^{-1/2}u + \hat{\btheta}_{1_n}) \right] \ du  \right| \\
    & + \frac{n^{-d_1/2}}{C_{\calP}} \left| \int_{\|u\| \leq \hat{\delta}\sqrt{n}} g(u) \overline{\Pi}_n (n^{-1/2}u + \hat{\btheta}_{1_n})\ du \right. \\
    & \qquad \qquad \qquad \left. \cdot \int_{\|u\| > \hat{\delta} \sqrt{n}}  \left[ \frac{|\text{det}(H_f(\hat{\btheta}_{1_n})^{-1})|^{-1/2} e^{-u^\top H_f(\hat{\btheta}_{1_n})u/2} }{(2\pi)^{d_{1}/2}} - n^{-d_{1}/2} \overline{\Pi}_n(n^{-1/2}u + \hat{\btheta}_{1_n}) \right] \ du  \right| \\
    &:= F_{2,1} + F_{2,2}.
\end{align*}

\subsection{\texorpdfstring{Controlling term $F_{2,1}$}{Controlling term F_{2,1}}}
\begin{align*}
    F_{2,1} :=& \left| \int_{\|u\| > \hat{\delta} \sqrt{n}} g(u) \left[ \frac{|\text{det}(H_f(\hat{\btheta}_{1_n})^{-1})|^{-1/2} e^{-u^\top H_f(\hat{\btheta}_{1_n})u/2} }{(2\pi)^{d_{1}/2}} - n^{-d_{1}/2} \overline{\Pi}_n(n^{-1/2}u + \hat{\btheta}_{1_n}) \right] \ du  \right| \\
    \leq & \left| \int_{\|u\| > \hat{\delta} \sqrt{n}} g(u) \frac{|\text{det}(H_f(\hat{\btheta}_{1_n})^{-1})|^{-1/2} e^{-u^\top H_f(\hat{\btheta}_{1_n})u/2} }{(2\pi)^{d_{1}/2}} \ du \right| + \left| \int_{\|u\| > \hat{\delta} \sqrt{n}} g(u) n^{-d_{1}/2} \overline{\Pi}_n(n^{-1/2}u + \hat{\btheta}_{1_n}) \ du  \right| \\
    \leq & \int_{\|u\| > \hat{\delta} \sqrt{n}} | g(u) | \frac{|\text{det}(H_f(\hat{\btheta}_{1_n})^{-1})|^{-1/2} e^{-u^\top H_f(\hat{\btheta}_{1_n})u/2} }{(2\pi)^{d_{1}/2}} \ du + \int_{\|u\| > \hat{\delta} \sqrt{n}} | g(u) | n^{-d_{1}/2} \overline{\Pi}_n(n^{-1/2}u + \hat{\btheta}_{1_n}) \ du.
\end{align*}

The upper bound of the first term of $F_{2,1}$ is derived from $|g| \leq 1$ and Lemma \ref{std normal lemma}.
\begin{align*}
    &\int_{\|u\| > \hat{\delta} \sqrt{n}} | g(u) | \frac{|\text{det}(H_f(\hat{\btheta}_{1_n})^{-1})|^{-1/2} e^{-u^\top H_f(\hat{\btheta}_{1_n})u/2} }{(2\pi)^{d_{1}/2}} \\
    &\leq \int_{\|u\| > \hat{\delta} \sqrt{n}} \frac{|\text{det}(H_f(\hat{\btheta}_{1_n})^{-1})|^{-1/2} e^{-u^\top H_f(\hat{\btheta}_{1_n})u/2} }{(2\pi)^{d_{1}/2}} \tag{$|g| \leq 1$} \\ 
    &\leq \exp \left(-\frac{1}{2} \left( \hat{\delta}\sqrt{n} - \sqrt{\tr[H_f(\hat{\btheta}_{1_n})^{-1}]} \right)^2 \left\| H_f(\hat{\btheta}_{1_n})^{-1}\right\|^{-1}_{op} \right) := \exp(\hat{\calT}_n) \label{eq:F21} \tag{S34}.
\end{align*}

The upper bound of the second term of $F_{2,1}$ is derived using Taylor's theorem, Assumption~\ref{assum1},~\ref{assum2}, and~\ref{assum5}, and $|g| \leq 1$.
\begin{align*}
    &\int_{\|u\| > \hat{\delta} \sqrt{n}} | g(u) | n^{-d_{1}/2} \overline{\Pi}_n(n^{-1/2}u + \hat{\btheta}_{1_n}) \ du \\
    &= \int_{\|t-\hat{\btheta}_{1_n}\| > \hat{\delta} } | g(\sqrt{n}(t - \hat{\btheta}_{1_n}) | \overline{\Pi}_n(t) \ dt \\
    &\leq \frac{\int_{\|t-\bar{\btheta}_{1_n}\| > \hat{\delta} - \|\bar{\btheta}_{1_n} - \hat{\btheta}_{1_n}\|}  | g(\sqrt{n}(t - \hat{\btheta}_{1_n}) | \pi(t) \exp\{\ell_n(t)\} \ dt}{ \int_{\|t-\bar{\btheta}_{1_n}\| > \bar{\delta}} \pi(t) \exp\{\ell_n(t)\} \ dt} \\
    &\leq \frac{\int_{\|t-\bar{\btheta}_{1_n}\| > \hat{\delta} - \|\bar{\btheta}_{1_n} - \hat{\btheta}_{1_n}\|}  | g(\sqrt{n}(t - \hat{\btheta}_{1_n}) | \pi(t) \exp\{\ell_n(t) - \ell_n(\bar{\btheta}_{1_n})\} \ dt}{ \int_{\|t-\bar{\btheta}_{1_n}\| > \bar{\delta}} \pi(t) \exp\{\ell_n(t) - \ell_n(\bar{\btheta}_{1_n})\} \ dt} \\
    &\leq \frac{\exp(-n \hat{\kappa})\int_{\|t-\bar{\btheta}_{1_n}\| > \hat{\delta} - \|\bar{\btheta}_{1_n} - \hat{\btheta}_{1_n}\|}  | g(\sqrt{n}(t - \hat{\btheta}_{1_n}) | \pi(t) \ dt}{ \int_{\|t-\bar{\btheta}_{1_n}\| > \bar{\delta}} \pi(t) \exp\{\ell_n(t) - \ell_n(\bar{\btheta}_{1_n})\} \ dt} \tag{Assumption~\ref{assum5}} \\
    &\leq \frac{\exp(-n \hat{\kappa})}{ \int_{\|t-\bar{\btheta}_{1_n}\| > \bar{\delta}} \pi(t) \exp\{\ell_n(t) - \ell_n(\bar{\btheta}_{1_n})\} \ dt} \tag{$|g| \leq 1$} \\
    &= \frac{n^{d_{1}/2} \exp(-n \hat{\kappa})}{ \int_{\|u\| > \bar{\delta}\sqrt{n}} \pi(n^{-1/2}u + \bar{\btheta}_{1_n}) \exp\{\ell_n(n^{-1/2}u + \bar{\btheta}_{1_n}) - \ell_n(\bar{\btheta}_{1_n})\} \ du} \\
    &\leq \frac{n^{d_{1}/2} \exp(-n \hat{\kappa})}{ \int_{\|u\| > \bar{\delta}\sqrt{n}} \pi(n^{-1/2}u + \bar{\btheta}_{1_n}) \exp\{ -(u^\top J(\bar{\btheta}_{1_n}) u) / 2\} \ du} \tag{Taylor's theorem, Assumption~\ref{assum1}} \\
    &\leq \frac{\bar{M}_2 n^{d_{1}/2} \exp(-n \hat{\kappa})}{\int_{\|u\| > \bar{\delta}\sqrt{n}} \exp\{ -(u^\top J(\bar{\btheta}_{1_n}) u) / 2\} \ du} \tag{Assumption~\ref{assum2}} \\
    &\leq \frac{| \text{det} (J(\bar{\btheta}_{1_n}))^{-1} |^{-1/2} \bar{M}_2 n^{d_{1}/2} \exp(-n \hat{\kappa})}{ (2 \pi)^{d_{1}/2} (1 - \exp(\bar{\calT}_J) )}. \label{eq:F21buond} \tag{S35}
\end{align*}

In summary, the final upper bound for $F_{2,1}$ is expressed as follows:
\begin{align*}
    F_{2,1} \leq \exp(\hat{\calT}_n) + \frac{| \text{det} (J(\bar{\btheta}_{1_n}))^{-1} |^{-1/2} \bar{M}_2 n^{d_{1}/2} \exp(-n \hat{\kappa})}{ (2 \pi)^{d_{1}/2} (1 - \exp(\bar{\calT}_J) )}. \label{eq:F21bound2} \tag{S36}
\end{align*}

\subsection{\texorpdfstring{Controlling term $F_{2,2}$}{Controlling term F_{2,2}}}
Since $g$ is an indicator function, $g$ that maximizes absolute value integral is a constant $1$, so the maximum value of $F_{2,1}$ is when $g \equiv 1$, and that value is $F_{2,2}$. Hence it suffices to bound $F_{2,2}$ by evaluating $F_{2,1}$ at $g \equiv 1$.

\begin{align*}
    F_{2,2} &:= \frac{n^{-d_1/2}}{C_{\calP}} \left| \int_{\|u\| \leq \hat{\delta}\sqrt{n}} g(u) \overline{\Pi}_n (n^{-1/2}u + \hat{\btheta}_{1_n})\ du \right. \\
    & \qquad \qquad \qquad \left. \cdot \int_{\|u\| > \hat{\delta} \sqrt{n}}  \left[ \frac{|\text{det}(H_f(\hat{\btheta}_{1_n})^{-1})|^{-1/2} e^{-u^\top H_f(\hat{\btheta}_{1_n})u/2} }{(2\pi)^{d_{1}/2}} - n^{-d_{1}/2} \overline{\Pi}_n(n^{-1/2}u + \hat{\btheta}_{1_n}) \right] \ du  \right| \tag{$|g| \leq 1$} \\
    &\leq \left| \int_{\|u\| > \hat{\delta} \sqrt{n}}  \left[ \frac{|\text{det}(H_f(\hat{\btheta}_{1_n})^{-1})|^{-1/2} e^{-u^\top H_f(\hat{\btheta}_{1_n})u/2} }{(2\pi)^{d_{1}/2}} - n^{-d_{1}/2} \overline{\Pi}_n(n^{-1/2}u + \hat{\btheta}_{1_n}) \right] \ du  \right|  \\
    &\leq \exp(\hat{\calT}_n) + \frac{| \text{det} (J(\bar{\btheta}_{1_n}))^{-1} |^{-1/2} \bar{M}_2 n^{d_{1}/2} \exp(-n \hat{\kappa})}{ (2 \pi)^{d_{1}/2} (1 - \exp(\bar{\calT}_J) )}. \label{eq:F22bound} \tag{S37}
\end{align*}

Note that conditions $\sqrt{\tr \left[H_f(\hat{\btheta}_{1_n})^{-1}\right] / n} < \hat{\delta}$ and $\sqrt{\tr \left[J(\bar{\btheta}_{1_n})^{-1}\right]/n} < \bar{\delta}$ are required from the inequalities \eqref{eq:F21}, \eqref{eq:F21buond} and \eqref{eq:F22bound}, which are derived from Lemma \ref{std normal lemma}. These conditions are stated in Assumption~\ref{assum4}.

From \eqref{eq:F1_3}, \eqref{eq:F21bound2}, and \eqref{eq:F22bound}, the TV distance between the rescaled variational posterior $\overline{\Pi}_n$ and its Laplace approximation admits the following upper bound:
\begin{align*}
d_{TV} \left(
\mathcal{N}\!\left(0 , H_f(\hat{\btheta}_{1_n})^{-1} \right),
\overline{\Pi}_n
\right)
\;\le\;&
\frac{\sqrt{3}\,\tr \!\left[H_f(\hat{\btheta}_{1_n})^{-1}\right]\hat{M}_1}
{4\sqrt{\left(\lambda_{\min}(H_f(\hat{\btheta}_{1_n}))-\hat{\delta}\hat{M}_1\right)
\left(1-\exp(\hat{\calT}_n)\right)}}\,n^{-1/2} \\
&\quad +\; 2\exp(\hat{\calT}_n)
\;+\;
\frac{2\,\big|\det(J(\bar{\btheta}_{1_n})^{-1})\big|^{-1/2}\bar{M}_2}
{(2\pi)^{d_1/2}\big(1-\exp(\bar{\calT}_J)\big)}\,
n^{d_1/2}\exp(-n\hat{\kappa}).
\end{align*}

\end{proof}

Note that if assumptions~\ref{assum1}--\ref{assum5} are satisfied and if we fix the dimension $d_1$, all bounds will vanish as $n \rightarrow \infty$ at a rate of $\frac{1}{\sqrt{n}}$, as long as $\hat{\kappa} \gg \frac{d_1+1}{2} \cdot \frac{\log n}{n}$. These bounds are completely non-asymptotic and computable. For more details, see Appendix D of \citet{kasprzak2025laplace}.

\section{Proof of Theorem~\ref{thm:yenv}}
\label{sec:E}

This section provides the proof of Theorem~\ref{thm:yenv} in the main text. The proof proceeds by verifying that, under the reparameterized Bayesian response envelope model, the model-specific Assumption~\ref{assum6} implies Assumptions~\ref{assum1}--\ref{assum5} with probability tending to one. Once these conditions are established, the conclusion follows directly from Theorem~\ref{main-theorem} by identifying the nonconjugate block $\btheta_1$ with the envelope parameter $\bA$ and the conjugate block $\btheta_2$ with $\bZ=(\btmu,\tilde{\boldsymbol{\eta}},\btOmega,\btOmega_0)$.

For clarity, we first collect several auxiliary lemmas controlling the local
curvature and localization properties of the variational log-likelihood
$\ell_n^\dagger(\bA)$ and the variational log-posterior $\tf_n^\dagger(\bA)$. These lemmas are then combined to verify Assumptions~\ref{assum1}--\ref{assum5} under Assumption~\ref{assum6}, completing the proof.

\begin{lemma}[Local strong concavity of the variational log-likelihood and log-posterior]
\label{lem:local-concavity}
Suppose Assumption~\ref{assum6} holds. Then, with probability tending to one, there
exists a constant $\delta>0$ such that, for all sufficiently large $n$,
\[
\inf_{\bA\in\mathcal N}
\lambda_{\min}\!\Big(-\tfrac{1}{n}\nabla^2_{\vecr(\bA)}\ell^{\dagger}_n(\bA)\Big)
\ \ge\ \delta.
\]

Moreover, under the matrix-normal prior
$\bA\sim\mathcal{MN}_{r-u,u}(\bA_0,\bU_0^{(\bA)},\bV_0^{(\bA)})$, the variational
log-posterior $\tf_n^\dagger(\bA)=\ell_n^\dagger(\bA)+\log\pi(\bA)$ satisfies
\[
\inf_{\bA\in\mathcal N}
\lambda_{\min}\!\Big(-\tfrac{1}{n}\nabla^2_{\vecr(\bA)}\tf^{\dagger}_n(\bA)\Big)
\ \ge\ \delta.
\]

Consequently, $\ell^{\dagger}_n(\bA)$ and $\tf^{\dagger}_n(\bA)$ are uniformly strongly
concave on the convex set $\mathcal N$ and each admits a unique maximizer on $\mathcal N$,
denoted by
\[
\bbA_n := \argmax_{\bA\in\mathcal N}\ell^{\dagger}_n(\bA),
\qquad
\hbA_n := \argmax_{\bA\in\mathcal N}\tf_n^\dagger(\bA).
\]

Furthermore, if $\rho-\|\bbA_n-\bA^\star\|_F>0$ and
$\rho-\|\hbA_n-\bA^\star\|_F>0$, then there exist radii
\[
0<\bar{\delta}_y\le \rho-\|\bbA_n-\bA^\star\|_F,
\qquad
0<\hat{\delta}_y\le \rho-\|\hbA_n-\bA^\star\|_F,
\]
such that $\ell^{\dagger}_n$ and $\tf_n^\dagger$ each have a unique maximizer on the
corresponding closed balls
$\mathcal B(\bbA_n,\bar{\delta}_y)$ and $\mathcal B(\hbA_n,\hat{\delta}_y)$.
\end{lemma}
The proof of Lemma~\ref{lem:local-concavity} is given in Section~\ref{sec:E.1}. Lemma~\ref{lem:local-concavity} establishes the local strong concavity and uniqueness of the variational log-likelihood and variational log-posterior on $\mathcal N$, thereby verifying Assumption~\ref{assum1}(a),(b) and Assumption~\ref{assum3}(a),(b) for the Bayesian response envelope model.

\begin{lemma}[Uniform third-derivative bound on a local neighborhood]
\label{lem:third-derivative}
Suppose Assumption~\ref{assum6} holds and the parameter dimensions are fixed as $n\to\infty$.
Then there exists a finite constant $M_1<\infty$ (independent of $n$) such that, for all
sufficiently large $n$,
\[
\sup_{\bA\in\mathcal N}\Big\|\,\tfrac{1}{n}\nabla^3_{\vecr(\bA)}\ell_n^\dagger(\bA)\,\Big\|^{*}
\ \le\ M_1,
\qquad
\sup_{\bA\in\mathcal N}\Big\|\,\tfrac{1}{n}\nabla^3_{\vecr(\bA)}\tf_n^\dagger(\bA)\,\Big\|^{*}
\ \le\ M_1.
\]
In particular, on the event that $\mathcal B(\bbA_n,\bar\delta_y)\subseteq\mathcal N$ and
$\mathcal B(\hbA_n,\hat\delta_y)\subseteq\mathcal N$ (as in Lemma~\ref{lem:local-concavity}),
the same bound holds on these closed balls.
\end{lemma}
The proof of Lemma~\ref{lem:third-derivative} is given in Section~\ref{sec:E.2}. Lemma~\ref{lem:third-derivative} provides the uniform third-derivative control required by the nonasymptotic Laplace approximation bound. Assumptions~\ref{assum1}(c) and~\ref{assum3}(c) require that the rescaled third derivatives of the variational log-likelihood $\ell_n^\dagger$ and the variational log-posterior $\tf_n^\dagger$ be uniformly bounded on neighborhoods of their maximizers. Lemma~\ref{lem:local-concavity} guarantees that the relevant local neighborhoods $\mathcal B(\bbA_n,\bar\delta_y)$ and $\mathcal B(\hbA_n,\hat\delta_y)$ are contained in the fixed compact set $\mathcal N$ with probability tending to one. Therefore, the uniform bound on $\mathcal N$ established in Lemma~\ref{lem:third-derivative} immediately implies the corresponding bounds on these balls, verifying Assumptions~\ref{assum1}(c) and~\ref{assum3}(c) for the Bayesian response envelope model.

\begin{lemma}[Finite bound for $M_2$ centered at the local MVLE]
\label{lem:bound-M2}
Let $d_{\bA}=(r-u)u$ and let $\bSigma_0^{(\bA)}=\bV_0^{(\bA)}\otimes\bU_0^{(\bA)}$ be the fixed prior covariance,
so that $\vecr(\bA)\sim\mathcal N(\vecr(\bA_0),\bSigma_0^{(\bA)})$.
Fix a compact neighborhood $\mathcal N=\mathcal B(\bA^\star,\rho)$ and let
\[
\bbA_n \in \arg\max_{\bA\in\mathcal N}\ell_n^\dagger(\bA).
\]
Then for any fixed $\bar\delta>0$,
\[
\sup_{\|\bA-\bbA_n\|_F\le \bar\delta}\ \Big| \frac{1}{\pi(\bA)} \Big|
\;\le\; M_2,
\]
where $M_2<\infty$ is a deterministic constant (independent of $\bA$ and $n$).
In particular, Assumption~\ref{assum2} holds for the $\bA$-update on $\mathcal N$.
\end{lemma}
The proof of Lemma~\ref{lem:bound-M2} is given in Section~\ref{sec:E.3}. Lemma~\ref{lem:bound-M2} ensures that Assumption~\ref{assum2} is satisfied for the response envelope model.

\begin{lemma}[Outer--shell descent around the local MVLE under fixed $q_n^\dagger$]
\label{lem:variational-posterior}
Suppose Assumption~\ref{assum6} holds.
Let $\bbA_n\in\argmax_{\bA\in\mathcal N}\ell^{\dagger}_n(\bA)$ and
$\hbA_n\in\argmax_{\bA\in\mathcal N}\tf_n^\dagger(\bA)$.
Let $\bar\delta_y,\hat\delta_y$ be radii as in Lemma~\ref{lem:local-concavity}, chosen so that
$\mathcal B(\bbA_n,\bar\delta_y)\subset\mathcal N$ and $\mathcal B(\hbA_n,\hat\delta_y)\subset\mathcal N$,
and without loss of generality $\hat\delta_y\le \bar\delta_y$.

Let
\[
M_1\ :=\ \sup_{\bA\in\mathcal N}\Big\|\tfrac{1}{n}\nabla^3_{\vecr(\bA)}\ell^{\dagger}_n(\bA)\Big\|^{*},
\]
which is finite for all sufficiently large $n$ by Lemma~\ref{lem:third-derivative}.
Define
\[
\rho_n\ :=\ \hat\delta_y-\|\bbA_n-\hbA_n\|_F.
\]
Then, with probability tending to one and for all sufficiently large $n$, $0<\rho_n\le\bar\delta_y$ and
\[
\sup_{\ \rho_n\le \|\bA-\bbA_n\|_F\le \bar\delta_y}\ 
\frac{\ell_n(\bA)-\ell_n(\bbA_n)}{n}
\ \le\ -\,\hat\kappa_n,
\]
where
\[
\hat\kappa_n\ :=\ \min \!\left\{ \frac{\delta}{2}\,\rho_n^2\;-\;\frac{M_1}{6}\,\rho_n^3,\ \ \frac{\delta}{2}\,\bar{\delta}_y^2\;-\;\frac{M_1}{6}\,\bar{\delta}_y^3 \right\},
\]
and $\delta>0$ is the uniform strong-concavity constant from Lemma~\ref{lem:local-concavity}.
In particular, $\hat\kappa_n>0$ whenever $\rho_n<3\delta/M_1$ and $\bar\delta_y<3\delta/M_1$.

Moreover, under the matrix-normal (Gaussian) prior on $\vecr(\bA)$,
\[
\|\hbA_n-\bbA_n\|_F = O(n^{-1}),
\]
hence $\rho_n\to \hat\delta_y$ and $\rho_n>0$ with probability tending to one.
\end{lemma}
The proof of Lemma~\ref{lem:variational-posterior} is given in Section~\ref{sec:E.4}. Lemma~\ref{lem:variational-posterior} establishes a uniform outer--shell descent property for the variational log-likelihood around the local MVLE $\bbA_n$. Specifically, it shows that outside a shrinking neighborhood of $\bbA_n$—but still within the region where local quadratic control holds—the variational log-likelihood decreases at a rate proportional to $n$. As a consequence, the induced variational posterior assigns exponentially small mass to this outer shell.

This result verifies the tail-decay condition required in Assumption~\ref{assum5} and plays a crucial role in controlling the contribution of the nonlocal region to the TV distance between the exact CAVI update and its Laplace approximation.

\paragraph{Explicit verification of Assumption~\ref{assum4} for the $\bA$-block.}
Let $d_{\bA}=(r-u)u$ and work on the high-probability event in Assumption~\ref{assum6}.
By Lemma~\ref{lem:local-concavity}, there exists $\delta>0$ such that, for all sufficiently large $n$,
\[
-\nabla^2_{\vecr(\bA)}\ell_n^\dagger(\bA)\ \succeq\ n\delta\,\bI_{d_{\bA}},
\qquad \forall\bA\in\mathcal N.
\]
Under the Gaussian (matrix-normal) prior on $\vecr(\bA)$,
$\nabla^2_{\vecr(\bA)}\log\pi(\bA)\preceq 0$, hence
\[
-\nabla^2_{\vecr(\bA)}\tf_n^\dagger(\bA)
=
-\nabla^2_{\vecr(\bA)}\ell_n^\dagger(\bA)-\nabla^2_{\vecr(\bA)}\log\pi(\bA)
\succeq
n\delta\,\bI_{d_{\bA}},
\qquad \forall\bA\in\mathcal N.
\]
In particular, at the local maximizers $\bbA_n$ and $\hbA_n$,
\[
\tr\!\Big(\big[-\nabla^2_{\vecr(\bA)}\ell_n^\dagger(\bbA_n)\big]^{-1}\Big)\le \frac{d_{\bA}}{n\delta},
\qquad
\tr\!\Big(\big[-\nabla^2_{\vecr(\bA)}\tf_n^\dagger(\hbA_n)\big]^{-1}\Big)\le \frac{d_{\bA}}{n\delta}.
\]
Therefore the local Gaussian scales in Assumption~\ref{assum4} are $O(n^{-1/2})$.

Moreover, since $\nabla \ell_n^\dagger(\bbA_n)=0$ and $\nabla \tf_n^\dagger(\hbA_n)=0$,
there exists $\tilde{\bA}_n$ on the line segment between $\bbA_n$ and $\hbA_n$ such that
\[
-\nabla^2_{\vecr(\bA)}\ell_n^\dagger(\tilde{\bA}_n)\,\big(\vecr(\hbA_n)-\vecr(\bbA_n)\big)
=\nabla_{\vecr(\bA)}\log\pi(\hbA_n).
\]
Using $-\nabla^2\ell_n^\dagger(\tilde{\bA}_n)\succeq n\delta\,\bI$ and compactness of $\mathcal N$ yields
\[
\|\hbA_n-\bbA_n\|_F
=
\|\vecr(\hbA_n)-\vecr(\bbA_n)\|_2
\le
\frac{1}{n\delta}\sup_{\bA\in\mathcal N}\|\nabla_{\vecr(\bA)}\log\pi(\bA)\|_2
=O(n^{-1}).
\]
Consequently, the proximity and scale inequalities required in Assumption~\ref{assum4} hold for all sufficiently large $n$.

Together, Lemmas~\ref{lem:local-concavity}, \ref{lem:third-derivative},
\ref{lem:bound-M2}, \ref{lem:variational-posterior}, and above paragraph verify
Assumptions~\ref{assum1}--\ref{assum5} for the reparameterized Bayesian response envelope model under Assumption~\ref{assum6}.

\subsection{Proof of Lemma~\ref{lem:local-concavity}}
\label{sec:E.1}
We first give and prove Proposition~\ref{prop:DCT2} to ensure the interchangeability of the second derivative with expectation.

\begin{proposition}[Exchange of differentiation and expectation (second order)]
\label{prop:DCT2}
Suppose Assumption~\ref{assum6} holds. Fix the current variational factor $q(\bZ)=q(\btmu)\,q(\tilde{\boldsymbol\eta})\,q(\btOmega)\,q(\btOmega_0)$
used in our CALVI scheme, where $q(\btmu),q(\tilde{\boldsymbol\eta})$ are Gaussian and
$q(\btOmega),q(\btOmega_0)$ are inverse-Wishart.
Let $\log p(\bZ,\bA\mid\bbY_n)$ be the reparameterized complete-data log-posterior in \eqref{eq:12}.
Fix a bounded neighborhood $\mathcal N$ of $\bA^\star$.

Then there exists a measurable function $h_2(\bZ)$, independent of $\bA$, such that
\[
\sup_{\bA\in\mathcal N}\Big\|\nabla^2_{\vecr(\bA)}\log p(\bZ,\bA\mid\bbY_n)\Big\|_{op}\ \le\ h_2(\bZ),
\qquad
\E_{q(\bZ)}[h_2(\bZ)]<\infty.
\]
Consequently, for every $\bA\in\mathcal N$,
\[
\nabla^2_{\vecr(\bA)}
\int q(\bZ)\,\log p(\bZ,\bA\mid\bbY_n)\,d\bZ
=
\int q(\bZ)\,\nabla^2_{\vecr(\bA)} \log p(\bZ,\bA\mid\bbY_n)\,d\bZ.
\]
\end{proposition}

\begin{proof}[Proof of Proposition~\ref{prop:DCT2}]
In \eqref{eq:12}, the parameter $\bA$ appears only through the terms
$\log|\bI_u+\bA^\top\bA|$, the affine maps $\bC_{\bA}$ and $\bD_{\bA}$, and the Gaussian prior term.
(All remaining terms are $\bA$--free.) Recall $\bC_{\bA}=\bK+\bL\bA$ and $\bD_{\bA}=\bL-\bK\bA^\top$,
so $\bC_{\bA}$ and $\bD_{\bA}$ are affine in $\bA$.

By Assumption~\ref{assum6}, for every $\bA\in\mathcal N$,
\[
\|\bA\|_{op}\le \|\bA\|_F \le \|\bA^\star\|_F+\rho \;=:\; B_{\mathcal N}<\infty.
\]

\smallskip
\noindent\textbf{Step 1: Construct an $\bA$--uniform Hessian bound.}
Write the $\bA$--dependent part of \eqref{eq:12} (up to an additive constant) as
\begin{align*}
\log p(\bZ,\bA\mid\bbY_n)
&= c_n\,\log|\bI_u+\bA^\top\bA|
-\frac12\tr\!\Big(\bC_{\bA}^\top \bS_1(\btmu)\,\bC_{\bA}\,\btOmega^{-1}\Big)
-\frac12\tr\!\Big(\bD_{\bA}^\top \bS_0(\btmu)\,\bD_{\bA}\,\btOmega_0^{-1}\Big) \\
&\quad -\frac12\tr\!\Big((\bA-\bA_0)^\top \bU_0^{(\bA)-1}(\bA-\bA_0)\,\bV_0^{(\bA)-1}\Big)
\;+\; \tr\!\Big(\tilde{\boldsymbol\eta}\,\bH(\btmu)\,\bC_{\bA}\,\btOmega^{-1}\Big),
\end{align*}
where $c_n=\frac{2n+\nu^{(1)}+\nu^{(0)}}{2}$ and
\[
\bS_1(\btmu):=\bbY_{\btmu}^\top\bbY_{\btmu}+\bB_0\bM\bB_0^\top+\psi^{(1)}\bI_r,
\qquad
\bS_0(\btmu):=\bbY_{\btmu}^\top\bbY_{\btmu}+\psi^{(0)}\bI_r.
\]
The last trace term is \emph{linear} in $\bA$ (because $\bC_{\bA}$ is affine in $\bA$), hence it has zero
second derivative in $\bA$ and does not contribute to the Hessian.

\smallskip
\emph{(a) Log-determinant term.}
Let $g(\bA)=\log|\bI_u+\bA^\top\bA|$ and $J(\bA)=(\bI_u+\bA^\top\bA)^{-1}$.
Then $\nabla_{\bA} g(\bA)=2\bA J(\bA)$ and for any direction $\bU$,
\[
D(\nabla_{\bA} g(\bA))[\bU]
=2\bU J(\bA)-2\bA J(\bA)\,(\bA^\top\bU+\bU^\top\bA)\,J(\bA).
\]
Using $\|J(\bA)\|_{op}\le 1$ and $\|\bA^\top\bU+\bU^\top\bA\|_F\le 2\|\bA\|_F\|\bU\|_F$, we obtain
\[
\|D(\nabla_{\bA} g(\bA))[\bU]\|_F
\le (2+4\|\bA\|_F^2)\,\|\bU\|_F
\le (2+4B_{\mathcal N}^2)\,\|\bU\|_F,
\]
hence
\[
\sup_{\bA\in\mathcal N}\|\nabla^2_{\vecr(\bA)} g(\bA)\|_{op}\ \le\ 2+4B_{\mathcal N}^2.
\]

\smallskip
\emph{(b) Quadratic trace terms.}
Expanding $\bC_{\bA}=\bK+\bL\bA$ gives
\[
\tr(\bC_{\bA}^\top \bS_1 \bC_{\bA}\,\btOmega^{-1})
=
\tr(\bA^\top \bL^\top\bS_1\bL\,\bA\,\btOmega^{-1}) + \text{(terms at most linear in $\bA$)}.
\]
Therefore the $\bA$--Hessian of this term equals the Hessian of the quadratic form
$\tr(\bA^\top \bL^\top\bS_1\bL\,\bA\,\btOmega^{-1})$.
Using \eqref{eq:14} yields
\[
\Big\|\nabla^2_{\vecr(\bA)} \tr(\bA^\top \bL^\top\bS_1(\btmu)\bL\,\bA\,\btOmega^{-1})\Big\|_{op}
\le 2\,\|\btOmega^{-1}\|_{op}\,\|\bL^\top\bS_1(\btmu)\bL\|_{op}.
\]
Multiplying by the factor $1/2$ in \eqref{eq:12} gives the contribution
$\|\btOmega^{-1}\|_{op}\,\|\bL^\top\bS_1(\btmu)\bL\|_{op}$.
Similarly, expanding $\bD_{\bA}=\bL-\bK\bA^\top$ shows that
\[
\tr(\bD_{\bA}^\top \bS_0 \bD_{\bA}\,\btOmega_0^{-1})
=
\tr(\bA^\top \btOmega_0^{-1}\bA\,\bK^\top\bS_0(\btmu)\bK)+\text{(terms at most linear in $\bA$)},
\]
hence
\[
\Big\|\nabla^2_{\vecr(\bA)} \tr(\bD_{\bA}^\top \bS_0(\btmu) \bD_{\bA}\,\btOmega_0^{-1})\Big\|_{op}
\le 2\,\|\btOmega_0^{-1}\|_{op}\,\|\bK^\top\bS_0(\btmu)\bK\|_{op},
\]
so the $1/2$ factor yields the contribution
$\|\btOmega_0^{-1}\|_{op}\,\|\bK^\top\bS_0(\btmu)\bK\|_{op}$.

\smallskip
\emph{(c) Gaussian prior term.}
The prior term is quadratic in $\vecr(\bA)$, hence its Hessian is the constant matrix
$-(\bV_0^{(\bA)})^{-1}\otimes(\bU_0^{(\bA)})^{-1}$ and therefore bounded on $\mathcal N$:
\[
\Big\|\nabla^2_{\vecr(\bA)}\log\pi(\bA)\Big\|_{op}
=\big\|(\bV_0^{(\bA)})^{-1}\otimes(\bU_0^{(\bA)})^{-1}\big\|_{op}
=\|(\bV_0^{(\bA)})^{-1}\|_{op}\,\|(\bU_0^{(\bA)})^{-1}\|_{op}.
\]

\smallskip
Collecting the three contributions above and using $\|\bL^\top\bS\bL\|_{op}\le\|\bS\|_{op}$,
$\|\bK^\top\bS\bK\|_{op}\le\|\bS\|_{op}$, we obtain for all $\bA\in\mathcal N$,
\[
\big\|\nabla^2_{\vecr(\bA)}\log p(\bZ,\bA\mid\bbY_n)\big\|_{op}
\le
c_n(2+4B_{\mathcal N}^2)
+ \|\btOmega^{-1}\|_{op}\,\|\bS_1(\btmu)\|_{op}
+ \|\btOmega_0^{-1}\|_{op}\,\|\bS_0(\btmu)\|_{op}
+ C_\pi,
\]
where $C_\pi:=\|(\bV_0^{(\bA)})^{-1}\|_{op}\,\|(\bU_0^{(\bA)})^{-1}\|_{op}$.
Finally, $\|\bS_1(\btmu)\|_{op}\le \|\bbY_{\btmu}\|_F^2+C_{S,1}$ and
$\|\bS_0(\btmu)\|_{op}\le \|\bbY_{\btmu}\|_F^2+C_{S,0}$ for constants $C_{S,1},C_{S,0}<\infty$.
Therefore we may define
\[
h_2(\bZ)
:=
c_n(2+4B_{\mathcal N}^2)
+\big(\|\btOmega^{-1}\|_{op}+\|\btOmega_0^{-1}\|_{op}\big)\,\big(\|\bbY_{\btmu}\|_F^2+1\big)
+ C_0,
\]
for a finite constant $C_0<\infty$ absorbing $C_{S,1},C_{S,0},C_\pi$ and other fixed terms, so that
$\sup_{\bA\in\mathcal N}\|\nabla^2_{\vecr(\bA)}\log p(\bZ,\bA\mid\bbY_n)\|_{op}\le h_2(\bZ)$.

\smallskip
\noindent\textbf{Step 2: Integrability of $h_2(\bZ)$.}
Under the mean-field factorization, $\btmu$ is independent of $(\btOmega,\btOmega_0)$ under $q$.
Since $q(\btmu)$ is Gaussian, $\E_q\|\bbY_{\btmu}\|_F^2<\infty$.
Since $q(\btOmega)$ and $q(\btOmega_0)$ are inverse-Wishart, $\E_q\tr(\btOmega^{-1})<\infty$ and
$\E_q\tr(\btOmega_0^{-1})<\infty$, hence also
$\E_q\|\btOmega^{-1}\|_{op}<\infty$ and $\E_q\|\btOmega_0^{-1}\|_{op}<\infty$
because $\|\cdot\|_{op}\le \tr(\cdot)$ for positive semidefinite matrices.
Therefore $\E_q[h_2(\bZ)]<\infty$.

\smallskip
\noindent\textbf{Step 3: Dominated convergence for second derivatives.}
Let $F(\bA):=\int q(\bZ)\log p(\bZ,\bA\mid\bbY_n)\,d\bZ$.
By Proposition~\ref{prop:DCT}, $F$ is differentiable and
$\partial_j F(\bA)=\int q(\bZ)\,\partial_j\log p(\bZ,\bA\mid\bbY_n)\,d\bZ$.
Fix indices $j,k$ and let $e_k$ be the $k$th Euclidean basis vector in the coordinates $\vecr(\bA)$.
For $t\neq 0$ such that $\bA_t:=\bA+t\,\vecr^{-1}(e_k)\in\mathcal N$, define
\[
\Delta_t(\bZ)
:=
\frac{\partial_j\log p(\bZ,\bA_t\mid\bbY_n)-\partial_j\log p(\bZ,\bA\mid\bbY_n)}{t}.
\]
By the mean value theorem, for each $\bZ$ there exists $\xi=\xi(\bZ,t)\in(0,1)$ such that
\[
\Delta_t(\bZ)=\partial_{jk}^2\log p(\bZ,\bA+\xi(\bZ,t)(\bA_t-\bA)\mid\bbY_n).
\]
Therefore $|\Delta_t(\bZ)|\le \|\nabla^2_{\vecr(\bA)}\log p(\bZ,\cdot\mid\bbY_n)\|_{op}\le h_2(\bZ)$.
Moreover, by the $C^2$-smoothness of \eqref{eq:12} in $\bA$,
$\Delta_t(\bZ)\to \partial_{jk}^2\log p(\bZ,\bA\mid\bbY_n)$ pointwise as $t\to 0$.
Since $h_2$ is $q$--integrable, dominated convergence yields
\[
\partial_{jk}^2 F(\bA)
=
\int q(\bZ)\,\partial_{jk}^2\log p(\bZ,\bA\mid\bbY_n)\,d\bZ.
\]
Doing this for all $j,k$ gives the matrix identity
\[
\nabla^2_{\vecr(\bA)}F(\bA)
=
\int q(\bZ)\,\nabla^2_{\vecr(\bA)}\log p(\bZ,\bA\mid\bbY_n)\,d\bZ,
\]
which completes the proof.
\end{proof}

We prove Lemma~\ref{lem:local-concavity}.

\begin{proof}[Proof of Lemma~\ref{lem:local-concavity}]
By Proposition~\ref{prop:DCT2}, $\ell^{\dagger}_n(\bA)$ is twice differentiable on $\mathcal N$ and we may interchange
$\nabla^2_{\vecr(\bA)}$ with the $q_n^\dagger$-integral. Hence, for all $\bA\in\mathcal N$,
\begin{align*}
\nabla^2_{\vecr(\bA)}\ell^{\dagger}_n(\bA)
&=
\nabla^2_{\vecr(\bA)} \int{q_n^\dagger(\bZ)} \log p(\bbY_n,\bZ \mid\bA) d\bZ\\
&=
\E_{q_n^\dagger(\bZ)}\!\big[\nabla^2_{\vecr(\bA)}\log p(\bbY_n,\bZ \mid\bA)\big] := -\mathbf J_n(\bA),
\end{align*}
where we used the definitions in Assumption~\ref{assum6}(b).
Therefore by Assumption~\ref{assum6}(b),
\[
\inf_{\bA\in\mathcal N}\lambda_{\min}\!\Big( -\frac{1}{n}\nabla^2_{\vecr(\bA)}\ell^{\dagger}_n(\bA) \Big) \ge \delta
\]
Next, since $\tf_n^\dagger(\bA)=\ell^{\dagger}_n(\bA)+\log\pi(\bA)$,
\[
-\frac{1}{n}\nabla^2_{\vecr(\bA)}\tf_n^\dagger(\bA)
=
-\frac{1}{n}\nabla^2_{\vecr(\bA)}\ell^{\dagger}_n(\bA)
-\frac{1}{n}\nabla^2_{\vecr(\bA)}\log\pi(\bA).
\]
Under the matrix-normal prior $\bA\sim\mathcal{MN}_{r-u,u}(\bA_0,\bU_0^{(\bA)},\bV_0^{(\bA)})$,
\[
\nabla^2_{\vecr(\bA)}\log\pi(\bA)=-(\bV_0^{(\bA)})^{-1}\otimes(\bU_0^{(\bA)})^{-1}\ \preceq\ 0,
\]
so $-\nabla^2\log\pi(\bA)\succeq 0$ and therefore
\[
-\frac{1}{n}\nabla^2_{\vecr(\bA)}\tf_n^\dagger(\bA)
\succeq
-\frac{1}{n}\nabla^2_{\vecr(\bA)}\ell^{\dagger}_n(\bA).
\]
Hence the same curvature lower bound $c_0$ holds for $\tf_n^\dagger$.

The uniform lower bounds on $-\nabla^2 \ell^{\dagger}_n$ and $-\nabla^2 \tf_n^\dagger$ imply that $\ell^{\dagger}_n$ and $\tf_n^\dagger$
are uniformly strongly concave on the convex set $\mathcal N$, hence each admits a unique maximizer on $\mathcal N$.

Finally, since $\mathcal N=\mathcal B(\bA^\star,\rho)$, if $\rho-\|\bbA_n-\bA^\star\|_F>0$ then any
$\bar\delta_y\in(0,\rho-\|\bbA_n-\bA^\star\|_F]$ satisfies $\mathcal B(\bbA_n,\bar\delta_y)\subseteq \mathcal N$.
The same argument applies to $\hbA_n$ and $\hat\delta_y$.
\end{proof}

\subsection{Proof of Lemma~\ref{lem:third-derivative}}
\label{sec:E.2}

\begin{proof}[Proof of Lemma~\ref{lem:third-derivative}]
Fix the current iterate $q_n^\dagger(\bZ)$ and consider the induced $\bA$-objective
$\tf_n^\dagger(\bA)=\ell^{\dagger}_n(\bA)+\log\pi(\bA)$ for the fixed-$q_n^\dagger$ $\bA$-update.
Up to an additive constant in $\bA$, the explicit expression in~\eqref{eq:tAobj} can be written as
\[
\tf_n^\dagger(\bA)
=
\kappa\,g(\bA) + Q_n(\bA),
\qquad
g(\bA):=\log\big|\bI_{r-u}+\bA\bA^\top\big|,
\qquad
\kappa:=\frac{2n+\nu^{(1)}+\nu^{(0)}}{2},
\]
where $Q_n(\bA)$ collects all remaining $\bA$-dependent terms.
Under the fixed-$q_n^\dagger$ update, all terms in $Q_n(\bA)$ are at most quadratic in the entries of $\bA$
(trace-quadratic and trace-linear forms), hence
\[
\nabla^3_{\vecr(\bA)}Q_n(\bA)\equiv 0.
\]
Moreover, under the matrix-normal (Gaussian) prior on $\vecr(\bA)$,
$\log\pi(\bA)$ is a quadratic form in $\vecr(\bA)$ and therefore
\[
\nabla^3_{\vecr(\bA)}\log\pi(\bA)\equiv 0.
\]
Consequently,
\[
\nabla^3_{\vecr(\bA)}\tf_n^\dagger(\bA)
=
\kappa \,\nabla^3_{\vecr(\bA)} g(\bA),
\qquad
\nabla^3_{\vecr(\bA)}\ell^{\dagger}_n(\bA)
=
\nabla^3_{\vecr(\bA)}\tf_n^\dagger(\bA),
\]
for all $\bA$.

Next, the map $\bA\mapsto g(\bA)=\log|\bI_{r-u}+\bA\bA^\top|$ is $C^\infty$ on $\mathbb R^{(r-u)\times u}$ because $\bI_{r-u}+\bA\bA^\top$ is positive definite for every $\bA$. Hence $\bA\mapsto \|\nabla^3_{\vecr(\bA)}g(\bA)\|^{*}$ is continuous.
Since $\mathcal N$ is compact, the extreme value theorem implies that
\[
C_{3,\mathcal N}
:=
\sup_{\bA\in\mathcal N}\Big\|\nabla^3_{\vecr(\bA)}g(\bA)\Big\|^{*}
\ <\ \infty.
\]
Therefore, for all $\bA\in\mathcal N$,
\[
\Big\|\tfrac{1}{n}\nabla^3_{\vecr(\bA)}\tf_n^\dagger(\bA)\Big\|^{*}
=
\frac{\kappa}{n}\,\Big\|\nabla^3_{\vecr(\bA)}g(\bA)\Big\|^{*}
\le
\frac{\kappa}{n}\,C_{3,\mathcal N}.
\]
Since $(\nu^{(1)}+\nu^{(0)})$ is fixed,
\[
\frac{\kappa}{n}=1+\frac{\nu^{(1)}+\nu^{(0)}}{2n}\le 2
\quad\text{for all }n\ge \nu^{(1)}+\nu^{(0)}.
\]
Thus, for all sufficiently large $n$,
\[
\sup_{\bA\in\mathcal N}\Big\|\tfrac{1}{n}\nabla^3_{\vecr(\bA)}\tf_n^\dagger(\bA)\Big\|^{*}
\le 2C_{3,\mathcal N}.
\]
Setting $M_1:=2C_{3,\mathcal N}$ proves the claim.

Finally, by Lemma~\ref{lem:local-concavity} the closed balls
$\mathcal B(\bbA_n,\bar\delta_y)$ and $\mathcal B(\hbA_n,\hat\delta_y)$ are subsets of $\mathcal N$
with probability tending to one, hence the same uniform bound holds on these balls.
\end{proof}

\subsection{Proof of Lemma~\ref{lem:bound-M2}}
\label{sec:E.3}
\begin{proof}[Proof of Lemma~\ref{lem:bound-M2}]
Under the Gaussian (matrix-normal) prior $\vecr(\bA)\sim\mathcal N(\vecr(\bA_0),\bSigma_0^{(\bA)})$,
\[
\frac{1}{\pi(\bA)}
=(2\pi)^{d_{\bA}/2}\,|\bSigma_0^{(\bA)}|^{1/2}
\exp\!\left\{\frac12(\vecr(\bA)-\vecr(\bA_0))^\top \bSigma_0^{(\bA)-1}(\vecr(\bA)-\vecr(\bA_0))\right\}.
\]
Using $x^\top \Sigma^{-1}x \le \|x\|_2^2/\lambda_{\min}(\Sigma)$, we obtain
\[
\frac{1}{\pi(\bA)}
\le
(2\pi)^{d_{\bA}/2}\,|\bSigma_0^{(\bA)}|^{1/2}
\exp\!\left\{\frac{\|\vecr(\bA)-\vecr(\bA_0)\|_2^2}{2\,\lambda_{\min}(\bSigma_0^{(\bA)})}\right\}.
\]
For any $\bA$ with $\|\bA-\bbA_n\|_F\le \bar\delta$,
\[
\|\vecr(\bA)-\vecr(\bA_0)\|_2
\le
\|\vecr(\bA)-\vecr(\bbA_n)\|_2 + \|\vecr(\bbA_n)-\vecr(\bA_0)\|_2
\le
\bar\delta + \|\vecr(\bbA_n)-\vecr(\bA_0)\|_2.
\]
Since $\bbA_n\in\mathcal N=\mathcal B(\bA^\star,\rho)$ by definition,
\[
\|\vecr(\bbA_n)-\vecr(\bA_0)\|_2
\le
\|\vecr(\bbA_n)-\vecr(\bA^\star)\|_2 + \|\vecr(\bA^\star)-\vecr(\bA_0)\|_2
\le
\rho + \|\vecr(\bA^\star)-\vecr(\bA_0)\|_2.
\]
Combining the last two displays yields
\[
\sup_{\|\bA-\bbA_n\|_F\le \bar\delta}\|\vecr(\bA)-\vecr(\bA_0)\|_2
\le
\bar\delta+\rho+\|\vecr(\bA^\star)-\vecr(\bA_0)\|_2.
\]
Substituting this bound into the prior inequality completes the proof and gives the stated constant $M_2$.
\end{proof}

\subsection{Proof of Lemma~\ref{lem:variational-posterior}}
\label{sec:E.4}

\begin{proof}[Proof of Lemma~\ref{lem:variational-posterior}]
Fix $\bA\in\mathcal N$ such that $\|\bA-\bbA_n\|_F\le \bar\delta_y$.
By Taylor's theorem applied to $\ell^{\dagger}_n$ along the segment from $\bbA_n$ to $\bA$,
there exists $\btA\in[\bbA_n,\bA]\subset\mathcal N$ such that
\[
\ell^{\dagger}_n(\bA)-\ell^{\dagger}_n(\bbA_n)
=
\frac{1}{2}\langle \bA-\bbA_n,\ \nabla^2_{\vecr(\bA)}\ell^{\dagger}_n(\btA)\,(\bA-\bbA_n)\rangle
+\frac{1}{6}\,\nabla^3_{\vecr(\bA)}\ell^{\dagger}_n(\btA)\big[\vecr(\bA-\bbA_n)\big]^{\otimes 3}.
\]
Equivalently,
\[
\frac{\ell^{\dagger}_n(\bA)-\ell^{\dagger}_n(\bbA_n)}{n}
\le
-\frac{1}{2}\,\lambda_{\min}\!\Big(-\tfrac{1}{n}\nabla^2_{\vecr(\bA)}\ell^{\dagger}_n(\btA)\Big)\,\|\bA-\bbA_n\|_F^2
+\frac{1}{6}\,\Big\|\tfrac{1}{n}\nabla^3_{\vecr(\bA)}\ell^{\dagger}_n(\btA)\Big\|^{*}\,\|\bA-\bbA_n\|_F^3.
\]
By Lemma~\ref{lem:local-concavity}, with probability tending to one and for all sufficiently large $n$,
\[
-\tfrac{1}{n}\nabla^2_{\vecr(\bA)}\ell^{\dagger}_n(\bA)\ \succeq\ \delta\,\mathbf I_{d_{\bA}}
\qquad\text{for all }\bA\in\mathcal N,
\]
and by Lemma~\ref{lem:third-derivative}, for all sufficiently large $n$,
\[
\sup_{\bA\in\mathcal N}\Big\|\tfrac{1}{n}\nabla^3_{\vecr(\bA)}\ell^{\dagger}_n(\bA)\Big\|^{*}\ \le\ M_1.
\]
Let $t:=\|\bA-\bbA_n\|_F\in[0,\bar\delta_y]$. Then on the above high-probability event,
\[
\frac{\ell^{\dagger}_n(\bA)-\ell^{\dagger}_n(\bbA_n)}{n}
\ \le\ f(t)
\ :=\ -\frac{\delta}{2}\,t^2+\frac{M_1}{6}\,t^3.
\]
Since $f'(t)=-\delta+\frac{M_1}{2}t^2=t\big(\frac{M_1}{2}t-\delta\big)$,
$f$ is decreasing on $(0,2\delta/M_1]$ and increasing on $[2\delta/M_1,\infty)$.
Therefore, for any shell $t\in[\rho_n,\bar\delta_y]$,
\[
\sup_{\ \rho_n\le \|\bA-\bbA_n\|_F\le \bar\delta_y}\ \frac{\ell^{\dagger}_n(\bA)-\ell^{\dagger}_n(\bbA_n)}{n}
\ \le\ \max\{f(\rho_n),\,f(\bar\delta_y)\}
\ =\ -\,\min\{-f(\rho_n),\,-f(\bar\delta_y)\}.
\]
Define
\[
\tilde\delta_n
:=
\min \!\left\{ \frac{\delta}{2}\,\rho_n^2-\frac{M_1}{6}\,\rho_n^3,\ \ \frac{\delta}{2}\,\bar\delta_y^2-\frac{M_1}{6}\,\bar\delta_y^3 \right\},
\]
to obtain
\[
\sup_{\ \rho_n\le \|\bA-\bbA_n\|_F\le \bar\delta_y}\ 
\frac{\ell^{\dagger}_n(\bA)-\ell^{\dagger}_n(\bbA_n)}{n}
\ \le\ -\,\tilde{\delta}_n.
\]

It remains to show $\rho_n>0$ with probability tending to one and $\|\hbA_n-\bbA_n\|_F=O(n^{-1})$.
Because $\bbA_n$ and $\hbA_n$ are interior maximizers on the balls from Lemma~\ref{lem:local-concavity}$,$
their first-order conditions give
\[
\nabla_{\vecr(\bA)}\ell^{\dagger}_n(\bbA_n)=0,
\qquad
\nabla_{\vecr(\bA)}\ell^{\dagger}_n(\hbA_n)+\nabla_{\vecr(\bA)}\log\pi(\hbA_n)=0.
\]
By the mean-value form of Taylor's theorem applied to $\nabla\ell_n$,
there exists $\btA_n$ on the segment between $\bbA_n$ and $\hbA_n$ such that
\[
-\nabla^2_{\vecr(\bA)}\ell^{\dagger}_n(\btA_n)\,\big(\vecr(\hbA_n)-\vecr(\bbA_n)\big)
=\nabla_{\vecr(\bA)}\log\pi(\hbA_n).
\]
By Lemma~\ref{lem:local-concavity}, $-\nabla^2\ell^{\dagger}_n(\btA_n)\succeq n \delta\,\mathbf I$, hence
\[
\|\hbA_n-\bbA_n\|_F
=
\|\vecr(\hbA_n)-\vecr(\bbA_n)\|_2
\ \le\ \frac{1}{n \delta}\ \big\|\nabla_{\vecr(\bA)}\log\pi(\hbA_n)\big\|_2.
\]
Under the matrix-normal prior $\vecr(\bA)\sim\mathcal N(\vecr(\bA_0),\bSigma_0^{(\bA)})$,
\[
\nabla_{\vecr(\bA)}\log\pi(\bA)= -\,\bSigma_0^{(\bA)-1}\big(\vecr(\bA)-\vecr(\bA_0)\big).
\]
Since $\hbA_n\in\mathcal N$ and $\mathcal N$ is compact, there exists a finite constant $C_\pi<\infty$
(independent of $n$) such that $\sup_{\bA\in\mathcal N}\|\nabla\log\pi(\bA)\|_2\le C_\pi$.
Therefore,
\[
\|\hbA_n-\bbA_n\|_F \ \le\ \frac{C_\pi}{n \delta}\ =\ O(n^{-1}).
\]
Since $\hat\delta_y>0$ is fixed and $\rho_n=\hat\delta_y-\|\bbA_n-\hbA_n\|_F$,
it follows that $\rho_n>0$ with probability tending to one and $\rho_n\to\hat\delta_y$.
This completes the proof.
\end{proof}

\section{Additional simulation studies results}
\label{sec:F}
Figures~\ref{beta all1} and~\ref{beta all2} present the marginal posterior distributions for all elements of $\bbeta$ under CALVI, MH--Gibbs, and Manifold--Gibbs when the true envelope dimension is $u^{*}=2$. Consistent with Figure~\ref{fig:beta marginal}, the three methods yield broadly similar marginal distributions, with only minor differences in posterior spread across coefficients.
\begin{figure}
    \centering
    \includegraphics[width=1\linewidth]{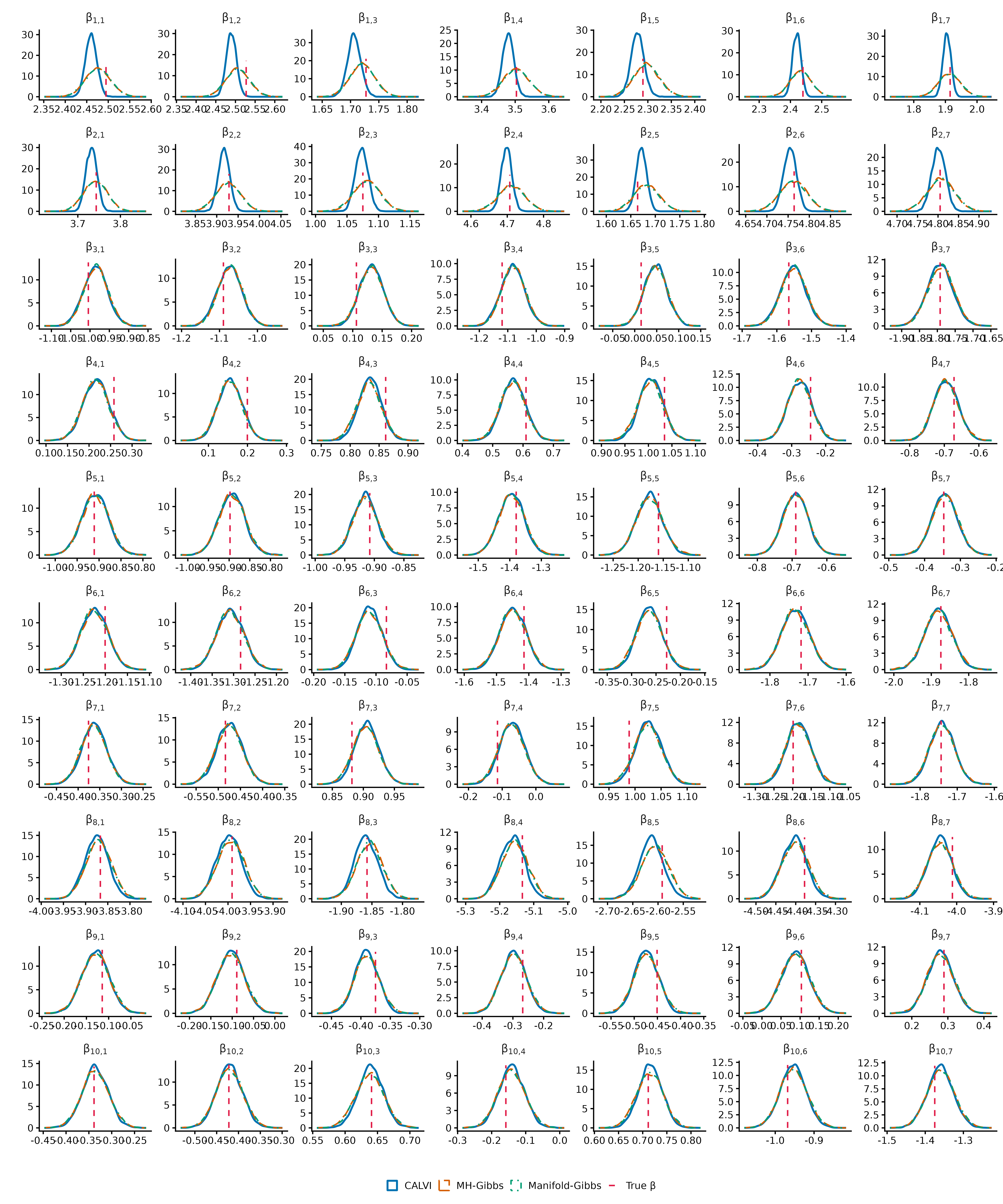}
    \caption{Marginal variational and posterior distributions for $\beta_{ij}, \ i=1,\dots,10 \ , j=1,\dots,7$ based on the first replicate of the simulated data. Red dotted lines denote the true values of the parameters.}
    \label{beta all1}
\end{figure}

\begin{figure}
    \centering
    \includegraphics[width=1\linewidth]{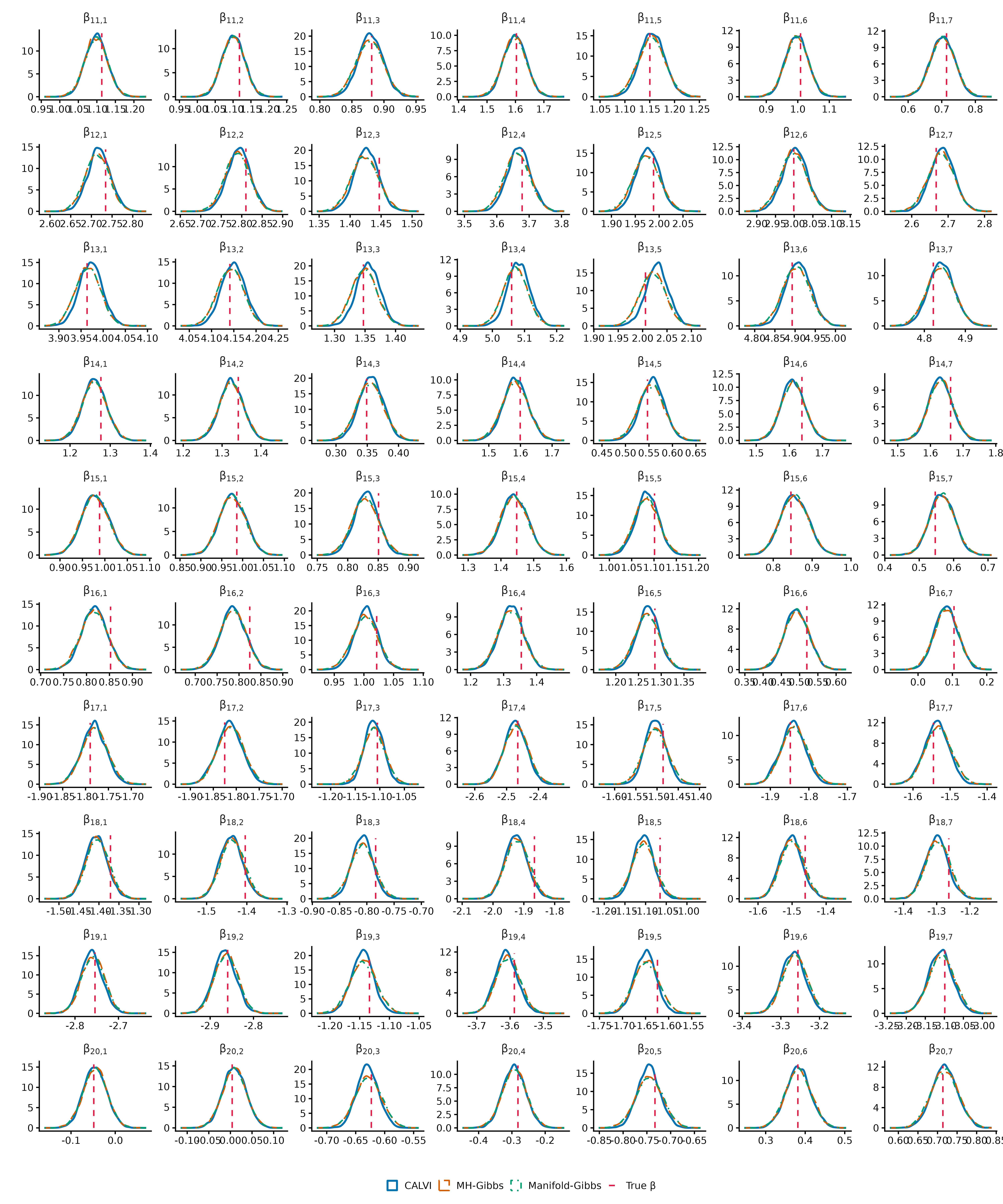}
    \caption{Marginal variational and posterior distributions for $\beta_{ij}, \ i=11,\dots,20 \ , j=1,\dots,7$ based on the first replicate of the simulated data. Red dotted lines denote the true values of the parameters.}
    \label{beta all2}
\end{figure}

\section*{Acknowledgements}
This research was supported by the BK21 FOUR (Fostering Outstanding Universities for Research, No.\ 5120200913674) funded by the Ministry of Education (MOE, Korea) and the National Research Foundation of Korea (NRF).

\bibliographystyle{apalike}
\bibliography{envelope}

\end{document}